\documentclass{article}[11pt]
\usepackage[utf8]{inputenc}
\usepackage[margin=1in]{geometry}

\usepackage{setspace}
\usepackage[usenames,dvipsnames,svgnames,table]{xcolor}
\usepackage[colorlinks=true,linkcolor=blue,citecolor=ForestGreen,filecolor=magenta,urlcolor=cyan]{hyperref}     
\usepackage{bbm}
\usepackage{amsmath,amsthm,amssymb}
\usepackage{thmtools}
\usepackage{thm-restate}
\usepackage{makecell}
\definecolor{Greenish}{rgb}{0.8,1,.7}
\DeclareRobustCommand{\hlgreenish}[1]{{\sethlcolor{Greenish}\hl{#1}}}
\usepackage{algorithm}
 \usepackage{xspace} 
 \usepackage{enumerate}
 
\allowdisplaybreaks
\usepackage{multirow}
\usepackage{dsfont}
\usepackage{url}  
\usepackage{booktabs}       
\usepackage{amsfonts}       
\usepackage{nicefrac}       
\usepackage{microtype}      
\usepackage{csquotes}
\usepackage{aligned-overset}
 \usepackage{orcidlink}
\usepackage{wrapfig}
\usepackage{diagbox}
\usepackage[toc,page,header]{appendix}
\usepackage{minitoc}
\usepackage{bm}
\usepackage{caption}
\usepackage{subcaption}

\usepackage[noend]{algpseudocode}
\usepackage{algorithm}

\usepackage{mathtools}
\usepackage{tabularx}
\usepackage{amssymb}

\usepackage{soul}
\usepackage[capitalize]{cleveref}

\newcommand{\Pro}[1]{\mathbf{P} \left[\,#1\,\right]}
\renewcommand{\Pr}[2]{\mathbf{P}_{#1} \left[\,#2\,\right]}

\newcommand{\E}[1]{\mathbf{E}[\,#1\,]}
\newcommand{\Var}[1]{\mathbf{Var}\left[\,#1\,\right]}
\newcommand{\VAR}[2]{\mathbf{Var}_{#1}\left[\,#2\,\right]}
\newcommand{\Ex}[1]{\mathbf{E} \left[\,#1\,\right]}
\newcommand{\EX}[2]{\mathbf{E}_{#1} \left[\,#2\,\right]}
\newcommand{\Cov}[3]{\mathbf{Cov}_{#1}\left[\,#2,#3\,\right]}

\newcommand{\Condtwo}[2]{#1~\Big|~#2}
\newcommand{\ind}{\mathds{1}}
\newcommand{\vol}[1]{\operatorname{vol}(#1)}
\newcommand{\push}{{\rm{\texttt{PUSH}}}\xspace}
\newcommand{\pull}{{\rm{\texttt{PULL}}}\xspace}
\newcommand{\pullindex}{\mbox{\rm{\tiny{\texttt{PULL}}}}\xspace}
\newcommand{\pushindex}{\mbox{\rm{\tiny{\texttt{PUSH}}}}\xspace}

\newcommand{\pp}{{\rm{\texttt{PUSH-PULL}}}\xspace}

\newcommand{\Cgrow}{C_{\mathrm{grow}}\xspace}
\newcommand{\Cshrink}{C_{\mathrm{shrink}}\xspace}
\newcommand{\Psigrow}{\Psi\xspace}
\newcommand{\Phigrow}{\Phi\xspace}

\renewcommand{\tilde}{\widetilde}
\renewcommand{\hat}{\widehat}
\renewcommand{\epsilon}{\varepsilon}
\newcommand{\eps}{\varepsilon}

\newtheorem{theorem}{Theorem}[section]  

\newtheorem{lemma}[theorem]{Lemma}

\newtheorem{corollary}[theorem]{Corollary}
\newtheorem{remark}[theorem]{Remark}

\newtheorem{claim}[theorem]{Claim}
\newtheorem{definition}[theorem]{Definition}

\numberwithin{equation}{section}

\newcommand{\NC}{{\hyperlink{p3}{\ensuremath{\mathcal{P}_1}}}\xspace}
\newcommand{\Mono}{{\hyperlink{p1}{\ensuremath{\mathcal{P}_2}}}\xspace}
\newcommand{\MonoT}{{\hyperlink{pt1}{\ensuremath{\tilde{\mathcal{P}}_2}}}\xspace}
\newcommand{\BEG}{{\hyperlink{p2}{\ensuremath{\mathcal{P}_3}}}\xspace}
\newcommand{\BEGT}{{\hyperlink{pt2}{\ensuremath{\tilde{\mathcal{P}}_3}}}\xspace}
\newcommand{\PFour}{{\hyperlink{p4}{\ensuremath{\mathcal{P}_4}}}\xspace}

\title{Rumors with Changing Credibility}

\author{Charlotte Out\thanks{Department of Computer Science \& Technology, University of Cambridge, UK, \texttt{ceo33@cam.ac.uk}, \orcidlink{0000-0003-1316-6336}} \and Nicol\'{a}s Rivera\thanks{Universidad de Valpara\'iso, Vapara\'iso, Chile, \texttt{nicolas.rivera@uv.cl}, \orcidlink{0000-0003-3368-9708}} 	\and
Thomas Sauerwald\thanks{Department of Computer Science \& Technology, University of Cambridge, UK, \texttt{thomas.sauerwald@cl.cam.ac.uk}, \orcidlink{0000-0002-0882-283X}} \and    John Sylvester\thanks{Department of Computer Science, University of Liverpool, UK, \texttt{john.sylvester@liverpool.ac.uk}, \orcidlink{0000-0002-6543-2934}}
}

\date{}
 \begin{document}

\maketitle
 
\begin{abstract} 

Randomized rumor spreading processes diffuse information on an undirected graph and have been widely studied. In this work, we present a generic framework for analyzing a broad class of such processes on regular graphs. Our analysis is protocol-agnostic, as it only requires the expected proportion of newly informed vertices in each round to be bounded, and a natural negative correlation property.

This framework allows us to analyze various protocols, including \texttt{PUSH}, \texttt{PULL}, and \texttt{PUSH-PULL}, thereby extending prior research. Unlike previous work, our framework accommodates message failures at any time $t\geq 0$ with a probability of $1-q(t)$, where the \textit{credibility} $q(t)$  is any function of time. This enables us to model real-world scenarios in which the transmissibility of rumors may fluctuate, as seen in the spread of ``fake news'' and viruses.  Additionally, our framework is sufficiently broad to cover dynamic graphs. 
 
\end{abstract}

\section{Introduction}

The rise of online social networks has facilitated a way for network users to rapidly obtain information, express their opinion, and stay in touch with friends and family. However, at the same time the large scale information cascades enabled by these new social technologies provide fertile ground for the spread of misinformation, rumors and hoaxes. This in turn can have severe consequences such as public panic, growing polarization, the manipulation of political events, and also economic damage. For instance, in 2013 a rumor that President Obama was injured in two explosions at the White House led to $\$90$ billion USD being temporarily wiped off the value of United States stock market~\cite{peter2013bogus}. In the same year the World Economic Forum report \cite{howell2013digital} listed ``massive digital misinformation'' as one of the main risks for the modern society. More recently we have seen the spread of misinformation surrounding the Covid-19 pandemic \cite{Burns}. Consequently, there has been a growing body of work aiming to gain insights into the rumor spreading dynamics \cite{Fakenews3,Fakenews1,Fakenews2,ZannettouSBK19}. 

For a long time, randomized rumor spreading protocols such as the \push, \pull and \pp protocols have been used to model the dissemination of information on graphs, e.g., \cite{boyd2006randomized,demers1987epidemic,karp2000randomized}. Both by mathematical analysis on ``scale free'' graphs in addition to experimental results on real-world social networks, it has been demonstrated that these protocols (in particular, \pp) spread a rumor to a large fraction of vertices in a very short time (e.g.,~\cite{doerr2012rumors}).

However, one shortcoming of the previous works that analyze these protocols is the assumption that the probability with which an individual believes the rumor, when receiving it, is constant over time -- in fact, in many studies it is assumed that this \emph{credibility} is equal to one in all rounds. In real world settings, one can imagine that the occurrence of emergent events (such as an earthquake or a new possibly lethal decease) can intensify the formation and propagation of rumors due to their suddenness and urgency, followed by a decrease in credibility once more information has become available. A related example is the spread of viruses, where counter-measures such as vaccination or social distancing, but also seasonal effects may affect the transmissibility over time, potentially even periodically/non-monotonically. 

Moreover, it is often assumed that the graph is fixed throughout the execution of randomized rumor spreading protocols, which is rather restrictive since many networks, e.g., social networks, P2P networks or communication networks, are subject to frequent changes.

To address these issues, we introduce a new methodology for analyzing randomized rumor spreading protocols that allows us to study \push, \pull, and \pp processes under the presence of a time-changing credibility (or transmissibility) function $q(t)$ and dynamic graphs $(G_t)_{t\geq 0}$. However, our method is more general and allows us to study a broader class of spreading processes on dynamic graphs. To show the effectiveness of our analysis, we recover known results for the \push, \pull, and \pp protocols in the context of a constant credibility function $q$, and provide analysis for specific time-dependent credibility functions $q(t)$.

\subsection{Our Contribution}

In this work, we present a general framework for analyzing a large class of randomized rumor spreading models. Our main results give concentration for the number of vertices informed after a certain stopping time. These results are very general however we show in detail how they can be applied to several models.  
\begin{itemize}
    \item \textbf{Broad Class of Spreading Processes.} {Instead of using protocol specific characteristics, our framework only requires some mild conditions on the spreading process (i.e., bounded expected growth and a natural negative correlation property; see~\cref{def:processes})}.  
    This allows our setting to cover many models of randomized rumor spreading, beyond the standard \pull and \push models (see \cref{lem:protocol_growth}, the final bullet point below, and \cref{sec:examples}).
    \item \textbf{Credibility Function $q(t)$.} Our model allows for a time-dependent credibility function $q(t)\in[0,1]$, which specifies how transmissible the rumor is in each step. This can be seen as a major generalization of the prevalent notion of ``robustness'' in the literature, which usually refers to the uniform fault model with $q$ fixed over $t$. Unlike in previous models, our credibility functions can be arbitrary, in particular they do not need to be monotone.
    \item \textbf{Stopping time Criterion.} 
    We introduce a new technical tool based on a stopping time criterion. Roughly, for some desired number of vertices $B$ to be informed, the stopping time triggers when a sum of expected growth factors of the process exceeds a threshold depending on $B$. The aforementioned growth factors are conditional expectations of the proportion of new vertices informed in the next step. We show that if this stopping criterion is met, then $B$ vertices are informed with high probability (see \cref{thm:AzumaGrowing}). This is complemented by \cref{thm:shrinkingphase} with a dual statement on the shrinking of the uninformed vertices.
    Both results are significantly more general than previous analyses, which usually rely on a growth factor ``target'' that is independent of $t$ and the set of informed vertices.
    \item \textbf{Dynamic Graphs.} Due to the general nature of our framework and stopping criteria, our analysis ``abstracts away'' the graph and the specific spreading process. Hence, we can cover sequences of dynamic regular graphs $(G_t)_{t \geq 0}$ instead of a fixed graph $G$. This flexibility comes from the fact that the connectivity of each $G_t$ is captured by the growth factor of the process at round $t$, which in turn determines the stopping criterion. In particular, we do not require the graph to be connected at each step (see \cref{rem:con}).  
    \item \textbf{Applications.} We prove several new results for general and specific credibility functions.
    First, for general credibility functions, we combine our stopping time criterion with a simple lower bound based on sub-martingales. Together, they reveal a threshold phenomenon, very roughly saying that for expander graphs the quantity $\sum_{k=0}^{t} \log(1+q(k))$  approximates $\log(|I_t|)$, where $I_t$ is the set of vertices informed by time $t$ (see \cref{sec:applicationArbitrary}). 
    
    After that, we turn to some specific credibility functions, including additive, multiplicative and Power-Law
    (see \cref{sec:additive_models,sec:applicationMultiplicative,sec:applicationPowerLaw} for the respective definitions and results). There, we prove several dichotomies in terms of the decay of $q(t)$.
    
    Despite the generality and abstract nature of our main results, 
    we also recover some previous results for \emph{static} graphs (and time-invariant $q(t)$) as a special case; however, our results for \push, \pull and \pp additionally apply to \emph{dynamic} graphs (see, e.g., the results in \cref{sec:fixed_credibility}). 
     
\end{itemize}

\subsection{Related Work}

\paragraph{Classical Protocols and Robustness.}Given a rumor spreading process on an $n$-vertex graph, define the spreading time by $T(n)$ as the first time all vertices are informed. The spreading time of \push was first investigated on complete graphs by Frieze and Grimmett~\cite{frieze1985shortest}. Pittel~\cite{pittel1987spreading} improved on this, showing that for \push on the complete graph, the spreading time is given by $T(n) = \log_2(n)+\log(n) \pm f(n)$ with probability (w.p.) $1-o(1)$, for any $f(n)=\omega(1)$. Karp, Schindelhauer, Schenker and V\"ocking~\cite{karp2000randomized} investigated the \pp model (and variants) with a focus on the total number of messages sent. In particular, they exploit the phenomenon that once a constant fraction of vertices are informed, \pull manages to inform all vertices in just $O(\log \log n)$ rounds.

Doerr and Kostrygin~\cite{doerr2017randomized} derived a bound on the expected spreading time $\Ex{T(n)}$ of \push, replicating the bound from \cite{pittel1987spreading} but only with an additive $O(1)$ error instead of $f(n)$.  
Furthermore, \cite{doerr2017randomized} also considered \pull and \pp on complete graphs, and determined these spreading times up to and additive $O(1)$ error. They also presented a more general result for the uniform fault model, where the leading factors are delicate functions of the (time-invariant) credibility $q \in (0,1]$. We are able to recover a with high probability version of the upper bounds from \cite{doerr2017randomized} for \push, \pull and \pp (see \cref{sec:fixed_credibility}).
 
Fountoulakis, Huber and Pangiotou~\cite{fountoulakis2010reliable} considered the uniform fault setting of \push on random graphs with $n$ vertices where each edge is present w.p. $p= \omega(\log n/n)$. They proved that, up to lower-order terms, the same bound as for the complete graph holds. For the model without faults, Fountoulakis and Panagiotou~\cite{FP13} presented a tight analysis for \push on random $d$-regular graph for any constant $d \geq 3$. Panagiotou, Perez-Gimenez, Sauerwald and Sun \cite{PPSS15} analyzed \push on almost-regular strong expanders, recovering the runtime bound for complete graphs up to low order terms (see \cref{eq:strongexpander} for the definition of strong expander for regular graphs). 

Finally, Daknama, Panagiotou and Reisser~\cite{daknama_panagiotou_reisser_2021} greatly extended and unified these lines of works in terms of the graph classes considered, and the uniform fault model. Among other results, they proved that the aforementioned results from \cite{doerr2017randomized} (for \push, \pull and \pp) also hold for almost-regular strong expanders, without any change in the leading factor. Our framework allows us to recover the upper bounds in \cite{daknama_panagiotou_reisser_2021} for regular graphs as well as dynamic sequences of regular graphs (see \cref{sec:fixed_credibility}).

For general graphs (including highly non-regular ones), Chierichetti, Giakkoupis, Lattanzi and Panconesi~\cite{GiakkoupisConductance} proved an upper bound of $O(\log n / \varphi)$ on the time to inform all vertices for \pp, where $\varphi$ is the conductance of the graph. A similar, but more complicated bound was shown by Giakkoupis \cite{G14} for the \pp model, where the conductance is replaced by the vertex expansion. The results of both works also extend to \push and \pull, if the graph is (approximately) regular.

\paragraph{Dynamic Graphs.} 
Extending the aforementioned bounds for conductance and vertex expansion, Giakkoupis, Sauerwald and Stauffer \cite{giakkoupis2014randomized} proved similar bounds for dynamic graphs in the \pp model, where each graph $G_{t\geq 0}=(V,E_{t\geq 0})$ must be $d_t$-regular. In particular, they proved that if the sum of the conductances over rounds $0,1,\ldots,T$ is $\Omega(\log n)$, then by round $T$ all vertices are informed. Pourmiri and Mans \cite{pourmiri2020tight} analyzed an asynchronous version of \pp. While some of their positive results are similar to the ones in \cite{giakkoupis2014randomized}, they also established dichomoties between the synchronous and asynchronous version on dynamic graphs. Our approach can be seen as a refinement and generalization of the methods employed in these two works, since our stopping time aggregates over the (random) conductances of the sets $I_t$, for $t=0,1,\ldots,T$, and it works for arbitrary, so-called $\Cgrow$-growing and $\Cshrink$-shrinking processes.

Finally, Clementi, Crescenzi, C.~Doerr, Fraigniaud, Pasquale and Silvestri~\cite{CDF16} 
analyzed \push on a random dynamic graph model called Edge Markovian Evolving Graph, and proved a runtime bound of $O(\log n)$ for certain parameter ranges of their model. Ideas and techniques related to rumor spreading have also been employed in the analysis of components in a temporal random graph model \cite{BeckerCCKRRZ23,CasteigtsRRZ21}.

\paragraph{Other Models with Time Dependent Credibility Functions.} The inclusion of a local time dependent forgetting rate in the SIR model \cite{kermack1927contribution} was empirically investigated by Zhao, Xie, Gao, Qiu, Wang, and Zhang \cite{zhao2013rumor}, leading to $q(t) := \mu - e^{\beta \cdot t}$, for $0\leq \mu - e^{\beta \cdot t}\leq 1 $, for $\mu$ and $\beta$ parameters indicating the initial credibility and the speed with which the credibility decreases. Very recently, Zehmakan, Out and Khelejan \cite{zehmakan2023rumors} studied a version of the Independent Cascade model \cite{kempe2003maximizing} where $q(t)$ is a variant of the multiplicative credibility function (with $\alpha = 1/2$, see \cref{def:multcred}), but additionally is edge dependent (i.e.\ a function $q(t,uv)$, $uv\in E(G)$) and depends on the Jaccard similarity between two vertices $u$ and $v$.

\section{Models and Notation}
\label{sec:models_notation}
We will cover some basic notation before introducing the models studied in this paper. 
\subsection{Notation}

\paragraph{Graph Notation.}
Throughout this paper, all considered graphs $G=(V,E)$ will be simple  and undirected. We denote $n:=|V|$ and $m:=|E|$.  For a node $v\in V$, $N(v):=\{w\in V: \{w,v\} \in E\}$ is the \emph{neighborhood} of $v$, and $\deg(v):= |N(v)|$ is called the \emph{degree} of $v$. We say a graph is \textit{regular} if every vertex has the same degree. For $U \subseteq V$ we let $N_U(v) := \{w\in U: \{v,w\}\in E\}=N(v) \cap U$, and denote $\deg_{U}(v) := |N_U(v)|$. We will also consider \emph{dynamic graphs}, which can be thought of as a sequence of graphs $(G_t)_{t\geq 0}$ where each graph $G_t=(V,E_t)$ is on the same vertex set, however the edge sets $E_t$ may change over time. 

For any two sets $U,W \subseteq V$, we let $e(U,W):= |\{\{u,w\}\in E: u \in U, w\in W\}|$ denote the number of edges between $U$ and $W$. The \emph{volume} of a set $U \subseteq V$ is the sum of the degrees of the vertices in $U$, $\vol{U}:=\sum_{u \in U} \deg(u)$. We let $A$ be the adjacency matrix of $G$ and denote the degree matrix by $D := \operatorname{diag}(\mathbf{d})$, where $\mathbf{d}(u) = \deg(u)$, which is the matrix with the degrees of the vertices on the diagonal and the rest of the entries equal to $0$. Lastly, we let $1=\lambda_1 \geq \lambda_2 \geq \dots \geq \lambda_n$ be the eigenvalues of the normalized adjacency matrix $D^{-1/2}AD^{-1/2}$ and let $\lambda := \max\{|\lambda_2|,|\lambda_3|,\dots ,|\lambda_n|\} \geq 0$.  

  We say that a regular graph $G$ of degree $d$ is a \textit{strong expander} if,
 \begin{equation}\label{eq:strongexpander}
     \lim_{n\to \infty} \lambda \to 0.
 \end{equation} 
 Note that a necessary requirement for that is $d \rightarrow \infty$. As noted in other works on rumor spreading, the class of random $d$-regular graphs with $d=\omega(1)$  forms an example of strong expander graphs with w.p. $1-o(1)$ \cite{BFSU98,sarid2022spectral}. We refer to \cite{daknama_panagiotou_reisser_2021,PPSS15} for the exact definition of strong expander graphs when $G$ is almost-regular.
 
 The \textit{conductance} \cite{jerrum1989approximating} of any vertex set $\emptyset \subsetneq S \subsetneq V$ in a graph $G=(V,E)$ is \begin{equation*}
    \varphi_G(S) := \frac{e(S, V\setminus S)}{\min\left(\vol{S}, \vol{V \setminus S}\right)}. 
\end{equation*}
If the graph $G$ or graph sequence $(G_t)_{t\geq 0}$ is clear from the context, we drop the subscript. The conductance of $G$ is in turn defined as,
\begin{equation*}
    \varphi(G):= \min_{\emptyset \subsetneq S \subsetneq V} \frac{e(S, V\setminus S)}{\min\left(\vol{S}, \vol{V\setminus S}\right)}.
\end{equation*} 
 
\paragraph{Model Notation.}
As mentioned, we will consider random processes on a sequence of $d_t$-regular graphs, $(G_t)_{t\geq 0}$ where each $G_t$ has a common vertex set $V$. We always assume that $d_t >0$ (i.e., we do not consider the empty graph). These processes produce a sequence of sets $(I_t)_{t\geq 0}$ where $I_t$ is the set of \textit{informed} vertices at time $t$ (i.e., after $t$ rounds are completed) and $I_{t} \subseteq I_{t+1}\subseteq V$ for all $t\geq 0$. Similarly, we let $U_t := V\setminus I_t$ denote the set of \emph{uninformed} vertices at time $t\geq 0$. Lastly, we define $\Delta_{t} := I_{t}\setminus I_{t-1}$ to be the set of vertices that get informed in round $t$. Further notation relating to such process is given in \cref{sec:ourproc}.
 
 \paragraph{Mathematical Notation and Assumptions.} We use  asymptotic notation $\mathcal{O}(\cdot), o(\cdot),\Omega(\cdot),\omega(\cdot),\Theta(\cdot),\dots$ throughout, this is always defined relative to the number of vertices $n$.  All logarithms are to base $e$, unless indicated otherwise. We let $n$ tend to infinity and say an event $\mathcal{E}$ happens \textit{with high probability} (w.h.p.) if it occurs w.p. $1-o(1)$. For $f: X\to \mathbb{R}$ a non-negative real-valued function with domain $X$, we let $\operatorname{Supp}(f):= \{x \in X : f(x) \neq 0\}$. We define $\mathfrak{F}^{t}$ to be the filtration corresponding to the first $t$ rounds of the process, in particular $\mathfrak{F}^{t}$ reveals $I_{0}, I_{1}, \dots, I_{t}$. For brevity, we set
\[ \Pr{t}{\cdot} := \Pro{\cdot \mid \mathfrak{F}^t} , \qquad \EX{t}{\cdot} := \Ex{\cdot\mid \mathfrak{F}^t}, \quad\text{and}\quad  \VAR{t}{\cdot } := \Ex{ \left( \cdot - \Ex{ \cdot \mid \mathfrak{F}^t} \right)^2 \mid \mathfrak{F}^t}. \]

 \subsection{Standard Rumor Spreading Protocols and Credibility Function \texorpdfstring{$q(t)$}{q(t)}}\label{sec:standard_rumor}

Given any graph sequence, $G_{t\geq 0}=(V,E_{t\geq 0})$ initially one node $v^*$ in graph $G_0$ is informed of the rumor, i.e., $I_0=\{v^{*}\}$. We recall the definition of the \pull, \push, and \pp protocols \cite{frieze1985shortest, karp2000randomized}. In the \pull model, in every round $t=0,1,\ldots$, every \emph{uninformed} vertex $v$ chooses a neighbor $u$ uniformly and independently at random. If $u$ is informed, then as a response $u$ transmits the rumor to $v$, so $v$  becomes informed. In the \push protocol, in each round, every \emph{informed} node $v$ chooses a neighbor $u$ uniformly at random, and transmits the rumor to $u$. Lastly, \pp is the combination of both strategies: In each round, if the node knows the rumor, it chooses a random neighbor to send the rumor to. Otherwise, it chooses a random neighbor to request the rumor from.

We can extend the \pull, \push and \pp models by including a credibility function $q(t)$ for $q(t) : \mathbb{N} \cup \{0\} \to [0,1]$ and $t\geq 0$. In the \pull, \push and \pp \emph{with credibility $q(t)$} models, at the beginning of each round $t=0,1,\ldots,$  for any uninformed node $ v \not\in I_{t-1}$ and for each transmission of the rumor to $v$  (regardless of whether that was due to a \push or \pull transmission), it becomes informed with w.p. $q(t)$ independently, and remains uninformed otherwise\footnote{Hence if in a round, an uninformed vertex receives $k$ transmissions (regardless of whether these are \pull or \push transmissions), then the probability it gets informed is $1- (1-q(t))^k$, i.e.\ each transmission is independent.}. This is depicted for the \pp model in \cref{alg:pull_model}. Notice that $q(t)$ may be time-dependent, and also that when $q(t)= q =1$ we return to the standard \pull, \push, and \pp models, whereas with $q(t) = q$ being a constant in $(0,1)$ we recover the ``uniform failure'' model studied in \cite{daknama_panagiotou_reisser_2021,doerr2017randomized}.

\begin{algorithm}
   \caption{Round $t \in \mathbb{N} \cup \{0\}$ of \pp with credibility function $q(t)$}
  \label{alg:pull_model}
   \begin{algorithmic}[1] 
       \State \textbf{Input}: $G_{t}, I_{t}, q(t)$
       \State \textbf{Initialize}: $\Delta_{t+1} \gets \emptyset$
                   \For{each $v \in I_{t}$}       \Comment{\push}
           \State Sample a neighbor $v' \in N_{G_{t}}(v)$ chosen uniformly at random.
          \If{$v' \not\in \Delta_{t+1}$} 
           \State With probability $q(t)$, $\Delta_{t+1} \gets \Delta_{t+1} \cup \{v'\}$ 
\EndIf
       \EndFor
       \For{each $v \in V \setminus I_{t}$}  \Comment{\pull}
           \State Sample a neighbor $v' \in N_{G_{t}}(v)$ chosen uniformly at random.
          \If{$v' \in I_{t}$} 
           \State With probability $q(t)$, $\Delta_{t+1} \gets \Delta_{t+1} \cup \{v\}$ 
\EndIf
       \EndFor
       \State $I_{t+1} \gets I_t \cup \Delta_{t+1}$
   \end{algorithmic}
\end{algorithm}
 
\subsection{Our Class of Spreading Processes}\label{sec:ourproc}

We now introduce two general spreading processes, that are crucial to our framework. This is an abstraction of the aforementioned examples of \push, \pull and \pp with credibility function $q(t)$, since we are now only considering the expected growth (or shrinking) factors. We point out that these may depend on several quantities such as the conductance of the informed set $I_t$ (or uninformed set $U_t$, respectively), and $q(t)$ of course.
 
\begin{definition}[Growing and Shrinking Processes]\label{def:processes}
Let $(G_t)_{t\geq 0}$ be a sequence of graphs. Let $\mathcal{P}$ be a stochastic process on $(G_t)_{t\geq 0}$ with a sequence of informed vertices $(I_t)_{t\geq 0} \subseteq V(G_t)$ and uninformed vertices  $U_i=V(G_t)\setminus I_t$ for all $t\geq 0$. We begin by defining the following property of such a process

\begin{itemize}
    \item \hypertarget{p3}{$\mathcal{P}_1$} (\textbf{Negative Correlation}): For any round $t \geq 0$ and any subset $S \subseteq U_t$,
    \[
      \Pr{t}{ \bigcap_{u \in S} \{ u \in I_{t+1} \} } \leq \prod_{u \in S} \Pr{t}{ u \in I_{t+1} }.
    \]
\end{itemize}For some time-independent value $\Cgrow > 0$ we say that  $\mathcal{P}$ is a $\Cgrow$-growing process if it satisfies \NC and 
\begin{itemize}
    \item \hypertarget{p1}{$\mathcal{P}_2$}  (\textbf{Monotonicity}): For any round $t \geq 0$, it holds deterministically that $I_t \subseteq I_{t+1}$ (and $|I_0|\geq 1$),
    \item \hypertarget{p2}{$\mathcal{P}_3$} (\textbf{Bounded Expected Growth}): For any round $t \geq 0$ the \emph{expected growth factor satisfies},
    \[
    \EX{t}{ \frac{|\Delta_{t+1}|}{|I_t|}  }\leq \Cgrow.
    \] 
\end{itemize}
Similarly, for some time-independent $\Cshrink < 1$,  $\mathcal{P}$ is a $\Cshrink$-shrinking process if it satisfies \NC and 
\begin{itemize}
\item \hypertarget{pt1}{$\widetilde{\mathcal{P}}_2$}  (\textbf{Monotonicity}): For any round $t \geq 0$, it holds deterministically that $U_t \supseteq U_{t+1}$ (and $|U_0|\leq n/2$),

    \item \hypertarget{pt2}{$\widetilde{\mathcal{P}}_3$} (\textbf{Bounded Expected Shrinking}): For any round $t \geq 0$ the \emph{expected shrinking factor satisfies},
   
     \[
   \EX{t}{ \frac{|\Delta_{t+1}|}{|U_t|} ~  } \leq \Cshrink.
    \] 
   \end{itemize}
\end{definition}
For convenience, we also define for all rounds $t \geq 0$ a ``combined'' growth/shrinking factor as
\[
 \delta_t := \EX{t}{\frac{|\Delta_{t+1}|}{\min\left(|I_t|,|U_t|\right)}} = \max\left( \EX{t}{\frac{|\Delta_{t+1}|}{|I_t|}}, \EX{t}{\frac{|\Delta_{t+1}|}{|U_t|}}\right).
\]
 We now prove that the negative correlation property immediately implies a strong upper bound on the variance of the growth (shrinking) factor. The same result was derived in~\cite{daknama_panagiotou_reisser_2021} for \push, \pull and \pp using the concept of self-bounding functions. 

 \begin{restatable}{lemma}{boundedvariance}\label{lem:bounded_variance_property}
Consider any stochastic process with sequence of informed vertices $(I_t)_{t \geq 0}$ satisfying $\NC$ . Then, also the following property also holds:
\begin{itemize}
    \item \hypertarget{p4}{$\mathcal{P}_4$} (\textbf{Bounded Variance}): For any round $t \geq 0$,
    \[
     \VAR{t}{|\Delta_{t+1}|} \leq \EX{t}{|\Delta_{t+1}|}.
    \] 
\end{itemize}
\end{restatable}

\begin{proof}
Let $X_u$ be the indicator which is $1$ if node $u$ gets informed through \push in round $t+1$, and $0$ otherwise. By the variance-covariance formula,
\begin{equation}
    \VAR{t}{|\Delta_{t+1}|} = \VAR{t}{\sum_{u \in U_t}X_u} = \sum_{u \in U_t}\VAR{t}{X_u} + \sum_{u \neq z \colon u,z\in U_t}\Cov{t}{X_u}{X_z}\label{eq:varpushdeltatp1}
\end{equation}
Let us start by analyzing $\Cov{t}{X_u}{X_z}$. We note that,
\begin{align*}
    \Cov{t}{X_u}{X_z} &= 
   \EX{t}{X_u\cdot X_z} 
    - \EX{t}{X_u}\cdot \EX{t}{X_z}\\
    &= \Pr{t}{ X_u = 1 \cap X_z = 1 } 
     - \Pr{t}{X_u = 1 } \cdot \Pr{t}{X_z=1} \\ 
         &\stackrel{(a)}{\leq} \Pr{t}{ X_u = 1} \cdot \Pr{t}{ X_z = 1 } 
     - \Pr{t}{X_u = 1 } \cdot \Pr{t}{X_z=1} \\
     &= 0, 
\end{align*}
where $(a)$ uses the negative correlation property (\NC).
Moreover, since $X_u$ is a Bernoulli random variable, we get that $\VAR{t}{X_u} = \Pr{t}{X_u =1}\cdot \Pr{t}{X_u=0} \leq \Pr{t}{X_u=1}.$
Returning to \cref{eq:varpushdeltatp1}, we obtain
\begin{equation*}
    \VAR{t}{|\Delta_{t+1}|} \leq \sum_{u\in U_t} \VAR{t}{X_u} \leq \sum_{u\in U_t} \Pr{t}{X_u=1} = \EX{t}{\Delta_{t+1}}.\qedhere 
\end{equation*} 
\end{proof}

\subsection{Specific Protocols and Growth Factors}\label{sec:growth}

In this subsection, we analyze specific protocols (in particular, \push, \pull and \pp with credibility function $q(t)$) and verify that they are $\Cgrow$-growing and $\Cshrink$-shrinking processes in the sense of \cref{def:processes}.

Let $(G_t)_{t\geq 0}$ be a sequence of regular graphs. Recall that in our setting $|I_0|=1$ and $\Delta_{t+1} = I_{t+1} \setminus I_t$. In order to capture the progress of the rumor spreading process between the rounds $t_1$ and $t_2$, we observe the following identities,
\begin{align*}
    & \frac{ |I_{t_2}| }{ |I_{t_1}|} = \prod_{t=t_1}^{t_2-1} \frac{| I_{t+1}|}{|I_t| } = \prod_{t=t_1}^{t_2-1} \frac{|I_t| +|\Delta_{t+1}|}{|I_t| } = \prod_{t=t_1}^{t_2-1} \left(1 + \frac{|\Delta_{t+1}|}{|I_t|}\right)\\
    & \frac{ |U_{t_2}| }{ |U_{t_1}|} = \prod_{t=t_1}^{t_2-1}\frac{|U_{t+1}|}{|U_t|} = \prod_{t=t_1}^{t_2-1}\frac{|U_t| + |\Delta_{t+1}|}{|U_t|}
 = \prod_{t=t_1}^{t_2-1} \left(1 - \frac{|\Delta_{t+1}|}{|U_t|}\right).
\end{align*}
As such, we prove upper and lower bounds on the expectation of the growth factor, $\frac{|\Delta_{t+1}|}{\min\left(|I_t|,|U_t|\right)}$ of the \push, \pull and \pp protocols.

\begin{lemma}\label{lem:PPP-P}
Let $t \geq 0$ be any round, $G_t$ a $d_t$-regular graph with $n$ vertices and $d_t\geq 1$, and $q(t)$ an arbitrary credibility. Then,
 \begin{enumerate}[(i)]\itemsep1pt
 	\item\label{itm:Push}   for the \push protocol,    \[
 			q(t) \cdot \left(1-\frac{q(t)}{2} \right) \cdot \varphi_t(I_t) \leq \EX{t}{\frac{|\Delta_{t+1}|}{\min( |I_t|,|U_t|)}}         \leq q(t) \cdot \varphi(I_t),
 			\]
 		\item\label{itm:Pull} for the \pull protocol,\begin{equation*}
 			\EX{t}{\frac{|\Delta_{t+1}|}{\min\left(|I_t|, |U_t|\right)}}  =  q(t) \cdot \varphi(I_t),
 		\end{equation*}  
	\item\label{itm:PP}  and for the \pp protocol,	\begin{equation*}
	  \frac{3}{2} \cdot q(t) \cdot \left(1- \frac{q(t)}{2}\right) \cdot \varphi(I_t)\leq     \EX{t}{\frac{|\Delta_{t+1}|}{\min\left(|I_t|, |U_t|\right)}}  \leq 2 \cdot q(t) \cdot  \varphi(I_t).
	\end{equation*}
\end{enumerate}

\end{lemma}

 \begin{proof}
For Item \eqref{itm:Push} we start with the lower bound. For any vertex $u \in U_t$, we note that, 
    \begin{equation}\label{eq:deltafrom1minus}
        \EX{t}{|\Delta_{t+1}|} = \sum_{u \in U_t}\Pr{t}{u \in I_{t+1}} = \sum_{u \in U_t} \left(1 - \Pr{t}{u \in U_{t+1}}\right).
    \end{equation}
    Further, as $1-x\leq e^{-x}$ for all $x \in \mathbb{R}  $, we have
    \begin{align}
        \Pr{t}{u \in U_{t+1}}& = 
        \prod_{v \in N(u)\cap I_t} \left(1 - \frac{q(t)}{d_t}   \right)  
        = \left(1 - \frac{q(t)}{d_t}\right)^{\deg_{I_t}(u)} \notag  \leq  \exp\left(-\frac{ q(t) \cdot \deg_{I_t}(u)}{ d_t } \right). \intertext{Since $\exp(x) \leq 1 +x+\frac{1}{2} x^2$ for any $x \in [-1,0]$ (which can be applied since $\frac{q(t) \cdot \deg_{I_t}(u)}{d_t} \leq 1$) we have} 
             \Pr{t}{u \in U_{t+1}}    &\leq  1 - \frac{ q(t) \cdot \deg_{I_t}(u)}{ d_t } + \frac{1}{2} \cdot \left( \frac{ q(t) \cdot \deg_{I_t}(u)}{ d_t } \right)^2\leq 1- \frac{  q(t) \cdot \left(1-\frac{q(t)}{2} \right) \cdot \deg_{I_t}(u)}{ d_t } ,\label{eq:push_grow_lower}
    \end{align}
      where the second inequality follows since  
      $\frac{\deg_{I_t}(u)}{d_t} \leq 1.$  Thus, by \eqref{eq:deltafrom1minus}
    \begin{align*}
     \EX{t}{|\Delta_{t+1}|} &\geq \sum_{u \in U_t} \frac{  q(t) \cdot \left(1-\frac{q(t)}{2} \right) \cdot \deg_{I_t}(u)}{ d_t } \\
    &= \frac{1}{d_t} \cdot q(t) \cdot \left(1-\frac{q(t)}{2} \right) \cdot e(I_t,U_t) \\
    &= q(t) \cdot \left(1-\frac{q(t)}{2} \right) \cdot \varphi(I_t) \cdot \min(|I_t|,|U_t|).
    \end{align*}

To prove the upper bound on $\Ex{\frac{|\Delta_{t+1}|}{\min(|I_t|,|U_t|)}}$ we will bound $|\Delta_{t+1}|$ from above. First, we note that
\begin{align*}
 \Pr{t}{u \in U_{t+1}} = \prod_{v \in N(u)\cap I_t} \left(1 - \frac{q(t)}{d_t}   \right)  &= \left(1 - \frac{q(t)}{d_t}\right)^{\deg_{I_t}(u)} \geq 1 - \deg_{I_t}(u) \cdot \frac{q(t)}{d_t}, 
        \end{align*}
where the last step follows by Bernoulli's inequality. Therefore,
 by \eqref{eq:deltafrom1minus} \begin{align*}
    \EX{t}{|\Delta_{t+1}|} \leq \sum_{u \in U_t} \deg_{I_t}(u) \cdot \frac{q(t)}{d_t} = \frac{1}{d_t}\cdot e(I_t,U_t) \cdot q(t) = q(t)\cdot \varphi_t(I_t) \cdot \min(|I_t|,|U_t|).
\end{align*}

  \noindent\textit{Proof of \eqref{itm:Pull}:}  Let $u \in U_t$ and let $X_u$ be the indicator variable which is $1$ if $u$ gets informed in round $t+1$ through \texttt{PULL}, and $0$ otherwise. We note that $\Pr{t}{X_u=1}= \frac{\deg_{I_t}(u)}{d_t} \cdot q(t)$. Hence,
    \begin{align*}
        \EX{t}{|\Delta_{t+1}|}  
        & = \sum_{u \in U_t} \Pr{t}{X_u=1} 
          = \sum_{u \in U_t} q(t) \cdot \frac{\deg_{I_t}(u)}{d_t}
         = q(t)\cdot \frac{1}{d_t}\cdot e(I_t, U_t) = q(t) \cdot \varphi(I_t) \cdot \min\left(|I_t|, |U_t|\right).
    \end{align*} 
\medskip 

\noindent\textit{Proof of \eqref{itm:PP}:}  The upper bound  follows by adding up the two upper bounds on $|\Delta_{t+1}|$ for \push and \pull.
 Observe that $ \Pr{t}{u \in U_{t+1}} = \left(1 - \frac{q(t) \cdot \deg_{I_t}{(u)}}{d_t} \right)\cdot         \prod_{v \in N(u)\cap I_t} \left(1 - \frac{q(t)}{d_t}   \right)  $. Thus, similarly to \cref{eq:push_grow_lower},
\begin{align*}
      \EX{t}{|U_{t+1}|} 
        &\leq \sum_{u \in U_t}\left(1 - \frac{q(t) \cdot \deg_{I_t}{(u)}}{d_t} \right)\cdot  \left( 1 - \frac{ q(t) \cdot \deg_{I_t}(u)}{ d_t } + \frac{1}{2} \cdot \left( \frac{ q(t) \cdot \deg_{I_t}(u)}{ d_t } \right)^2 \right)   
\\
&= \sum_{u \in U_t} \left(1 - \frac{2 \cdot q(t) \cdot \deg_{I_t}(u)}{d_t} + \frac{3}{2} \cdot \left( \frac{q(t) \cdot \deg_{I_t}(u)}{d_t} \right)^2 - \frac{1}{2} \cdot \left( \frac{q(t) \cdot \deg_{I_t}(u)}{d_t} \right)^3 \right).
\intertext{
Since for $z \in [0,1]$, $1-2 \cdot z + \frac{3}{2} \cdot z^2 - \frac{1}{2} \cdot z^3 \leq 1 - \frac{3}{2} \cdot z\left(1-\frac{z}{2} \right)$, it follows that}
 \EX{t}{|U_{t+1}|} &\stackrel{(a)}{\leq} \sum_{u \in U_t} \left(1 - \frac{\frac{3}{2}  \cdot q(t) \cdot  \left(1- \frac{q(t)}{2}\right) \cdot \deg_{I_t}(u)}{d_t} \right) \\
&\stackrel{(b)}{=} |U_t| - \frac{3}{2} \cdot q(t) \cdot \left(1-\frac{q(t)}{2} \right) \cdot \varphi_t(I_t) \cdot \min\left(|I_t|, |U_t|\right) ,
\end{align*}
where $(a)$ used that $\frac{\deg_{I_t}(u)}{d_t} \leq 1$ and $(b)$ that $e(I_t,U_t) = d_t \cdot \varphi_t(I_t) \cdot \min\left(|I_t|, |U_t|\right)$.
\end{proof}

\renewcommand{\arraystretch}{2}
\begin{table}[t]
\begin{center}
\begin{tabular}{|c|c|c|}
\hline
\multirow{2}{*}{} & \multicolumn{2}{c|}{$\delta_t$}\\ \cline{2-3}
                 & Lower Bound & Upper Bound           \\ \hline\hline
\pull               & \multicolumn{2}{c|}{$q(t) \cdot \varphi(I_t)$}             \\ \hline
\push              & $q(t) \cdot \left(1-\frac{q(t)}{2} \right) \cdot \varphi_t(I_t)$ & $q(t) \cdot \varphi(I_t)$  \\ \hline
\pp               & $\frac{3}{2} \cdot q(t) \cdot \left(1- \frac{1}{2} q(t) \right) \cdot \varphi_t(I_t)$  & $2 \cdot q(t) \cdot \varphi(I_t)$ \\ \hline
\end{tabular}
\caption{Basic lower and upper bounds on the expected growth factor $\delta_t$ for \push, \pull and \pp in terms of $q(t)$ and the conductance $\varphi(I_t)$ on regular graphs.}
\label{tab:GrowthFactors}
\end{center}
\end{table}

Next we prove tighter bounds for the \push and \pp protocol if the graph is a strong expander.
\begin{lemma}\label{lem:expushSTRONG}
Consider the \push protocol, and let $t \geq 0$ be any round where with $|I_t|\leq n/2$ and  $G_t$ a $d_t$-regular graph with $n$ vertices. Then, for $q(t)$ an arbitrary credibility and $\beta := \lambda + \frac{|I_t|}{n}$,
\[
 \EX{t}{ \frac{|\Delta_{t+1}|}{|I_t|} } \geq q(t) \cdot \left(1 - 7 \sqrt{\beta} \right).
\]
For the same setting in the \pp protocol,
\[
 \EX{t}{ \frac{|\Delta_{t+1}|}{|I_t|} } \geq q(t) \cdot \left(2 - 12 \sqrt{\beta} \right). 
\]
\end{lemma} 
\begin{proof}
Let us first prove the result for the \push protocol.
We follow a similar approach to \cite[Proof of Lemma 2.5]{PPSS15}. 
We define
\[
A:=\left\{ u \in U_t \colon \deg_{I_t}(u) \geq 2d_t \cdot \sqrt{\beta} \right\}.
\]
By definition of $A$ and \cref{lem:weakexpandermixinglemma},
\begin{equation}\label{eq:A_bound}
 |A| \cdot 2d_t \sqrt{\beta} \leq e(A,I_t) \leq \frac{d_t |A| \cdot |I_t|}{n} + \lambda d_t \sqrt{|A| \cdot |I_t|} \leq \frac{d_t |A| \cdot |I_t|}{n} + d_t \beta \sqrt{|A| \cdot |I_t|}.
\end{equation}
This implies
\[
 |A| \cdot 2 \sqrt{\beta} \leq  \frac{ |A| \cdot |I_t|}{n} + \beta \sqrt{|A| \cdot |I_t|}.
\]
Dividing by $|A|$ and rearranging this, we can upper bound $|A|$ as follows
\begin{align*}
|A| &\leq \beta^2 \cdot |I_t| \cdot \left( \frac{1}{2\sqrt{\beta} - \frac{|I_t|}{n}} \right)^2.
\end{align*}
We assume that $\beta \leq 1/49$ (as otherwise the lower bound would trivial). Therefore, we have $2 \sqrt{\beta} \geq 2 \beta \geq 2 \cdot \frac{|I_t|}{n}$, since $\beta := \lambda + \frac{|I_t|}{n}$ and $\lambda \geq 0$. Thus, the denominator is bounded from below by $\sqrt{\beta}$, which yields
\begin{equation}
    |A| \leq \beta \cdot |I_t|. \label{eq:UBA}
\end{equation}
Next define $B:= U_t \setminus A$. By a corollary of the strong expander mixing lemma (\cref{cor:exmixingcor}) and \eqref{eq:A_bound},
\begin{align}
    e(B,I_t) &= e(U_t,I_t) - e(A,I_t) \notag \\
    &\stackrel{(a)}{\geq} (1-\lambda) \cdot d_t \cdot \frac{|U_t| \cdot |I_t|}{n} - \frac{d_t |A| |I_t|}{n} - d_t \cdot \beta \sqrt{|A| |I_t|} \notag  \\
   &\stackrel{(b)}{\geq} \left(1 - \lambda - \frac{|A|}{|U_t|} \right) \cdot d_t \cdot \frac{|U_t| \cdot |I_t|}{n} - d_t \cdot \beta^{3/2} \cdot |I_t| \notag  \\
   &\stackrel{(c)}{\geq} \left(1 - 2 \left(\lambda + \frac{|I_t|}{n}\right) - 2 \beta^{3/2} \right) \cdot d_t \cdot \frac{|U_t| \cdot |I_t|}{n} \notag  \\
   & \stackrel{(d)}{\geq} (1 - 4 \beta) \cdot d_t \cdot \frac{(n-|I_t|) \cdot |I_t|}{n} \notag \\
   &\geq (1 - 5 \beta) \cdot d_t \cdot |I_t|, \label{eq:first-eq}
\end{align}
where $(a)$ holds by applying \cref{lem:strongexpandermixinglemma} to $e(U_t,I_t)$ and \cref{eq:A_bound}, $(b)$ by \cref{eq:UBA}, $(c)$ since $|U_t|\geq n/2$ and $(d)$ by the definition of $\beta$ and as $\beta \leq 1/49 \leq 1$ (and thus $\beta^{3/2} \leq \beta$).
Therefore, by \cref{eq:push_grow_lower}
\begin{align*}
    \EX{t}{|U_{t+1} \cap B|} &\leq \sum_{u \in U_t \cap B} \left(1 - \frac{ q(t) \cdot \deg_{I_t}(u)}{ d_t } + \frac{1}{2} \cdot \left( \frac{ q(t) \cdot \deg_{I_t}(u)}{ d_t } \right)^2 \right) \\
    &\leq \sum_{u \in U_t \cap B} \left( 1 - \frac{ q(t) \cdot \deg_{I_t}(u)}{ d_t } \cdot \left(1 - \frac{ q(t) \cdot \deg_{I_t}(u)}{ d_t }\right) \right) \\
    &\stackrel{(a)}{\leq} \sum_{u \in U_t \cap B} \left( 1 - \frac{ q(t) \cdot \deg_{I_t}(u)}{ d_t } \cdot \left(1 - 2 \sqrt{\beta} \right) \right) \\
    &= |U_{t} \cap B| - \left(1 - 2 \sqrt{\beta} \right)  \cdot \frac{q(t)}{d_t} \cdot e(B,I_t) \\
    &\stackrel{(b)}{\leq} |U_{t} \cap B| - \left(1 - 2 \sqrt{\beta} \right)  \cdot \left(1 - 5 \beta \right) \cdot q(t) \cdot |I_t|,
\end{align*}
where $(a)$ used the definition of $B$ and $q(t) \leq 1$, and $(b)$ used \cref{eq:first-eq}. Rearranging and using that $\beta \leq 1$ yields
\begin{align*}
    \EX{t}{ |\Delta_{t+1}|} \geq \left(1 - 7 \sqrt{\beta} \right)   \cdot q(t) \cdot |I_t|.
\end{align*}
This concludes the proof for \push. 

We now turn to the \pp protocol, where the derivation is almost identical to the one of \push. Here we use that in the \pp protocol, a node remains uniformed if and only if it does not get informed by a \push transmission \emph{and} if it does not get informed by a \pull call. Since these two events are independent, we conclude
\begin{align*}
        \Pr{t}{u \in U_{t+1}} &= \left(1 - \frac{q(t)}{d_t} \right)^{\deg_{I_t}(u)} \cdot \left(1 - \frac{q(t) \cdot \deg_{I_t}(u)}{d_t} \right) \\
  &\leq \left(1 - \frac{q(t)}{d_t}\cdot \deg_{I_t}(u) + \frac{1}{2}\cdot \frac{q(t)}{d_t}\deg_{I_t}(v)\right)\cdot \left(1 - \frac{q(t)\cdot \deg_{I_t}(v)}{d_t}\right)\\
  &\leq 1 - 2 \cdot \frac{q(t) \cdot \deg_{I_t}(u)}{d_t} + \left( \frac{q(t) \cdot \deg_{I_t}(u)}{d_t} \right)^2.
\end{align*}
Now the analogous derivation as for \push yields the claim.
\end{proof}

The next lemma improves over the lower and upper bound in \cref{lem:PPP-P} \eqref{itm:Push} if $|I_t| \geq n/2$. 
Concerning the lower bound, we have $q(t) \cdot (1- \frac{q(t)}{2}) \leq 1-e^{-q(t)}$ since $e^{-z} \leq 1- z+ \frac{1}{2} z^2 = 1- z \cdot \left(1-\frac{z}{2} \right)$ for $z \in [0,1]$. Further, if $d_t=\omega(1)$ and $q(t) \cdot \varphi_t(I_t)$ is bounded below by a constant, then the upper bound below is tighter as $1-\exp(-x) \leq x$ for any $x \in \mathbb{R}$.
\begin{lemma}\label{lem:expushUT}
 
Consider the \push protocol, and let $t \geq 0$ be any round, $G_t$ is a $d_t$-regular graph with $n$ vertices and $q(t)$ an arbitrary credibility. Then,
 \begin{enumerate}[(i)]\itemsep1pt
 	\item\label{itm:push1}  $\EX{t}{ \frac{|\Delta_{t+1}|}{|U_t|}  } \geq  \left(1 - e^{-q(t)} \right) \cdot \varphi(I_t).$
 		\item\label{itm:push2connected} If $G_t$ is connected, then,
\begin{align*}
\EX{t}{ \frac{|\Delta_{t+1}|}{|U_t|} } &\leq 1 - e^{- \varphi(I_t) \cdot q(t)  } \cdot 
     \left(1 - \frac{\varphi(I_t) \cdot (q(t))^2}{d_t} \right).
\end{align*}
	
\end{enumerate}

\end{lemma}
\begin{proof}
We proceed by case distinction. We first consider the case where $d_t=1$ and $q(t)=1$. We note that for $d_t=1$ the graph $G_t$ is a perfect matching. We note that then,
\begin{equation*}
    \EX{t}{|U_{t+1}|} = |U_{t}| - |\Delta_{t+1}| = |U_t| - e(U_t,I_t).
\end{equation*}
Thus,
\begin{align*}
    \EX{t}{\frac{|\Delta_{t+1}|}{|U_t|}} = \frac{e(U_t,I_t)}{|U_t|} = \frac{\varphi(I_t)\cdot |U_t|}{|U_t|} = \varphi(I_t).
\end{align*}
Clearly, $\varphi(I_t)\geq (1 - e^{-1})\cdot \varphi(I_t)$, proving statement $(i)$. Moreover, $1 - e^{-\varphi(I_t)} \cdot (1 - \varphi(I_t))\geq 1 - (1 - \frac{\varphi(I_t)}{2})\cdot(1 - \varphi(I_t)) \geq \varphi(I_t)$ and thus statement $(ii)$ holds as well for this case. Now, we consider the scenario where $d_t > 1$ or  ($q_t < 1$ and $d_t \geq 1$). For the \push protocol,
    \begin{align}
        \EX{t}{ |U_{t+1}| } &= \sum_{u \in U_t} \left(1 - \frac{q(t)}{d_t} \right)^{\deg_{I_t}(u)} =: g.\label{eq:2.5pushexpressionG}
        \end{align}
        For fixed $\sum_{u \in U_t} \deg_{I_t}(u) = e(I_t,U_t)$, our goal is to estimate the last expression, viewed as a function over $(\deg_{I_t}(u))_{u \in U_t} \in \{0,1,\ldots,d_t\}$, using a Schur-convexity argument. To avoid discretization issues, we will first define a more ``generous'' function
         which has $|U_t| \cdot d_t$ real-valued variables $(z_k(u))_{u \in U_t, 1 \leq k \leq d_t} \in [0,d_t]$:
    \begin{align}
        h:= \frac{1}{d_t} \cdot \sum_{u \in U_t} \sum_{k=1}^{d_t} \left(1 - \frac{q(t)}{d_t} \right)^{z_k(u)},
        \label{eq:discretisation}
    \end{align}
    where $\sum_{u \in U_t} \sum_{k=1}^{d_t} z_k(u) =  e(I_t,U_t) \cdot d_t$.  Clearly, the maximum of $g$ is upper bounded by the maximum of $h$ (by choosing for all $u \in V$ and $1 \leq k \leq d_t$, $z_k(u) = \deg_{I_t}(u)$). Let us consider the function $f(x,z):= (1 - x)^z$ where $x$ corresponds to $\frac{q(t)}{d_t}$ in $h$ and $z$ to $z_k(u)$. 
    We have that $\frac{\partial f(x,z)}{\partial z}= (1 -x)^z \cdot \log\left(1 - x\right)$ and $\frac{\partial^2 f(x,z)}{\partial^2 z} = (1 - x)^z \cdot  \left(\log\left(1 - x\right)\right)^2 $. We note that the second derivative is greater or equal to $0$ for any $x \leq 1$. This proves that $f(x,z)$ is convex in $z$. Since $h$ (as a function in $(z_k)_{u \in U_t, 1 \leq k \leq d_t}$) is a sum of convex functions, $h$ itself is also convex. Further, as $h$ is symmetric in $(z_k)_{u \in U_t, 1 \leq k \leq d_t}$, it follows that $h$ is Schur-convex.
    \medskip
    
    \noindent\textit{Proof of \eqref{itm:push1} for $d_t >1$ or $(q_t < 1$ and $d_t \geq 1)$}: Since $z_k(u) \in [0,d_t]$, the function $h$ is maximized if all $z_k(u) \in \{0,d_t\}$, and therefore
        \begin{align*}
         \EX{t}{ |U_{t+1}| } &\leq \max_{(z_k(u))_{u \in U_t, 1 \leq k \leq d_t}} h \\ &\leq \frac{1}{d_t} \cdot \left( \frac{e(I_t,U_t) \cdot d_t}{d_t}  \cdot \left(1 - \frac{q(t)}{d_t} \right)^{d_t}  + \left(|U_t| \cdot d_t -  \frac{e(I_t,U_t) \cdot d_t}{d_t}  \right) \cdot \left(1 - \frac{q(t)}{d_t} \right)^{0} \right) \\
      &\leq  \frac{e(I_t, U_t)}{d_t}  \cdot \exp\left( -q(t) \right)  + \left(|U_t|  -  \frac{e(I_t,U_t)}{d_t}  \right) \cdot 1 \\ 
      &= |U_t| - \frac{e(I_t,U_t)}{d_t} \cdot \left(1 - e^{-q(t)} \right).
    \end{align*}
   Since $e(I_t,U_t) = d_t \cdot \varphi(I_t) \cdot |U_t|$,
        \begin{align*}
            \EX{t}{|\Delta_{t+1}|} &\geq \left(1 - e^{-q(t)} \right) \cdot \varphi(I_t) \cdot |U_t|.
        \end{align*}

    \noindent\textit{Proof of \eqref{itm:push2connected} for $d_t >1$ or $(q_t < 1$ and $d_t \geq 1)$}: Now, we assume that $G_t$ is connected, and thus $d_t \geq 2$. We note that the Schur-convex function $h$ is minimized if all arguments are equal, and therefore
\begin{align}
     \EX{t}{ |U_{t+1}| } &\geq \min_{(z_k(u))_{u \in U_t, 1 \leq k \leq d_t}} h \notag \\ 
     &\geq \frac{1}{d_t} \sum_{u \in U_t} \sum_{k=1}^{d_t} \left(1 - \frac{q(t)}{d_t} \right)^{ \frac{e(U_t,I_t)}{|U_t|}} \notag \\
     &\stackrel{(a)}{\geq} |U_t| \cdot \exp\left(- \frac{q(t)}{d_t} \cdot \frac{e(I_t,U_t)}{|U_t|} \cdot \left(1 + \frac{q(t)}{d_t} \right)  \right) \notag \\
     &\stackrel{(b)}{=} |U_t| \cdot \exp\left(- \varphi(I_t) \cdot q(t)  \right) \cdot 
     \exp\left(- \frac{\varphi(I_t) \cdot (q(t))^2}{d_t} \right) \notag \\
     &\stackrel{(c)}{\geq} |U_t| \cdot \exp\left(- \varphi(I_t) \cdot q(t)  \right) \cdot 
     \left(1 - \frac{\varphi(I_t) \cdot (q(t))^2}{d_t} \right) \label{eq:push_derivation},
\end{align}
having used in $(a)$ that $1-x \geq e^{-x-x^2}$ for $x \in [0,0.5]$ (here we used that $d_t \geq 2$, since $G_t$ is assumed to be a regular and connected graph); in $(b)$ that $e(I_t,U_t) = d_t \cdot \varphi(I_t) \cdot |I_t|$; and in $(c)$ that $\exp(-x) \geq 1-x$ for any $x \in \mathbb{R}$.
Therefore,
        \begin{align*}
            \EX{t}{|\Delta_{t+1}|} &\leq \left(1 - e^{- \varphi(I_t) \cdot q(t) } \cdot 
     \left(1 - \frac{\varphi(I_t) \cdot (q(t))^2}{d_t} \right)\right)\cdot |U_t|,
        \end{align*}as claimed.
\end{proof}
We now turn to the \pp protocol, and the next lemma improves over \cref{lem:PPP-P} \eqref{itm:PP}. 
\begin{lemma}\label{lem:LBexPPT}
Consider the \pp protocol, and let $t \geq 0$ be any round, $G_t$ is a $d_t$-regular graph with $n$ vertices and $q(t)$ an arbitrary credibility. Then,
     \begin{enumerate}[(i)]\itemsep1pt
 	\item\label{itm:ppPPP1}  $ \EX{t}{ \frac{|\Delta_{t+1}|}{|U_t|} } \geq   \Bigl( 1 - e^{-q(t)} \cdot \left(1 - q(t) \right) \Bigr) \cdot \varphi(I_t).$
 	\item\label{itm:ppPPP2}  $ \EX{t}{ \frac{|\Delta_{t+1}|}{|U_t|}  } \leq 1 - \left(1 - q(t) \right)^{  \varphi(I_t)} \cdot \left(1 - q(t) \cdot \varphi(I_t) \right).$	
\end{enumerate}
\end{lemma}

\begin{proof}
Let us first consider the case where $q(t)=1$ and $d_t=1$. Similarly as in the proof of \cref{lem:expushUT}, we note that for $d_t=1$ the graph $G_t$ is a perfect matching. We note that,
\begin{equation*}
    \EX{t}{|U_{t+1}|} = |U_{t}| - |\Delta_{t+1}| = |U_t| - e(U_t,I_t).
\end{equation*}
Thus,
\begin{align*}
    \EX{t}{\frac{|\Delta_{t+1}|}{|U_t|}} = \frac{e(U_t,I_t)}{|U_t|} = \frac{\varphi(I_t)\cdot |U_t|}{|U_t|} = \varphi(I_t).
\end{align*}
Clearly, $\varphi(I_t)\geq \Bigl( 1 - e^{-q(t)} \cdot \left(1 - q(t) \right) \Bigr) \cdot \varphi(I_t)$ and since $1 - \left(1 - q(t)\right)^{\varphi(I_t)}\cdot \left(1 - q(t)\cdot \varphi(I_t)\right) \geq 1 - \left(1 - q(t) \cdot \varphi(I_t)\right)\left(1 - q(t)\cdot \varphi(I_t)\right)\geq 1 - (1 - q(t)\cdot \varphi(I_t)) = \varphi(I_t)$, where the first inequality holds by Bernoulli's inequality (since $\varphi(I_t) \in [0,1]$), and in the last equality we use that $q(t)=1$. Hence, we have also proven statement $(ii)$ for this case. Let us now consider the scenario where $d_t > 1$ or ($d_t\geq 1$ and $q(t) < 1$). We note that for the \pp protocol, a node remains uniformed if and only if it does not get informed by a \push transmission \emph{and} if it does not get informed by a \pull call. Since those two events are independent, we conclude
\begin{align*}
        \Pr{t}{u \in U_{t+1}} &= \left(1 - \frac{q(t)}{d_t} \right)^{\deg_{I_t}(u)} \cdot \left(1 - \frac{q(t) \cdot \deg_{I_t}(u)}{d_t} \right).
        \end{align*}
        Hence,
        \[
 \EX{t}{|U_{t+1}|} = \sum_{u \in U_t}  \left(1 - \frac{q(t)}{d_t} \right)^{\deg_{I_t}(u)} \cdot \left(1 - \frac{q(t) \cdot \deg_{I_t}(u)}{d_t} \right) =: g.
        \]
        Similarly, as in the proof of \cref{lem:expushUT} let us define the function $h$ on $|U_t|\cdot d_t$ real-valued variables $(z_k(u))_{u \in U_t, 1 \leq k \leq d_t} \in [0,d_t]$ to avoid discretization issues.  
        \begin{equation*}
            h:= \frac{1}{d_t} \cdot \sum_{u \in U_t}\sum_{k=1}^{d_t} \left(1 - \frac{q(t)}{d_t}\right)^{z_k(u)}\cdot \left(1 - \frac{q(t) \cdot z_k(u)}{d_t}\right),
        \end{equation*}
        with the additional constraint that $\sum_{u \in V} \sum_{k=1}^{d_t} z_k(u) = \sum_{u \in U_t} \deg_{I_t}(u) \cdot d_t$. We note that the maximum of $g$ is upper bounded by the maximum of $h$, and similarly, the minimum of $g$ is lower bounded by the minimum of $h$. 
        In order to show that $h$ is Schur-convex, we define the function
        \[
f(x,z):= (1-x)^z \cdot (1- x \cdot z),
        \]
        where $x$ corresponds to $\frac{q(t)}{d_t}$ and $z$ corresponds to $z_k(u)$ for some $1 \leq k \leq d_t$ and some $u \in V$.
        The  first derivative is 
        \[
        \frac{\partial f(x,z)}{\partial z} = -(1-x)^z \cdot \left( (xz-1) \log(1-x) + x \right), \]
        and the second derivative is 
        \[
        \frac{\partial^2 f(x,z)}{\partial^2 z} := -(1-x)^z \cdot \log(1-x) \cdot \left[ (zx-1) \log(1-x) + 2 x \right].
        \]
        We claim that $\frac{\partial^2 f(x,z)}{\partial^2 z} \geq 0$.
        Note that $0 \leq x = \frac{q(t)}{d_t} \leq \frac{1}{d_t}$ and $0 \leq z = z_k(u) \leq d_t$. Hence, we always have
        \[
 -(1-x)^{z} \cdot \log(1-x) > 0,
        \]
        since both factors are strictly negative. Therefore, it remains to verify that
        \[
 (zx-1) \log(1-x) + 2 x \geq 0.
        \]
        However, this holds since $zx-1 \leq d_t \cdot \frac{1}{d_t} - 1 \leq 0$ and $\log(1-x) < 0$, as well as $x \geq 0$.
        This proves that $f(x,z)$ is convex in $z$. Since $h$ is a sum of convex functions, $h$ (as a function in $(z_k)_{u \in U_t, 1 \leq k \leq d_t}$) itself is also convex. Further, as $h$ is symmetric in $(z_k)_{u \in U_t, 1 \leq k \leq d_t}$, it follows that $h$ is Schur-convex.
        \medskip
        
    \noindent\textit{Proof of \eqref{itm:ppPPP1} for $d_t >1$ or $(q_t < 1$ and $d_t \geq 1)$}: Let us now prove the first statement. We have,
    \begin{align*}
             \EX{t}{|U_{t+1}|} &\leq  \max_{(z_k(u))_{u \in U_t, 1 \leq k \leq d_t}} h\\
             &\leq \frac{1}{d_t}\cdot \left(\frac{e(I_t,U_t) \cdot d_t}{d_t}\cdot \left(1-\frac{q(t)}{d_t}\right)^{d_t} \cdot \left(1 - q(t)\right) + \left(|U_t| \cdot d_t - \frac{e(I_t,U_t) \cdot d_t}{d_t}\right)\cdot \left(1 - \frac{q(t)}{d_t}\right)^0 \cdot 1\right)\\
             &\leq \frac{e(I_t,U_t)}{d_t}\cdot \exp\left(-q(t)\right)\cdot \left(1 - q(t)\right) + \left(|U_t| - \frac{e(I_t,U_t)}{d_t}\right)\cdot 1\\
             &= |U_t| - \frac{e(I_t,U_t)}{d_t}\left(1 - \exp\left(-q(t)\right)\cdot\left(1 -q(t)\right)\right).
        \end{align*}
        Since $e(I_t,U_t)=\varphi(I_t) \cdot |U_t| \cdot d_t$, it follows that
        \[
 \EX{t}{|\Delta_{t+1}|} \geq \varphi(I_t) \cdot \Bigl( 1 - e^{- q(t)}\cdot \left(1 - q(t)\right)\Bigr)\cdot |U_t|.
        \]
        
        \noindent\textit{Proof of \eqref{itm:ppPPP2} for $d_t >1$ or $(q_t < 1$ and $d_t \geq 1)$}: For the upper bound on $\EX{t}{\frac{|\Delta_{t+1}|}{|U_t|} ~\Big|~ |I_t| \geq n/2}$, the Schur convex function $h$ is minimized if all arguments are equal, and therefore,
        \begin{align*}
            \EX{t}{|U_{t+1}|} &\geq \min_{(z_k(u))_{u \in U_t, 1 \leq k \leq d_t}} h\\
            & \geq \frac{1}{d_t} \sum_{u \in U_t}\sum_{k=1}^{d_t} \left(1 - \frac{q(t)}{d_t}\right)^{\frac{e(U_t,I_t)}{|U_t|}}\cdot \left(1 - \frac{q(t) \cdot e(U_t,I_t)}{d_t \cdot |U_t|}\right)\\
            & \stackrel{(a)}{=} |U_t| \cdot \left(1 - \frac{q(t)}{d_t}\right)^{\frac{e(U_t, I_t)}{|U_t|}}\cdot \left(1 - q(t) \cdot \varphi(I_t) \right)\\
            & = |U_t| \cdot \left(\underbrace{\left(1 - \frac{q(t)}{d_t}\right)^{\frac{d_t}{q(t)}}}_{:=A}\right)^{\frac{q(t)}{d_t}\cdot \frac{e(U_t, I_t)}{|U_t|}}\cdot \left(1 - q(t) \cdot \varphi(I_t) \right)\\
            & \stackrel{(b)}{\geq} |U_t| \cdot \left((1 - q(t))^{\frac{1}{q(t)}}\right)^{\frac{q(t)}{d_t}\cdot \frac{e(U_t, I_t)}{|U_t|}} \cdot \left(1 - q(t) \cdot \varphi(I_t) \right)\\
            & \stackrel{(c)}{=} |U_t| \cdot \left(1 - q(t) \right)^{  \varphi(I_t)} \cdot \left(1 - q(t) \cdot \varphi(I_t) \right),
        \end{align*}
        where in $(a)$ and $(c)$ we used that  $e(I_t,U_t) = d_t \cdot \varphi_t(I_t) \cdot |U_t|$ and $(b)$ holds since $A$, for fixed $q(t)$, is non-decreasing in $d_t$ and thus minimized for $d_t = 1$ (as we assume that $d_t\geq 1$). Thus,
        \begin{equation*}
            \EX{t}{\frac{|\Delta_{t+1}|}{|U_t|}} \leq 1 - \left(1 - q(t) \right)^{  \varphi(I_t)} \cdot \left(1 - q(t) \cdot \varphi(I_t) \right).\qedhere
        \end{equation*}
     
\end{proof}

A summary of these tighter bounds for \push, \pull and \pp is given in \cref{tab:GrowthFactorsTight}, and the more simple bounds are summarized in \cref{tab:GrowthFactors}. For strong expanders, similar bounds have been derived in \cite{daknama_panagiotou_reisser_2021,PPS15}.  

\renewcommand{\arraystretch}{2}
\begin{table}[ht]
\begin{center}
\begin{tabular}{|c|c|c|c|}
\hline
\multirow{2}{*}{} & \multicolumn{1}{c|}{$\delta_t$, $1 \leq |I_t| \leq n/2$} & \multicolumn{2}{c|}{$\delta_t$, $n/2 \leq |I_t| \leq n$}\\ \cline{2-4}
                 & Lower Bound  & Lower Bound & Upper Bound           \\ \hline \hline
\pull            & $q(t) \cdot \varphi(I_t)$ & \multicolumn{2}{c|}{$q(t) \cdot \varphi(I_t)$}\\ \hline
\push              & \small{$q(t) \cdot \Bigl(1 - 7 \sqrt{ \lambda + \frac{|I_t|}{n} } \Bigr)$} & \small{$ \Bigl(1 -  e^{-q(t)}  \Bigr) \cdot \varphi(I_t) $} & \small{$q(t) \cdot \varphi(I_t)$} \\ \hline
\texttt{P-P}               & \small{$ q(t) \cdot \Bigl( 2 - 12 \sqrt{ \lambda + \frac{|I_t|}{n} } \Bigr)$}  & \footnotesize{$\Bigl( 1 - e^{- q(t)} \cdot   \left(1 - q(t) \right)  \Bigr)\cdot \varphi(I_t)$} &\footnotesize{  $1 - \left(1 - q(t) \right)^{  \varphi(I_t)} \cdot \left(1 - q(t) \cdot \varphi(I_t) \right)$}      \\ \hline
\end{tabular}
\caption{Refined bounds in terms of $q(t)$ and the spectral expansion $\lambda$ on the expected growth factors of \push and \pp on regular graphs. These bounds are tighter than the more basic ones (see~\cref{tab:GrowthFactors}), whenever $\lambda=o(d_t)$ (which also implies $\varphi(I_t)=1-o(1)$ if $|I_t|=o(n)$ as well as $\varphi(I_t)=1-o(1)$ if $|U_t|=o(n)$). The $1-o(1)$ terms in the two upper bounds go to $1$ if $d_t \rightarrow \infty$ or $\varphi \rightarrow 0$ or $q(t) \rightarrow 0$ for all $t\geq 0$. }
\label{tab:GrowthFactorsTight}
\end{center}
\end{table}

Next, we state a simple but crucial fact:
\begin{restatable}{lemma}{negcorrelation}\label{lem:neg_correlation}
Let $(G_t)_{t\geq 0}$ be a sequence of $d_t$-regular graphs with $n$ vertices and let $q(t)$ be an arbitrary credibility function . Then, the \push, \pull and \pp protocol satisfy the negative correlation property (see \cref{def:processes}).  
\end{restatable}

\begin{proof}
For the \pull protocol, the property clearly holds (even with equality). Consider now the \push protocol (the case of \pp is analogous). Let $S=\{u_1,u_2,\ldots,u_s\} \subseteq U_t$ with $s=|S|$ . Then, expressing the left-hand side via conditional probabilities,
\begin{align*}
     \Pr{t}{ \bigcap_{u \in S} \left\{ X_u = 1 \right\} } &= 
     \prod_{i=1}^s \Pr{t}{ X_{u_i} = 1 ~\Bigg|~ \bigcap_{j=1}^{i-1} \{ X_{u_j} = 1 \}  } \leq
     \prod_{i=1}^s \Pr{t}{ X_{u_i} = 1  },
\end{align*}
where the last inequality holds, by noting that conditioning on some other vertices $u_1,\ldots,u_{i-1}$ getting informed may decrease the probability of $u_{i}$ getting informed by \push (for $u_i$ getting informed by \pull, the probability is always the same, regardless of the conditioning).
\end{proof}

Finally, we close this section by verifying that \push, \pull and \pp satisfy the condition in \cref{def:processes} for certain $\Cgrow$ and $\Cshrink$. Note that even for static graphs, \pull and \pp require a restriction on $q(t)$; this is since if $q(t)=1$, then on certain graphs (like the complete graph), \pull and \pp would only need $O(\log \log n)$ steps in the shrinking phase. However, for dynamic graphs, even for \push we require a restriction on $q(t)$; this is because otherwise $G_t$ could be a $1$-regular graph, i.e., a perfect matching so that each vertex in $U_t$ is matched to a vertex in $I_t$.

 \begin{restatable}{lemma}{protocolgrowth}\label{lem:protocol_growth}
 
   Let $(G_t)_{t \geq 0}$ be any sequence of $d_t$-regular graphs and let $q(t)$ be an arbitrary credibility function.
    \begin{enumerate}[(i)]\itemsep1pt
    \item\label{itm:PROTGROWTH_push}  
         The \push protocol is a $1$-growing process. Furthermore, if $q(t) \leq 1 - \epsilon$, for $\epsilon > 0$ (not necessarily constant), then the \push protocol is a $(1-\epsilon)$-shrinking process. Also, if all graphs in the sequence $(G_t)_{t \geq 0}$ are connected, then the
         \push protocol is a $(1-e^{-1} \cdot \frac{1}{2}$)-shrinking process.
      
        \item\label{itm:PROTGROWTH_pull}
        The \pull protocol is a $1$-growing process. Furthermore, if $q(t) \leq 1 - \epsilon$, for $\epsilon > 0$ (not necessarily constant), then the \pull protocol is a $(1-\epsilon)$-shrinking process.
    
        \item\label{itm:PROTGROWTH_pp} The \pp protocol is a $2$-growing process. Furthermore, if $q(t) \leq 1 -\epsilon$, for $\epsilon > 0$ (not necessarily constant), then the \pp protocol is a $(1-\epsilon^2)$-shrinking process. 
  
    \end{enumerate}
\end{restatable}

\begin{proof}
Firstly, we note that since in \pull, \push, and \pp, no informed node can become uninformed, they all satisfy (\Mono) (the monotonicity property). Moreover, (\NC) (the negative correlation property) holds for these three protocols by  \cref{lem:neg_correlation}. What is left to prove is (\BEG) (and $\BEGT$) for \pull, \push and \pp. 

 For \push,
\begin{align*}
    \EX{t}{\frac{|\Delta_{t+1}|}{|I_t|}} \leq \EX{t}{\frac{|\Delta_{t+1}|}{\min(|I_t|, |U_t|)}} \stackrel{(a)}{\leq} q(t) \cdot \varphi(I_t) \leq q(t),
\end{align*}
where $(a)$ holds by \cref{lem:PPP-P} \eqref{itm:Push}, and hence \push is a $1$-growing process. Furthermore, if $q(t) \leq 1 -\epsilon$, then by \cref{lem:PPP-P} $(i)$, we have for any $t \geq 0$,
\begin{align*}
    \EX{t}{ \frac{|\Delta_{t+1}|}{|U_t|} } \leq q(t)\cdot \varphi(I_t) \leq 1 - \epsilon,
\end{align*}
having used in the last step that $\varphi(I_t) \leq 1$. Finally, if we assume that all graphs of $(G_t)_{t \geq 0}$ are connected, then by \cref{lem:expushUT} $(ii)$, we have for any $t \geq 0$, 
\begin{equation*}
\EX{t}{ \frac{|\Delta_{t+1}|}{|U_t|} } \leq 1 - e^{- \varphi(I_t) \cdot q(t)  } \cdot 
     \left(1 - \frac{\varphi(I_t) \cdot (q(t))^2}{d_t} \right) \stackrel{(a)}{\leq} 1 - e^{- 1  } \cdot  \frac{1}{2},
\end{equation*}
where in $(a)$ we used $q(t) \leq 1$, $\varphi(I_t) \leq 1$ and the fact that if $(G_t)$ is connected and regular, then $d_t \geq 2$. This completes the proof of $(i)$.

\noindent\textit{Proof of \eqref{itm:PROTGROWTH_pull}}: For \pull,
\begin{align*}
   \EX{t}{\frac{|\Delta_{t+1}|}{|I_t|}}  \leq \EX{t}{\frac{|\Delta_{t+1}|}{\min(|I_t|, |U_t|)}} \stackrel{(a)}{=}  q(t) \cdot \varphi(I_t) \leq q(t),
\end{align*}
where $(a)$ holds by \cref{lem:PPP-P} \eqref{itm:Pull}. Furthermore, if $q(t) \leq 1 - \epsilon$, then by \cref{lem:PPP-P} \eqref{itm:Pull},
\begin{align*}
   \EX{t}{\frac{|\Delta_{t+1}|}{|I_t|}}  \leq \EX{t}{\frac{|\Delta_{t+1}|}{\min(|I_t|, |U_t|)}} = q(t) \cdot \varphi(I_t) \leq q(t) \leq 1 -\epsilon,
   \end{align*}
which completes the proof of $(ii)$.

\noindent\textit{Proof of \eqref{itm:PROTGROWTH_pp}}: Lastly, for \pp, for $|I_t| \leq n/2$ by \cref{lem:PPP-P} \eqref{itm:PP},
\begin{align*}
\EX{t}{\frac{|\Delta_{t+1}|}{|I_t|}} \leq 2 \cdot q(t) \cdot \varphi(I_t) \leq 2 \cdot q(t).
\end{align*}
Furthermore, if $q(t) \leq 1 - \epsilon$, then by \cref{lem:LBexPPT} $(ii)$,
\begin{align*}
    \EX{t}{\frac{|\Delta_{t+1}|}{|U_t|}} \leq 1 - \left( 1 - q(t) \right)^{\varphi(I_t)} \cdot \left(1 - q(t) \cdot \varphi(I_t) \right) \leq 1 - \epsilon^{\varphi(I_t)} \cdot \left( \epsilon \cdot \varphi(I_t) \right) \leq 1 - \epsilon^2,
\end{align*}
which completes the proof of $(iii)$.\end{proof}

\subsection{Other Examples}\label{sec:examples} We will briefly outline some other examples of $(\Cgrow,\Cshrink)$-spreading processes. We will not study these processes further in this paper, so for the sake of space we omit the proofs of membership.  
\begin{itemize}
    \item Variants of \push, \pull and \pp where vertices accept all incoming messages w.p. $q(t)$, independent of the number of messages received, otherwise reject all. This is an alternative interpretation of the \textit{credibility function} as being ``belief-based'', i.e.\ whenever a vertex receives at least one transmission (regardless of whether they are \push or \pull), it \emph{believes} in the rumor w.p. $q(t)$. Hence, the ``believed'' versions of \push, \pull and \pp are slower siblings of the ``transmission-based'' versions of \push, \pull and \pp as defined in \cref{sec:standard_rumor}.
    \item A variant of \push where \text{all} vertices transmit to a random neighbor in each step (uninformed vertices transmit an ``empty'' message, informed vertices transmit the rumor). Each uninformed vertex chooses at most one received message (chosen uniformly at random from all received messages, ignoring all others). If they receive a message with the rumor they are informed; otherwise they are not. This process was introduced by Daum, Kuhn and Maus \cite{DaumKM20}.
    \item The \emph{multiple call} model, where each vertex pushes the opinion to $k$ of random neighbors \cite{PPS15}, for constant $k$ (one could even consider $k$ to be dependent on the node as in \cite{PPS15}, or on the round $t$). This model can also be extended by using credibility functions.
    \item For any constant $\alpha\in [0,1]$, in each round $t\geq 0$, each node performs a pull with w.p. $\alpha$ and a push w.p. $1-\alpha$. This model can also support a credibility function.
    \item Variants of Broadcasting or Flooding models \cite{clementi2012information} where in each round each informed node sends the information to all its neighbors, however, edges may independently fail to transmit the message with some probability depending only on the edge.
\end{itemize}

\section{Lower Bounding the Number of Informed vertices}\label{sec:lower_bound}
 Our analysis will be split into two phases, a ``growing'' phase where $|I_t|\leq n/2$, and a ``shrinking'' phase where $|I_t|\geq n/2$.
\subsection{Growing phase: \texorpdfstring{$I_t \in [A,B]$}{Informed in [A,B]}} 
In this section, we prove a lower bound on the number of informed vertices after a stopping time $\tau_2$, which aggregates over the expected growth factors between round $1$ and $\tau_2-1$. In the following theorem (and throughout the rest of this paper) we use the convention that $\min\left\{\emptyset\right\} = \infty$. 

\begin{theorem}\label{thm:AzumaGrowing}
	Let $(G_t)_{t\geq 0}$ be any sequence of $d_t$-regular $n$-vertex graphs and consider a $\Cgrow$-growing process $\mathcal{P}$ with expected growth factors $\delta_t$. Let $t_1 \geq 0$ be any round, and let $A,B$ be thresholds satisfying $1 \leq A < B \leq n/2$ and $\xi:=10^{-30}$. Define the stopping time $\tau_2$ as
	\begin{equation}
		\tau_2 := \min \left( s \geq t_1 \colon \sum^{s-1}_{t=t_1} \log \left(1 + \delta_t \right) \geq \frac{ \log\left( \frac{B}{A} \right) + \left( \log\left( \frac{B}{A} \right) +  \log(1+\Cgrow) + 1 \right)^{2/3}}{\left(1-(1-\xi)\cdot |I_t|^{-\xi} \right)^2 }\right).\label{eq:prod_pre_grow}  
	\end{equation} 
	Then there is a constant $C_2 > 0$ such that
	\begin{equation*}
		\Pr{t_1}{ |I_{\tau_2}| < B  ~\Big|~ |I_{t_1}| \geq A } \leq \exp\left( - C_2 \cdot \left( \log\left( \frac{B}{A} \right) \right)^{1/3}  \right)+\Pr{t_1}{\tau_2 = \infty ~\Big|~ |I_{t_1}| \geq A}.
	\end{equation*}    
\end{theorem}

Recall that the growth factors $\delta_t$ are conditional expectations  given by  $\delta_t=\EX{t}{\frac{|\Delta_t|}{|I_t|}}$ in the growing phase, where $|I_t|\leq n/2$. Intuitively, the stopping time $\tau_2$ in \cref{cor:AzumaGrowing} can be viewed as a partial observer who does not know the sequence $I_t$, but only gets to know the expected growth factors in each round.  
\begin{remark}\label{rem:con} At first it might look challenging to apply \cref{thm:AzumaGrowing}, as one would need to control the probability that the stopping time is unbounded. However, in most applications we have a deterministic lower bound on the expected growth in each step and then, provided this bound is sufficiently large, this probability vanishes. See \cref{cor:AzumaGrowing} for a weaker but easier to apply variant of \cref{thm:AzumaGrowing} which leverages this idea. The use of this stopping time also allows \cref{thm:AzumaGrowing} to be very general. For instance, notice that $G_t$ is not required to always be connected; this gives flexibility when handling dynamic graphs. This works because if it is not possible to spread the rumor (say due to connectivity issues), then this implies $\tau_2$ is unbounded. 
\end{remark}

We will now give a brief overview of the proof of \cref{thm:AzumaGrowing}, followed by some helper lemmas and claims, and then complete the proof. The starting point is to analyze the growth rate of the number of informed vertices. To this end, we recall the following formula involving growth factors:
\begin{align}
 \frac{ |I_{\tau_2}| }{ |I_{t_1}|} &= \prod_{t=t_1}^{\tau_2-1} \frac{| I_{t+1}|}{|I_t| } = \prod_{t=t_1}^{\tau_2-1} \left(1 + \frac{|\Delta_{t+1}|}{|I_t|}\right). \label{eq:log_taking}
\end{align}
In order to transform this product into a sum of random variables, we first define for any $t\geq 0$, 
\begin{equation*}
     X_{t} := \log \left(1 + \frac{|\Delta_{t+1}|}{|I_t|}\right).
\end{equation*}
Then, by taking logarithms in \cref{eq:log_taking} we obtain that
\[
  \log\left( \frac{ |I_{\tau_2}| }{ |I_{t_1}|} \right) = \sum_{t=t_1}^{\tau_2-1} X_t.
\]
Our approach will be to lower bound the sum of these $X_t$'s. Therefore, we will consider the expected (logarithmic) growth in each step (i.e.\ $\EX{t}{X_t}$) (note that due to the dependence on $\mathcal{F}_t$ it is also a random variable). We then show that $\sum_{t=t_1}^{\tau_2-1}X_t$ is tightly concentrated around $\sum_{t=t_1}^{\tau_2-1}\EX{t}{X_t}$, using a variant of Azuma's concentration inequality (\cref{lem:azuma_variance}). In doing so, we face the following difficulty of relating the expectation of $X_t$ to the expected growth factor $\delta_t$. Specifically, we would like to apply the following approximation:
\begin{equation*}
    \EX{t}{ \log \left(1 + \frac{|\Delta_{t+1}|}{|I_t|}\right) } \approx \log\left(1 + \EX{t}{\frac{|\Delta_{t+1}|}{|I_t|}} \right) = \log\left(1 + \delta_t\right).
\end{equation*}
One direction in this approximation is immediate; since $\log(\cdot)$ is concave, Jensen's inequality gives us 
\[
	\EX{t}{ \log \left( 1 +  \frac{|\Delta_{t+1}|}{|I_t|} \right) } \leq \log \left(1 + \EX{t}{ \frac{|\Delta_{t+1}|}{|I_t|}} \right).
	\] 
It thus remains to bound the other direction, which amounts to proving an ``approximate reverse version'' of Jensen's inequality. This is fairly involved, but we manage to establish the following general lemma:

\begin{restatable}{lemma}{ReversedJensen}\label{lem:log_exp-JOHN}
For a fixed round $t\geq 0$, let $G_t$ be a regular $n$-vertex graph and consider a $\Cgrow$-growing process $\mathcal{P}$.  
If $|I_t| \in [A,n/2]$, then, for $\xi:=10^{-30}$, we have 
	\[
	\EX{t}{ \log \left( 1 +  \frac{|\Delta_{t+1}|}{|I_t|} \right) } \geq   \left(1-(1-\xi)\cdot |I_t|^{-\xi} \right)^2 \cdot \log \left( 1 + \EX{t}{ \frac{|\Delta_{t+1}|}{|I_t|}} \right).
	\] 
\end{restatable} 
Note that the first factor on the right-hand side of the inequality above is $(1-o(1))$ in the case $|I_t| = \omega(1)$ (i.e., a super-constant number of vertices are informed).  

\begin{proof}
For brevity, we will denote
 \[
 Z:= \frac{|\Delta_{t+1}|}{|I_t|}.
 \]
  We proceed by a case distinction dependent on the size of $\EX{t}{Z}$. Recall that as $\mathcal{P}$ is a $\Cgrow$-growing process, by \BEG and \PFour,
	\begin{align}
	 &\EX{t}{Z} =: \delta_t \leq \Cgrow, \notag\\      
      &\VAR{t}{Z} = \frac{1}{|I_t|^2} \cdot \VAR{t}{|\Delta_{t+1}|} \leq \frac{1}{|I_t|^2} \cdot \EX{t}{|\Delta_{t+1}|} \leq \frac{\delta_t}{|I_t|}. \label{eq:mom} \end{align}

	\medskip 
	
	\noindent \textbf{Case 1:} Let us assume $\delta_t  \geq 8 \cdot |I_t|^{-1/4}$. 
	We will first establish a concentration inequality, stating that $Z$ is not much smaller than its expectation.
	By Chebyshev's inequality,
	\begin{align}
		\Pr{t}{ \left|Z - \EX{t}{Z} \right| \geq \left(|I_t|\cdot \delta_t \right)^{-1/3} \cdot \EX{t}{Z} } &\leq \frac{\VAR{t}{Z}}{ \left(|I_t|\cdot \delta_t \right)^{-2/3}\cdot  (\EX{t}{Z})^2 } \notag \\
		&\stackrel{(a)}{\leq} \frac{  \delta_t /|I_t|}{ \left(|I_t|\cdot \delta_t \right)^{-2/3} \cdot  (\delta_t)^2} \notag \\ 	&\stackrel{}{=} \left(|I_t|\cdot \delta_t \right)^{-1/3},  \label{eq:cheby} 
	\end{align}
	where $(a)$ used the upper bound on the variance of $Z$, and the identity for $\EX{t}{Z}$ from \cref{eq:mom}. Observe that 
	\begin{equation}\label{eq:subbedin}  \left(|I_t|\cdot \delta_t \right)^{-1/3} \leq \left(|I_t|\cdot   8  |I_t|^{ -1/4}\right)^{-1/3}   = \tfrac{1}{2}\cdot |I_t|^{-1/4}, \end{equation}by the condition for Case 1. Therefore, as $Z \geq 0$ and $\log(1+Z) \geq 0$,
	\begin{align} 
		 \EX{t}{ \log ( 1 + Z ) }   & \geq \Pr{t}{ Z \geq \left(1-  \tfrac{1}{2}\cdot |I_t|^{-1/4}   \right) \cdot \EX{t}{Z} } \cdot \log\left(1 + \Bigl(1-  \tfrac{1}{2}\cdot |I_t|^{-1/4} \Bigr) \cdot \EX{t}{Z} \right) \notag \\
		& \stackrel{(a)}{\geq} \left(1-    \tfrac{1}{2}\cdot |I_t|^{-1/4}  \right) \cdot \left(1- \tfrac{1}{2}\cdot |I_t|^{-1/4}  \right) \cdot \log(1 + \EX{t}{Z}) \notag \\
		& \stackrel{}{=}  \left(1-  \tfrac{1}{2} \cdot |I_t|^{-1/4}  \right)^2 \cdot \log(1 + \EX{t}{Z}) \label{eq:case_one_conclusion}, 
	\end{align}
	where $(a)$ used \cref{eq:cheby}, \cref{eq:subbedin}, and \cref{clm:log_claim} with $z:= \EX{t}{Z} $ and $a:=\tfrac{1}{2} \cdot |I_t|^{-1/4} $. This completes Case 1.

	\medskip

\noindent\textbf{Case 2:} The other case is $\delta_t  \leq 8 \cdot |I_t|^{ -1/4}$.
	In this case $Z$ may not be concentrated (it was in Case 1) . However, $\EX{t}{Z}$ is so small that  $\EX{t}{\log(1+Z)}$ is close to $\log(1+\EX{t}{Z})$ as even ``outliers'' cannot contribute too much to $
	\log(1+Z)$.
Recall that $\EX{t}{|\Delta_{t+1}|} = \delta_t  \cdot |I_t|$. Thus, for any $\eta \geq 0$,
	\[
	\Pr{t}{ Z  \geq (1+\eta) \cdot \delta_t }
	= \Pr{t}{ |\Delta_{t+1}| \geq (1+\eta) \cdot \delta_t \cdot |I_{t}|  } = \Pr{t}{ |\Delta_{t+1}| \geq (1+\eta) \cdot \EX{t}{ \Delta_{t+1} } }.
	\] 
	Applying the general form of the Chernoff bound (\cref{lem:chernoff} statement (iv)),
	\begin{equation} \label{eq:chernoffZ}
			\Pr{t}{ Z  \geq (1+\eta) \cdot \delta_t } \leq \left( \frac{e}{1+\eta} \right)^{(1+\eta) \cdot \EX{t}{ |\Delta_{t+1}|}} = \left( \frac{e}{1+\eta} \right)^{(1+\eta)\cdot \delta_t \cdot |I_t|}.
	\end{equation}
	Set $\eta$ satisfying $1+\eta = \zeta \cdot ( | I_t | \cdot  \delta_t )^{-1}$ for any $\zeta \geq 128 |I_t|^{3/4}$, and observe that  
	\begin{equation}
		1+\eta \geq 128 |I_t|^{3/4} \cdot  ( | I_t | \cdot  \delta_t )^{-1} = 128 \cdot |I_t|^{-1/4} \cdot \delta_t^{-1} \geq 16 \geq e^{5/2}\label{eq:LBa},
	\end{equation}
Now returning to the bound from \cref{eq:chernoffZ},
	\begin{equation*}
		\Pr{t}{ Z \geq \zeta \cdot |I_t|^{-1}} \leq \left(\frac{e}{1+\eta}\right)^{\zeta}  \stackrel{(a)}{\leq} \left(\frac{1}{1+\eta}\right)^{3\zeta/5} =\left(\frac{1}{1+\eta}\right)^{ \zeta/10} \cdot \left(\frac{1}{1+\eta}\right)^{\zeta/2} \stackrel{(b)}{\leq} \frac{1}{1+\eta}  \cdot \left(\frac{1}{1+\eta}\right)^{\zeta/2}  , \notag 	\end{equation*}	where $(a)$ holds by \cref{eq:LBa} and $(b)$ since $\zeta/10 \geq 1$ and $1+\eta \geq 1$.  Inserting $1+\eta = \zeta \cdot (|I_t|  \cdot \delta_t)^{-1}$ gives
	\begin{equation}\label{eq:crucial_chernoff}
		\Pr{t}{ Z \geq \zeta \cdot |I_t|^{-1}}  
		\stackrel{}{\leq}   \frac{|I_t|  \cdot \delta_t }{\zeta  }   \cdot \left( \frac{|I_t|  \cdot \delta_t }{\zeta} \right)^{ \zeta/2}. 
	\end{equation} 
 We now define the set\begin{equation}\label{eq:defI}
	\mathcal{I} :=\left\{ \frac{0}{|I_t|},\frac{1  }{|I_t|},\frac{2  }{|I_t|},\ldots,   \frac{\lfloor 128|I_t|^{ 3/4}\rfloor}{|I_t|}   \right\}\subseteq \operatorname{Supp}(Z).
	\end{equation}
	Then, by the definition of expectation,
	\begin{equation}\label{eq:Iexp}
		\EX{t}{ \log (1+Z) } 
		= \sum_{  x \in \operatorname{Supp}(Z)} \Pr{t}{ Z = x} \cdot \log \left( 1+ x \right) 
		\geq \sum_{ x \in \mathcal{I}} \Pr{t}{ Z = x} \cdot \log \left( 1+ x \right),
	\end{equation}
 In the following, we would like to estimate $\log(1+x)$ by $x$, for each $x\in \mathcal{I}$. To that end, we do a case distinction: 
	\begin{itemize}
		\item \textbf{Case 2a: $x=0$}. In this case $\log(1+x)=\log(1)=0$ and $x=0$, so we have 
		\[
		\log(1+x)=x \cdot 1.
		\]
		\item \textbf{Case 2b: $ |I_t|^{ -1} \leq x \leq 128 |I_t|^{ -1/4} $}. For any $z \geq 0$, $\log(1+z) \geq z - z^2/2
		= z \cdot (1-z/2)$.
		Hence for any $x >  0$,  
		\[
		\log(1 + x) \geq x \cdot \left(1 - 64 |I_t|^{-1/4}  \right).
		\]
		\item \textbf{Case 2c: $ |I_t|^{ -1} \leq x \leq 128 |I_t|^{ -1/4} $ 
			and $128 |I_t|^{-1/4} \geq 1$}.
		Rearranging the second precondition implies  
		\begin{equation}\label{eq:case2c}
		x \geq |I_t|^{-1} \geq  128^{-4}.
		\end{equation} Now, we have
		\begin{align*}
			\log(1 + x) &=  x \cdot \frac{\log(1+x)}{x} \stackrel{(a)}{\geq} x \cdot \frac{\log(1 +  128 ^{-4})}{128 |I_t|^{-1/4}}\stackrel{(b)}{\geq} x \cdot \frac{\log(1 + 128 ^{-4})}{128} \stackrel{(c)}{:=} x \cdot C_1,
		\end{align*}
		where $(a)$ used the upper bound on $x$ from Case $2c$ in the denominator and \cref{eq:case2c} within the $\log(\cdot)$ in the numerator, and $(b)$ used that $|I_t|^{-1/4} \leq 1$, and $(c)$ defines the constant $C_1:=\frac{\log(1 + 128 ^{-4})}{128} \in (0,1/1000)$.
	\end{itemize}
	 
	Since when $128|I_t|^{-1/4}\leq 1$ holds, the estimate from Case 2b is tighter than the estimate we would get from Case 2c (if we were allowed to apply it), and the estimate in Case 2a is always tighter than the one from Case 2b and 2c, we conclude that for any $x\in \mathcal{I}$ 
	\[
	\log(1+x) \geq x \cdot \left(1 - \min \left( 64 |I_t|^{ -1/4}, 1 - C_1 \right) \right).
	\]
	Hence, by \cref{eq:Iexp} we have
	\begin{align}
 \EX{t}{ \log (1+Z) }  &\geq \sum_{x\in \mathcal{I} } \Pr{t}{ Z = x} \cdot x \cdot  \left(1 - \min \left( 64 |I_t|^{-1/4}, 1 - C_1 \right) \right) \notag \\
	 	&\geq \left(1 - \min \left( 64 |I_t|^{ -1/4}, 1 - C_1 \right) \right)  \cdot \left( \sum_{x\in \operatorname{Supp}(Z)} \Pr{t}{ Z = x} \cdot x - \sum_{x\in \operatorname{Supp}(Z) \setminus \mathcal{I}} \Pr{t}{ Z = x} \cdot x \right)\notag\\ 
		&=\left(1 - \min \left( 64 |I_t|^{-1/4}, 1 - C_1 \right)\right)  \cdot \left( \EX{t}{Z} - \sum_{x\in \operatorname{Supp}(Z) \setminus \mathcal{I}} \Pr{t}{ Z = x} \cdot x \right). \label{eq:first_exp_est}
	\end{align}
We proceed to bound the sum in the second factor on the right-hand side of \cref{eq:first_exp_est}, as follows
	\begin{align} 	  \sum_{x\in \operatorname{Supp}(Z) \setminus \mathcal{I}} \Pr{t}{ Z = x} \cdot x   
		&  \leq    \sum_{x\in \operatorname{Supp}(Z) \setminus \mathcal{I}} \Pr{t}{ Z \geq x} \cdot x \notag \\
	&  \stackrel{(a) }{=}    \sum_{\zeta=\lfloor 128 |I_t|^{3/4}\rfloor +1}^{\max(\operatorname{Supp}(Z) \setminus \mathcal{I} )\cdot |I_t|}  \Pr{t}{ Z \geq \frac{\zeta}{|I_t|}} \cdot \frac{\zeta}{|I_t|} \notag \\
	&  \stackrel{(b)}{\leq}     \sum_{\zeta=\lfloor 128 |I_t|^{3/4}\rfloor +1 }^{\infty}\frac{|I_t|  \cdot \delta_t }{\zeta  }   \cdot \left( \frac{|I_t|  \cdot \delta_t }{\zeta  } \right)^{ \zeta/2}  \cdot \frac{\zeta}{|I_t|} \notag \\
	&  \stackrel{}{=}   \delta_t\cdot   \sum_{\zeta=\lfloor 128 |I_t|^{3/4}\rfloor +1 }^{\infty} \left( \frac{|I_t|  \cdot \delta_t }{\zeta  } \right)^{ \zeta/2}   \notag \\
		& \stackrel{(c)}{\leq}   \delta_t\cdot   \sum_{\zeta=\lfloor 128 |I_t|^{3/4}\rfloor +1 }^{\infty} \left( \frac{|I_t|  \cdot 8|I_t|^{ -1/4} }{128 |I_t|^{3/4} } \right)^{ \zeta/2} \notag   \\
			& \stackrel{}{=}   \delta_t\cdot   \sum_{\zeta=\lfloor  128 |I_t|^{3/4}\rfloor +1 }^{\infty} \left( \frac{1}{16 } \right)^{ \zeta/2}  \notag  \\
		&\stackrel{(d)}{\leq } \delta_t \cdot  2^{-   128 |I_t|^{3/4} } \label{eq:second_exp_est}
\end{align}where $(a)$ holds by the definition of $\mathcal{I}$ from \cref{eq:defI}, $(b)$ holds by \cref{eq:crucial_chernoff}, $(c)$ holds since $\gamma \geq 128 |I_t|^{3/4}$, and $\delta_t \leq 8|I_t|^{ -1/4}$ as we are in Case 2, and finally $(d)$ holds by comparison with a geometric series. Recall that $\delta_t:=\EX{t}{Z}$, and so combining \cref{eq:first_exp_est} 	and \cref{eq:second_exp_est} gives
	\begin{align}
		\EX{t}{ \log (1+Z) } &\geq   \left(1 - \min \left( 64 |I_t|^{ -1/4}, 1 - C_1 \right) \right)  \cdot \left( 1 - 2^{-   128 |I_t|^{3/4} } 
		\right) 
		\cdot  \EX{t}{Z} \notag \\
		&\stackrel{(a)}{\geq}  \left(1 - \min \left( 64 |I_t|^{-1/4}, 1 - C_1 \right) \right)^2  \cdot  \EX{t}{Z} \notag\intertext{where $(a)$ holds since $2^{-   128 |I_t|^{3/4} }\leq 2^{-   128 }\leq 1-C_1$ as $C_1<1/1000$, and $2^{-   128 |I_t|^{3/4} } \leq 64|I_t|^{-1/4}$. Now,}
  \EX{t}{ \log (1+Z) }&\stackrel{(b)}{\geq} \left(1-(1-C_1^2)\cdot |I_t|^{-C_1^2} \right)^2 \stackrel{(c)}{\geq} \left(1-(1-10^{-30})\cdot |I_t|^{-10^{-30}} \right)^2 
		\label{eq:case_two_conclusion},
	\end{align}
	where  $(b)$ holds due to \cref{clm:tidyup} as $C_1\in (0,1/1000)$, and $(c)$ since $C_1:=\frac{\log(1 + 128 ^{-4})}{128} \geq 10^{-15}$.

	Now taking the worst-case lower bound over \cref{eq:case_one_conclusion} and \cref{eq:case_two_conclusion}, as well as noting that $|I_t| \geq A$ and $\EX{t}{Z} \geq \log(1+\EX{t}{Z})$, completes the proof.
\end{proof}

Lastly, before beginning the proof of \cref{thm:AzumaGrowing}, we first state the following helper claim.

\begin{restatable}{claim}{UBpreconGROWING}\label{claim:UBpreconGROWING}
For $\tau_2$ and $\xi:=10^{-30}$ as in \cref{thm:AzumaGrowing} and $1 \leq A \leq B \leq n/2$, the following holds,
\begin{equation*}
    \sum_{t=t_1}^{\tau_2-1} \delta_t \leq \frac{4}{\xi^2} \cdot \left( \log\left( \frac{B}{A} \right) +  \log(1 + \Cgrow) + 1 \right).
\end{equation*}
\end{restatable}

\begin{proof}
As the round $\tau_2$ was chosen to be minimal,  $A\geq 1$ and $\xi:=10^{-30}$, we have,
\begin{align*}
 \sum_{t=t_1}^{\tau_2-1} \log(1+\delta_t) &\leq \frac{ \log( \frac{B}{A} ) +\left( \log\left( \frac{B}{A} \right) + \log(1 + \Cgrow) + 1 \right)^{2/3} }{\left(1-(1-\xi)\cdot A^{-\xi} \right)^2} + \log(1+\delta_{\tau_2-1}) \notag\\
 &\leq \frac{ \log( \frac{B}{A} ) + \left( \log\left( \frac{B}{A} \right) + \log(1 + \Cgrow) + 1 \right)^{2/3} }{\xi^2} + \log(1 + \Cgrow)\\
 &\leq \frac{2}{\xi^2} \cdot \left( \log\left( \frac{B}{A} \right) +  \log(1 + \Cgrow) + 1 \right).
\end{align*}
Lastly, as $\log(1+z) \geq z/2$ for any $z \in [0,1]$ we get that
\begin{equation*}
    \sum_{t=t_1}^{\tau_2-1} \delta_t \leq \frac{4}{\xi^2} \cdot \left( \log\left( \frac{B}{A} \right) +  \log(1+\Cgrow) + 1 \right),
\end{equation*}as claimed.
\end{proof}
 
We are now ready to prove our lower bound on the informed set during the growing phase.   

\begin{proof}[Proof of \cref{thm:AzumaGrowing}]
Recall that, 
\begin{equation*}
     X_{t} := \log \left(1 + \frac{|\Delta_{t+1}|}{|I_t|}\right).
\end{equation*}
and if $|I_t| \leq n/2$  
\begin{equation*}
    \EX{t}{\frac{|\Delta_{t+1}|}{|I_t|}}:= \delta_t.
\end{equation*}
Moreover, let us define
    \begin{equation*}
        Y_t := \sum_{s=t_1}^{t-1} \left( X_s - \EX{s}{X_s} \right). 
    \end{equation*}
By construction, $(Y_t)_{t=t_1}^{\tau_2-1}$ is a zero-mean martingale with respect to $I_{t_1},I_{t_1+1},\ldots, I_{\tau_2-1}$. To apply concentration inequalities, we need to provide a bound ($M$) on $Y_t - Y_{t+1}$ when $|I_t|\leq n/2$. In this case, 
\begin{align}
   Y_t - Y_{t+1} &= \sum^{t-1}_{s=t_1}\left(X_s - \EX{s}{X_s}\right) - \sum^{t}_{s=t_1}\left(X_s - \EX{s}{X_s}\right) =  -\left(X_t - \EX{t}{X_t}\right).\nonumber\intertext{Now, using in $(a)$ that $X_t \geq 0$ deterministically, Jensen's inequality in $(b)$, and in $(c)$ the fact that $\mathcal{P}$ is a $\Cgrow$-growing process, we obtain}
  Y_t - Y_{t+1}   &\stackrel{(a)}{\leq} \EX{t}{\log\left(1 + \frac{|\Delta_{t+1}|}{|I_t|}\right)} \stackrel{(b)}{\leq} \log\left(1 + \EX{t}{\frac{|\Delta_{t+1}|}{|I_t|}}\right) \stackrel{(c)}{\leq} \log\left(1 + \Cgrow\right):= M \label{eqn:diffmartingaleY1}.
\end{align}

 We seek concentration for $Y_{\tau_2}$, however $\tau_2$ may be very large (even unbounded). Thus, we cannot use a standard version of Azuma's inequality, and we need to additionally consider the conditional variances, $\VAR{t}{X_t}$. To this end, we bound the variance for any round $t$ with $|I_t| \leq n/2$, by using \cref{lem:variance} in $(a)$, 
\[ 
\VAR{t}{ X_t }  = \VAR{t}{ \log \left(1 + \frac{|\Delta_{t+1}|}{|I_t|}\right)} \stackrel{(a)}{\leq} \VAR{t}{ \frac{|\Delta_{t+1}|}{|I_t|}}= \frac{1}{|I_t|^2} \cdot \VAR{t}{|\Delta_{t+1}|}.\] By Using \cref{lem:bounded_variance_property} and by recalling the definition $\delta_t = \EX{t}{\frac{|\Delta_{t+1}|}{|I_t|}}$, assuming $|I_t|\leq n/2$,  we get
\begin{equation}
	\VAR{t}{ X_t } \leq \frac{1}{|I_t|} \cdot \EX{t}{\frac{|\Delta_{t+1}|}{|I_t|}} =   \frac{1}{|I_t|} \cdot \delta_t. \label{eqn:vardiffmartingaleY1}
\end{equation} Note that by \cref{claim:UBpreconGROWING}, and using that $|I_t| \geq A$ for all $t \geq t_1$,
\begin{equation*}
    \sum_{t=t_1}^{\tau_2-1} \frac{1}{|I_t|} \cdot \delta_t \leq \frac{1}{A}\cdot \frac{4}{\xi^2}\left(\log\left(\frac{B}{A}\right)+\log\left(1 + \Cgrow\right) + 1 \right).
\end{equation*}
We are almost in a position to apply the concentration inequality (\cref{lem:azuma_variance}) to $Y_{\tau_2}$. The only slight tweak is that we will work with a martingale also stopped by $\tau := \min\{t \geq t_1 : |I_t| \geq n/2\}$, namely
    \begin{equation*}
        {\widehat Y}_{t}:= Y_{t \wedge \tau_2\wedge \tau},
    \end{equation*}
which is also a zero-mean martingale that satisfies \cref{eqn:diffmartingaleY1,eqn:vardiffmartingaleY1}. The reason for this is that the bound \eqref{eqn:vardiffmartingaleY1} assumes the inequality $|I_t|\leq n/2$ holds;  we loose nothing doing this because $|B|\leq n/2$.

Now, applying \cref{lem:azuma_variance} to $\widehat Y_t$ yields that for any $h>0$ and $t \geq t_1$
\begin{equation}
    \Pr{t}{\widehat Y_t< - h}< \exp\left(- \frac{h^2}{2 \cdot \left(\frac{1}{A}\cdot \frac{4}{\xi^2}\left(\log\left(\frac{B}{A}\right)+\log\left(1+\Cgrow\right) + 1 \right)  + \frac{1}{3}\left(\log\left(1 + \Cgrow\right)\cdot h\right)\right)}\right) . \label{eq:def_p}
\end{equation}
Let us set,
\begin{equation}\label{eq:defofh}
    h:=  \left(  \log\left( \frac{B}{A} \right) + \log(1 +\Cgrow) + 1 \right)^{2/3} 
  \geq 1. 
\end{equation}
Thus, for this $h$ and any round $t\geq t_1 $,
\begin{align}
    \Pr{t_1}{\hat{Y}_{t}<- h} &= \exp\left(-\frac{h^2}{2\left( \frac{4h^{3/2}}{A\cdot \xi^2}  + \frac{\log(1 + \Cgrow) h}{3}\right) }\right)\notag \\
    &\leq \exp\left(-\frac{h^2}{2\left( \frac{4h^{3/2}}{A\cdot \xi^2}  + \frac{\log(1 + \Cgrow) h^{3/2}}{3}\right) }\right) \notag\\
    &  =\exp\left(-C_2\cdot h^{1/2} \right),\label{eq:azumaupper}
\end{align}
where $C_2$ is given by $\left(\frac{8}{A \cdot \xi^2} + \frac{2}{3}\cdot \log\left(1 + \Cgrow\right)\right)^{-1}>0$. Observe that the right-hand side of \eqref{eq:azumaupper} is independent of $t$, this will be important later. However, at this point we must make the following claim:
\begin{equation}\label{eq:claim3.1}
\text{Conditional on $|I_{t_1}|\geq A$, we have } \{Y_{\tau_2 \wedge \tau} \geq - h\}\cap 	\{\tau_2\wedge\tau<\infty\} \subseteq  \{ |I_{\tau_2\wedge \tau }|\geq B\}\cap \{\tau_2\wedge\tau<\infty\}.
\end{equation} We prove this later, first we show how this, together with our earlier estimates, will establish the theorem. 

Returning to the proof, by \eqref{eq:azumaupper}, we have that for any integer $t\geq 0$,
\[ \Pr{t_1}{Y_{\tau_2 \wedge \tau}< - h,\;  \tau_2\wedge \tau\leq t} \leq \exp\left(-C_2\cdot h^{1/2} \right).\] 
Since the above bound holds for any integer $t\geq 0$, it follows that 
\begin{equation}\label{eq:probfinite} \Pr{t_1}{Y_{\tau_2 \wedge \tau}< - h,\;  \tau_2\wedge \tau <\infty} \leq \exp\left(-C_2\cdot h^{1/2} \right).\end{equation}
Observe  that  $|I_{\tau_2\wedge \tau }|\leq |I_{\tau_2}| $ by monotonicity (\Mono). Using this fact, then \eqref{eq:claim3.1}, and finally \eqref{eq:probfinite}, we have
\begin{align}
	\lefteqn{ \Pr{t_1}{ |I_{\tau_2}| < B  ~\Big|~ |I_{t_1}| \geq A }} \\ &\leq 	\Pr{t_1}{ |I_{\tau_2\wedge \tau}| < B  ~\Big|~ |I_{t_1}| \geq A } \notag \\
	&=   \Pr{t_1}{ |I_{\tau_2\wedge \tau}| < B,  \tau_2\wedge \tau <\infty ~\Big|~ |I_{t_1}| \geq A } + \Pr{t_1}{ |I_{\tau_2\wedge \tau}| < B, \tau_2\wedge\tau =\infty ~\Big|~ |I_{t_1}| \geq A }\notag \\
	&\leq   \Pr{t_1}{ Y_{\tau_2 \wedge \tau} < - h,  \tau_2\wedge \tau <\infty ~\Big|~ |I_{t_1}| \geq A } + \Pr{t_1}{  \tau_2\wedge\tau =\infty ~\Big|~ |I_{t_1}| \geq A }\nonumber\\
	&\leq \exp\left(-C_2\cdot h^{1/2} \right)+ \Pr{t_1}{  \tau_2  =\infty ~\Big|~ |I_{t_1}| \geq A }\notag, 
\end{align}
 which, recalling the definition \eqref{eq:defofh} of $h$, gives the bound in the statement. 
 
It remains to prove the claimed containment in \eqref{eq:claim3.1}. For that we analyze the behavior of $|I_{\tau_2\wedge \tau }|$ when the event $\{Y_{\tau_2 \wedge \tau} \geq - h\}\cap 	\{\tau_2\wedge\tau<\infty\}$ holds. We will split into two cases. 

\medskip 

\noindent \textbf{In the first case} $\{Y_{\tau_2 \wedge \tau} \geq - h\}\cap 	\{\tau<\infty, \tau \leq  \tau_2\}$. Hence, $|I_{\tau_2\wedge \tau }| = |I_{\tau}| \geq n/2 \geq B$. 

\medskip
 
\noindent \textbf{In the second case} $\{Y_{\tau_2 \wedge \tau} \geq - h\}\cap 	\{\tau_2<\infty, \tau_2 <  \tau\}$. Thus, $Y_{\tau_2\wedge \tau}  = Y_{\tau_2} $, and deterministically we have,
\[
Y_{\tau_2} = \sum_{t=t_1}^{\tau_2-1} \left(X_t - \EX{t}{X_t} \right) \geq - h.
\]
Rearranging this, we get that,
\begin{align*}
 \sum_{t=t_1}^{\tau_2-1} X_t &\geq \sum_{t=t_1}^{\tau_2-1} \EX{t}{X_t} - h = \sum^{\tau_2 -1}_{t=t_1}\EX{t}{\log\left(1 + \frac{|\Delta_{t+1}|}{|I_t|}\right)} - h \geq \gamma \cdot \sum_{t=t_1}^{\tau_2-1} \log \left( 1 + \EX{t}{\frac{|\Delta_{t+1}|}{|I_t|}} \right) - h,
\end{align*}
where the last inequality follows from \cref{lem:log_exp-JOHN}, and $\gamma:=\left(1-(1-\xi)\cdot |I_t|^{-\xi} \right)^2 $ for $\xi:=10^{-30}$. Since $\EX{t}{\frac{|\Delta_{t+1}|}{|I_t|}} = \delta_t$ for rounds $t$ with $|I_t| \leq n/2$, we conclude that
\begin{align*}
\sum_{t=t_1}^{\tau_2-1} X_t &\geq \gamma \cdot \sum_{t=t_1}^{\tau_2-1} \log \left( 1 + \delta_t \right) - h = \gamma \cdot \sum_{t=t_1}^{\tau_2-1} \log \left( 1 + \delta_t \right) -  \left( \log\left( \frac{B}{A}\right) +  \log(1+ \Cgrow) + 1 \right)^{2/3}. 
\end{align*}
Finally, by using that 
\begin{align*}
 \sum_{t=t_1}^{\tau_2-1} \log\left( 1 + \delta_t  \right) \geq \frac{ \log( \frac{B}{A} ) +  \left( \log\left( \frac{B}{A} \right)  +  \log(1 + \Cgrow) + 1 \right)^{2/3}}{\gamma},
\end{align*}
we conclude  that  $\sum_{t=t_1}^{\tau_2-1} X_t \geq \log(\frac{B}{A})$, i.e.\ that $|I_{\tau_2}|-|I_{t_1}|\geq B-A$, and thus $|I_{\tau_2}|\geq B$. \end{proof}

\subsection{Shrinking phase: \texorpdfstring{$|U_t| \in [C,D]$}{Uninformed in [C,D]}}\label{sec:shrinking}
In this section we consider the shrinking of the number of informed vertices. We prove an upper bound on the number of uninformed vertices after a stopping time $\tau_3 \geq t_2$, which now aggregates over the expected shrinking factors between round $t_2$ and $\tau_3-1$. 

\begin{restatable}{theorem}{shrinkingphase}
\label{thm:shrinkingphase}
Let $(G_t)_{t\geq 0}$ be any sequence of $d_t$-regular $n$-vertex graphs and consider a $\Cshrink$-shrinking process $\mathcal{P}$ with expected shrinking factors $\delta_t$. Let $C,D$ be thresholds satisfying $n/2 \geq C\geq D\geq \frac{3}{4}$ and $t_2 \geq 0$ be a round such that $|U_{t_2}|\leq C$. We define a stopping time $\tau_3 \in \mathbb{N}$ as
\begin{equation}
   \tau_3 := \min \left\{ s \geq t_2 \colon \sum^{\tau_3 -1}_{t=t_2}\log\left(1 -\delta_t\right) \leq  - \frac{1}{\gamma} \cdot \left(\log\left(\frac{C}{D}\right) + \left(\log\left(\frac{C}{D}\right) - \log\left(1 - \Cshrink\right) + 1\right)^{2/3}\right)\right\}, \label{eq:product_precondition_LB_phaseshrinking}
\end{equation}
where 
\begin{equation*}
\gamma := \left(1 - \min\left( \frac{1}{2 (1-\Cshrink) \cdot D} ,  \frac{1}{2} \right) \right).
\end{equation*}Then there is a constant $C_2 > 0$ such that
\begin{equation*}
    \Pr{t_2}{ |U_{\tau_3}| > D  ~\Big|~ |U_{t_2}| \leq C } \leq \exp\left(- C_2 \cdot \left(\log\left(\frac{C}{D}\right)\right)^{1/3}\right) + \Pr{t_2}{\tau_3 =\infty  ~\Big|~ |U_{t_2}| \leq C}.
\end{equation*}
\end{restatable}
 The proof of \cref{thm:shrinkingphase} follows a similar flow as the proof of \cref{thm:AzumaGrowing}. We note that, just like in the growing phase,
\begin{equation*}
    \frac{|U_{\tau_3}|}{|U_{t_2}|} = \prod^{\tau_3 - 1}_{t=t_2}\frac{|U_{t+1}|}{|U_{t_2}|}= \prod^{\tau_3 -1}_{t=t_2}\left(1 - \frac{|\Delta_{t+1}|}{|U_t|}\right).
\end{equation*}However, unlike in the growing phase, we cannot take the logarithm here, as both the numerator $|U_{t+1}|$ and the denominator $|U_t|$ may be zero. To avoid this, we define a different sequence with an artificial offset $\alpha$ should $|U_t|$ become zero:
\[
 \tilde{U}_{t} := \max\left(|U_{t}|, \alpha\right),
\]
where we choose $\alpha:=\frac{1}{2}$. We also define
\begin{equation}\label{eq:deltatilde}
 \tilde{\Delta}_{t+1} := \min\left( |\Delta_{t+1}|, |U_t| - \alpha \right).   
\end{equation}
Similarly to the growing phase, we will now introduce the logarithmic shrinking factor, defined as
    \begin{equation}\label{eq:X_t}
        X_t:= \log \left(  1 - \frac{\tilde{\Delta}_{t+1}}{|U_t|} \right)\leq 0.
    \end{equation}
and set,
\[
 \widehat{U}_{t} := \tilde{U}_{t \wedge \tau},
\]
where $\tau:= \min \{ t \geq t_2 \colon |U_t| \leq D \}$.  Now taking the logarithm yields, as long as $\tau_3 < \tau$,
\begin{align}
 \log \left( \frac{ \widehat{U}_{\tau_3} }{ U_{t_2} }  \right) 
 =  \sum_{t=t_2}^{\tau_3-1} \log \left(  1 - \frac{\tilde{\Delta}_{t+1}}{|U_t|} \right) = \sum_{t=t_2}^{\tau_3-1} X_t \label{eq:log_equation}
\end{align}
Recall that $|U_{t_2}| \leq C$, hence rearranging \cref{eq:log_equation} yields
\[
 \log\left( \widehat{U}_{\tau_3} \right) = \log(|U_{t_2}|) + \sum_{t=t_2}^{\tau_3-1} X_t 
 \leq \log(C) + \sum_{t=t_2}^{\tau_3-1} X_t.
\]
Furthermore if $\sum_{t=t_2}^{\tau_3-1} X_t < -\log(C) + \log(D)$, then it follows that
\[
 \widehat{U}_{\tau_3} \leq D  \quad \Rightarrow \quad |U_{\tau_3}| \leq \widehat{U}_{\tau_3}  \leq D.
\]
Hence our goal is to prove that
\begin{align}
 \sum_{t=t_2}^{\tau_3-1} X_t = \sum_{t=t_2}^{\tau_3-1} \log\left( 1 -  \frac{\tilde{\Delta}_{t+1}}{|U_t|} \right) \leq \log\left( \frac{D}{C} \right), \label{eq:goal}
\end{align}
with $D < C$. Note that in the last expression, both sides are negative.

Before we can start with the proof of \cref{thm:shrinkingphase}, we will state one crucial lemma that lower bounds the expectation and upper bounds the variance of the $X_t$, and subsequently, additional helper results.

\begin{restatable}{lemma}{pseudocapping}
\label{lem:pseudo_capping}
Consider a regular graph $G_t$ and a $\Cshrink$-shrinking process $\mathcal{P}$ with expected shrinking factors $\delta_t$ and round $t\geq 0$ with $1 \leq |U_t| \leq \frac{n}{2}$. Recall that $X_t= \log\left(1 - \frac{\tilde{\Delta}_{t+1}}{|U_t|} \right)$. Then, there is a constant $\kappa < 0$ such that,
\[
 \EX{t}{ X_t} \geq \kappa.
\]
Furthermore, there is a constant $\nu > 0$ such that, 
\[
\VAR{t}{ X_t} \leq \nu \cdot \delta_t = \nu \cdot \EX{t}{ \frac{|\Delta_{t+1}|}{|U_t|} }
\]
\end{restatable}
\begin{proof}[Proof of \cref{lem:pseudo_capping}]
Recall that $\alpha=1/2$.
In order to lower bound the expectation, note that for any $\chi \leq 0$, 
\begin{align}
    \EX{t}{ X_t } 
    &\stackrel{(a)}{\geq} \Pr{t}{ X_t > \chi } \cdot \chi + \Pr{t}{ X_t \leq \chi } \cdot \log\left( \frac{\alpha}{|U_t|} \right) \notag \\
    &\geq 1\cdot \chi + \Pr{t}{ X_t \leq \chi } \cdot \log\left( \frac{\alpha}{|U_t|} \right), \label{eq:exp_bound}
\end{align}
where $(a)$ holds since $\log\left( \frac{\alpha}{|U_t|} \right)
$ is the smallest (i.e., most negative) value the random variable 
$X_t$ can attain, and in the last inequality we use that $\chi \leq 0$. Therefore, it remains to find a suitable constant $\chi \leq 0$ such that $\Pr{t}{X_t \leq \chi}$ is of the order $-1/\left(\log\left( \frac{\alpha}{|U_t|} \right)\right)$ for any $|U_t| \geq 1$. Now, for any $\chi \leq  0$ we have,
\begin{align*}
    \Pr{t}{ X_t \leq \chi} &= \Pr{t}{ \log\left( 1 - \frac{ \widetilde{\Delta}_{t+1} }{|U_t|} \right)  \leq \chi } 
        = \Pr{t}{ 1 - \frac{ \widetilde{\Delta}_{t+1} }{|U_t|} \leq e^{\chi} } 
              = \Pr{t}{   \widetilde{\Delta}_{t+1}  \geq (1-e^{\chi}) \cdot |U_t| }.
\end{align*}
Further, due to the capping, $\EX{t}{\tilde{\Delta}_{t+1}} \leq \EX{t}{|\Delta_{t+1}|}=\delta_t \cdot |U_t|$. We will assume (with foresight) that from now on that the constant $\chi$  satisfies 

\begin{equation}\label{eq:thecforme}
    \chi\leq \min \left(2 \log(1 - \Cshrink), -1\right). 
\end{equation} To show that such a $\chi$ satisfies the desired properties above, we continue with a case distinction:
\begin{enumerate}
    \item[] \textbf{Case 1:} $\delta_t \leq \left(  \frac{1-e^\chi}{e} \right)^2$.
Choosing $\eta:= \frac{1-e^\chi}{\delta_t}-1$,  
the general Chernoff bound (\cref{lem:chernoff} (iv)) implies
\begin{align*}
     \Pro{  \tilde{\Delta}_{t+1} \geq (1-e^{\chi}) \cdot |U_t| } &\leq \Pro{  |\Delta_{t+1}| \geq  (1+\eta) \cdot \EX{t}{|\Delta_{t+1}|} } \\
    &\leq \left( \frac{e}{1+\eta} \right)^{(1+\eta) \cdot \EX{t}{ |\Delta_{t+1}|}} \\
    &\leq \left( \delta_t \cdot \frac{e}{1-e^{\chi}} \right)^{ (1-e^{-\chi}) \cdot |U_t| } \\
    &\stackrel{(a)}{\leq} \left( \delta_t \right)^{\frac{1}{2} \cdot (1-e^{-\chi}) \cdot |U_t| } \\
    &\stackrel{(b)}{\leq} \left( \delta_t \right)^{\frac{1}{4} \cdot |U_t| },
\end{align*}
where $(a)$ used that by the condition of Case $1$,
$\frac{e}{1-e^{\chi}} \leq (\delta_t)^{-1/2}$,
and $(b)$ used that for $\chi \leq -1$, $1-e^{\chi} \geq 1/2$. 
\item[] \textbf{Case 2:} $\delta_t \geq\left( \frac{1-e^\chi}{e} \right)^2 $.
We choose again $\eta := \frac{1-e^\chi}{\delta_t} -1$. Since $|U_t| \leq n/2$ by assumption, and $\mathcal{P}$ is a $\Cshrink$-shrinking process, by \cref{def:processes} (\BEGT), $\delta_t \leq \Cshrink$. Moreover, for $\chi \leq  2\log(1 - \Cshrink)$ by \cref{eq:thecforme}, we have
\[
 \frac{1-e^{\chi}}{\delta_t} \geq \frac{2\Cshrink - \Cshrink^2}{\Cshrink} = 2-\Cshrink,
\]
and consequently, $\eta\geq 1 - \Cshrink$. 

Here, by another version of the Chernoff bound (\cref{lem:chernoff}, $(iii)$).
\begin{align}
    \Pro{  \tilde{\Delta}_{t+1} \geq (1-e^{\chi}) \cdot |U_t| }
    &\leq \Pro{  |\Delta_{t+1}| \geq  (1+\eta) \cdot \EX{t}{|\Delta_{t+1}|} } \notag \\ &\leq \exp\left(- \frac{\eta^2}{2 + \eta} \cdot \EX{t}{|\Delta_{t+1}|}\right) \notag \\
    &\stackrel{(a)}{\leq} \exp\left(- \frac{\eta^2}{3\eta/(1 - \Cshrink)} \cdot \EX{t}{|\Delta_{t+1}|}\right) \notag \\
    &\stackrel{(b)}{\leq} \exp\left(- \frac{1 - \Cshrink}{3/(1 -\Cshrink)} \cdot \delta_t \cdot |U_{t}| \right) \notag \\
    &\stackrel{(c)}{=} \exp\left(- \frac{\left(1 - \Cshrink\right)^2}{3} \cdot  \delta_t \cdot |U_{t}| \right) \notag \\
    &= \left( \delta_t \right)^{ -1/\log( \delta_t) \cdot  \frac{\left(1 - \Cshrink\right)^2}{3} \cdot \delta_t \cdot |U_{t}| }, \label{eq:resume_chernoff}
    \end{align}
    where $(a)$ used that since $\eta \geq 1 - \Cshrink$ (and $1 - \Cshrink < 1$), we have $2+\eta \leq 2 \eta/(1-\Cshrink) + \eta \leq 3 \eta/(1-\Cshrink)$. In $(b)$ we used again that $\eta \geq 1 - \Cshrink$ as well as the fact that $\EX{t}{|\Delta_{t+1}|} =\delta_t \cdot |U_t|$, since we assumed that $|U_t|\leq n/2$.
    Next, in order to lower bound $\frac{\delta_t}{\log(\delta_t)}$, define the function 
    \[
  f(z) := \frac{z}{\log(z)},
    \]
    for the range $z \in [ \left( \frac{1-e^{\chi}}{e} \right)^2,\Cshrink]$. This function is non-decreasing in $z$, and therefore is minimized for $z=\left(\frac{1-e^{\chi}}{e}\right)^2 \geq \left(\frac{1-e^{-1}}{e}\right)^2 \geq  e^{-4}$, having used $\chi \leq -1$. Applying this to \cref{eq:resume_chernoff},
    \begin{align*}
        \Pro{  \tilde{\Delta}_{t+1} \geq (1-e^{\chi}) \cdot |U_t| }
    &\leq \left( \delta_t \right)^{ -\frac{1}{\log \left( e^{-4} \right)} \cdot  \frac{\left(1 - \Cshrink\right)^2}{3} \cdot e^{-4} \cdot |U_{t}| } = \left( \delta_t \right)^{ \frac{\left(1 - \Cshrink\right)^2}{12 e^{4}} \cdot |U_{t}| }.
    \end{align*}
\end{enumerate}

\noindent Combining the two cases, it follows that there is a constant $c_1 = c_1(\Cshrink) := \frac{\left(1 - \Cshrink\right)^2}{12 e^{4}} > 0$,
\begin{align}
    \Pro{ X_t \leq \chi} &= \Pro{  \tilde{\Delta}_{t+1}  \geq (1-e^{\chi}) \cdot |U_t| }  \leq \left( \delta_t \right)^{c_1 \cdot |U_t|}. \label{eq:upper_tail}
\end{align}
 
\noindent Therefore, returning to the crucial term in \cref{eq:exp_bound},
\begin{align}
\Pr{t}{ X_t \leq \chi} \cdot \log\left( \frac{\alpha}{|U_t|} \right) 
&\geq \left( \delta_t \right)^{c_1 \cdot |U_t|} \cdot \log\left( \frac{\alpha}{|U_t|} \right) \notag \\
&= - \left( \delta_t \right)^{c_1 \cdot |U_t|} \cdot \log\left( \frac{|U_t|}{\alpha} \right) \notag \\
&\stackrel{(a)}{\geq} - \left( \delta_t \right)^{c_1 \cdot |U_t|/2} 
\cdot \left(\Cshrink\right)^{c_1\cdot |U_t|/2 } \cdot \log\left( 2 \cdot |U_t| \right)  \notag \\
&\stackrel{}{=} - \left(\delta_t\right)^{c_1\cdot |U_t|/2} \cdot \exp\left( \log(\Cshrink) c_1\cdot |U_t|/2 + \log \log \left(2 \cdot |U_t|\right)  \right)\notag \\
&\stackrel{(b)}{\geq} -\left( 1 \right)^{c_1 \cdot |U_t|/2}  \cdot c_2 = -c_2, \label{eq:paramount}
\end{align}
where $(a)$ holds as $\delta_t \leq \Cshrink$, $\alpha=1/2$, and $(b)$ as $\delta_t \leq 1$ and $\exp\left(- \log\left( \frac{1}{\Cshrink} \right) c_1\cdot |U_t|/2 + \log \log \left(2 \cdot |U_t|\right)  \right)$ can be upper bounded by some constant $c_2 > 0$. Using this in \cref{eq:exp_bound} yields the first statement,
\begin{align*}
     \EX{t}{ X_t } 
    &\geq \chi - c_2 =: \kappa.
\end{align*}

For the upper bound on the variance, recall that the definition $\tilde{\Delta}_{t+1}:= \min\left( |\Delta_t|, |U_t| - \alpha \right)$. Let $X_1,X_2 \sim \tilde{\Delta}_{t+1}$ be two independent random variables with the same distribution as $\tilde{\Delta}_{t+1}$. Then,
\begin{align*}
     \lefteqn{ \VAR{t}{\log\left( 1 - \frac{\tilde{\Delta}_{t+1}}{|U_t|}\right)}} \\
      &\stackrel{(a)}{=} \frac{1}{2} \cdot \Ex{ \left( \log\left( 1 - \frac{X_1}{|U_t|}\right)-\log\left( 1 - \frac{X_2}{|U_t|}\right) \right)^2 } \\
      &\stackrel{(b)}{=} \frac{1}{2} \cdot \sum_{k_1,k_2 \in \mathrm{Supp}(\tilde{\Delta}_{t+1}) } \Pr{t}{ \tilde{\Delta}_{t+1} = k_1} \cdot \Pr{t}{ \tilde{\Delta}_{t+1} = k_2 } \cdot \left( \log\left(1 - \frac{k_1}{|U_t|}\right) - \log\left(1 - \frac{k_2}{|U_t|}\right) \right)^2 \\
      &\leq \underbrace { \frac{1}{2}  \left(\cdot \sum_{0 \leq k_1, k_2 \leq \min\left((1-e^{\chi}) \cdot |U_t|,|U_t|-\alpha\right)} \Pr{t}{ |\Delta_{t+1}| = k_1} \cdot \Pr{t}{ |\Delta_{t+1}| = k_2 } \cdot \left( \log\left(1 - \frac{k_1}{|U_t|}\right) - \log\left(1 - \frac{k_2}{|U_t|} \right) \right)^2\right) }_{:=A} \\
     & \quad \underbrace{ + 1 \cdot \Pr{t}{ \tilde{\Delta}_{t+1} \geq \min\left((1-e^{\chi}) \cdot |U_t|,|U_t|-\alpha\right) } \cdot \left( \log\left( \frac{|U_t|}{\alpha} \right) \right)^2 }_{:=B}, 
\end{align*}
where $(a)$ follows from \cref{lem:variance_basic},
and $(b)$ follows from the definition of expectation, and the fact that $X_1,X_2 \sim \tilde{\Delta}_{t+1}$ are independent.
We will now first bound $A$, and then $B$.
Regarding $A$,
Since $\log(.)$ is concave, we have for any $k_1 \geq k_2$,
\[
 \log\left(1 - \frac{k_2}{|U_t|} \right) \leq \log\left(1 - \frac{k_1}{|U_t|} \right) + \left( \frac{k_1-k_2}{|U_t|} \right) \cdot \frac{1}{1-\frac{k_1}{|U_t|}} \leq \log\left(1 - \frac{k_1}{|U_t|} \right) + \left( \frac{k_1-k_2}{|U_t|} \right) \cdot \frac{1}{e^\chi},
\]
where the last inequality used that $k_1 \leq (1-e^{\chi}) \cdot |U_t| $. Therefore, 
\begin{align*}
    A
    &\leq \frac{1}{2 \cdot e^{2\chi}} 
    \cdot \sum_{0 \leq k_1, k_2 \leq \min\left((1-e^{\chi}) \cdot |U_t|,|U_t|-\alpha\right)} 
      \Pr{t}{ |\Delta_{t+1}| = k_1} \cdot \Pr{t}{ |\Delta_{t+1}| = k_2 } \cdot \left( \frac{k_1}{|U_t|} - \frac{k_2}{|U_t|} \right)^2 \\
      &\leq \frac{1}{2 \cdot e^{2\chi}} \cdot 
      \sum_{0 \leq k_1,k_2 \leq |U_t|}  \Pr{t}{ |\Delta_{t+1}| = k_1} \cdot \Pr{t}{ |\Delta_{t+1}| = k_2 } \cdot \left( \frac{k_1}{|U_t|} - \frac{k_2}{|U_t|} \right)^2. \intertext{For the next step we use \cref{lem:variance_basic} in $(a)$ and the bounded variance property (\cref{lem:bounded_variance_property}) in $(b)$, giving}
     A &\leq \frac{1}{2 \cdot e^{2\chi}} \cdot \EX{t}{\frac{\left(X_1 - X_2\right)^2}{|U_t|}}
       \stackrel{(a)}{=} \frac{1}{ e^{2\chi}} \cdot \VAR{t}{ \frac{|\Delta_{t+1}|}{|U_t|} } \stackrel{(b)}{\leq} \frac{1}{ e^{2\chi}} \cdot \frac{1}{|U_{t}|} \cdot \EX{t}{|\Delta_{t+1}|}= \frac{1}{ e^{2\chi}} \cdot \delta_t.
\end{align*}
In order to upper bound $B$, we perform a case distinction. 

\textbf{First}, if $(1-e^{\chi}) \cdot |U_t| < |U_t| - \alpha$ (which by rearranging, yields $|U_t| > \frac{1}{2 e^{\chi}}$), then
  \cref{eq:upper_tail} gives
\begin{align}
    B &= \Pr{t}{\tilde{\Delta}_{t+1}  \geq (1-e^\chi) \cdot |U_t|} \cdot \left( \log\left( \frac{|U_t|}{\alpha} \right) \right)^2 \notag \\ 
    &\leq \left( \delta_t \right)^{c_1 \cdot |U_t|/2} 
\cdot \left( \delta_t \right)^{c_1 \cdot |U_t|/2} \cdot \left( \log\left(2 \cdot |U_t| \right) \right)^2 \notag \\
 &\stackrel{(a)}{\leq} \delta_t\cdot \left( \delta_t \right)^{c_1 \cdot |U_t|/2} \cdot \left( \log\left(2 \cdot |U_t| \right) \right)^2 \notag \\
&\stackrel{(b)}{\leq} \delta_t \cdot \left(\Cshrink \right)^{c_1 \cdot |U_t|/2 } \cdot \left( \log\left( 2 |U_t| \right) \right)^2 \notag \\
&= \delta_t \cdot \exp\left(\log\left(\Cshrink\right)\cdot c_1/2\cdot |U_t| + \log \log \left(2 \cdot |U_t|\right)\right) \notag \\
&\stackrel{(c)}{\leq} \delta_t \cdot c_3,\label{eq:eqdefofc3}
\end{align} 
where in $(a)$ we used that $c_1\cdot |U_t|/2\geq \frac{c_1}{4e^c} \geq 1  $ as  $c_1 := \frac{\left(1 - \Cshrink\right)^2}{12 e^{4}}$ , in $(b)$ that $\delta_t \leq \Cshrink$ since $\mathcal{P}$ is a $\Cshrink$-shrinking process and $|U_t| \leq n/2$. Lastly, in  $(c)$ we used that $\exp\left(- \log\left( \frac{1}{\Cshrink} \right) c_1/2\cdot |U_t| + \log \log \left(2 \cdot |U_t|\right)  \right)$ can be upper bounded by some constant $c_3 > 0$.

\textbf{Secondly}, if $(1-e^{\chi}) \cdot |U_t| \geq |U_t| - \alpha = |U_t| - \frac{1}{2}$ holds (i.e., $|U_t| \leq \frac{1}{2 e^{\chi}}$), then by Markov's inequality.
\begin{align*}
    \Pro{  \tilde{\Delta}_{t+1} \geq |U_t| - \alpha }
    &= \Pro{  \Delta_{t+1} \geq |U_t| } \leq \frac{ \Ex{ \Delta_{t+1} } }{|U_t|} = \delta_t,
\end{align*}
where the last equality holds since we assume that $|U_t|\leq n/2$ . Therefore, in the second case,
\begin{align*}
    B \leq \delta_t \cdot \left( \log ( 2 \cdot |U_t|) \right)^2 \leq \delta_t \cdot \frac{1}{\chi^2}.
\end{align*}
Combining our bounds on $A+B$ we finally conclude,
\begin{align*}
    \VAR{t}{\log\left( 1 - \frac{\tilde{\Delta}_{t+1}}{|U_t|}\right)} &\leq A+B \leq \frac{1}{e^{2\chi}} \cdot \frac{1}{|I_t|} \cdot \delta_t + 2 \cdot \delta_t \cdot \max\left\{ c_3, \frac{1}{\chi^2} \right\},
\end{align*}
for the constant $c_3 >0$ given implicitly  by \cref{eq:eqdefofc3} and any constant $\chi>0$ satisfying \cref{eq:thecforme}, for $\nu = \frac{1}{e^{2\chi}} + 2 \cdot \max \left\{c_3, \frac{1}{\chi^2}\right\}$, the second bound in the statement follows. 
\end{proof}

\begin{restatable}{claim}{UBpreconSHRINKING}\label{claim:UBpreconSHRINKING}
    Consider the setting of \cref{thm:shrinkingphase}. Let $\tau_3$ be as in \cref{thm:shrinkingphase} and $n/2 \geq C\geq D \geq \frac{3}{4}$. Then, the following holds,
    \begin{align*}
        \sum^{\tau_3-1}_{t=t_2} \delta_t \leq
    4 \cdot \log\left(\frac{C}{D}\right)  - 3 \cdot \log(1-\Cshrink) +2
    \end{align*}
 
\end{restatable}
\begin{proof}
    As the round $\tau_3$ was chosen to be minimal in \cref{eq:product_precondition_LB_phaseshrinking}
    \begin{align}
        \sum^{\tau_3-1}_{t=t_2}\log\left(1 - \delta_{t}\right) &\geq -\frac{1}{\gamma}\cdot \left(\log\left(\frac{C}{D}\right) + \left(\log\left(\frac{C}{D}\right) - \log\left(1 - \Cshrink\right) + 1\right)^{2/3}\right) + \log\left(1 - \delta_{\tau_3-1}\right) \notag\\
        &\stackrel{(a)}{\geq} -\frac{1}{\gamma} \cdot \left(\log\left(\frac{C}{D}\right) + \left(\log\left(\frac{C}{D}\right) - \log\left(1 - \Cshrink\right) + 1\right)^{2/3}\right) +  \log\left(1 - \Cshrink\right) \label{eq:RHSshrink}
    \end{align}
    where $(a)$ follows from $\mathcal{P}$ being a $\Cshrink$-shrinking process and $|U_t|\leq C \leq n/2$, and thus by (\BEGT), $\delta_{\tau_3 -1} \leq \Cshrink$. We note that the right-hand side of \cref{eq:RHSshrink} is minimized (i.e., the absolute value is maximized) if we use the bound $\gamma \geq 1/2$ (which is the minimum value $\gamma$ can take), then we conclude that
    \begin{align*}
    \sum^{\tau_3-1}_{t=t_2}\log\left(1 - \delta_{t}\right) &\geq -2 \cdot \left(\log\left(\frac{C}{D}\right) + \left(\log\left(\frac{C}{D}\right) - \log\left(1 - \Cshrink\right) + 1\right)^{2/3}\right) + \log\left(1 - \Cshrink\right)\\
    &\stackrel{(a)}{\geq} -2 \cdot \left(\log\left(\frac{C}{D}\right) + \left(\log\left(\frac{C}{D}\right) - \log\left(1 - \Cshrink\right) + 1\right) \right) + \log\left(1 - \Cshrink\right) \\
    &= -4 \cdot \log\left(\frac{C}{D}\right) + 3\cdot \log(1-\Cshrink) -2,
    \end{align*}
    where $(a)$ used that $\log\left(\frac{C}{D}\right) > 0$ and $-\log(1- \Cshrink) > 0$.  As $-x \geq \log(1 - x)$ for any $x > 0 $, we have
    \begin{align*}
         - \sum^{\tau_3-1}_{t=t_2} \delta_t \geq \sum^{\tau_3 -1}_{t=t_2}\log\left(1 - \delta_t\right) \geq -4 \cdot \log\left(\frac{C}{D}\right) + 3 \cdot \log(1-\Cshrink) -2,
    \end{align*}
    and multiplying both sides by $-1$ completes the proof.
\end{proof}

The next lemma is based on Jensen's inequality and the fact that $|\Delta_{t+1}|$ and $\tilde{\Delta}_{t+1}$ are identical except when the process informs all vertices.

\begin{restatable}{lemma}{simprevjensentwo}\label{lem:simple_reversed_jensen_two}
	Consider the setting of \cref{thm:shrinkingphase}. For any round $t \geq 0$ with $n/2 \leq |I_t| \leq n-1$, 
	\[
	\EX{t}{ X_t }   
	\leq \left(1 - \min\left( \frac{1}{2 \left(1 - \Cshrink\right) \cdot |U_t|} , \frac{1}{2} \right) \right)\cdot \log\left(1 - \delta_t \right), 
	\]   
\end{restatable}
\begin{proof}
	Recall that $\alpha=1/2$ and $\tilde{\Delta}_{t+1} := \min\left( |\Delta_{t+1}|, |U_t| - \alpha \right)$ from \eqref{eq:deltatilde}.	First, by Jensen's inequality,
	\begin{align}
		\EX{t}{ X_t } &= \EX{t}{ \log\left( 1 - \frac{\tilde{\Delta}_{t+1}}{|U_t|} \right) } \leq \log\left( 1 - \EX{t}{\frac{\tilde{\Delta}_{t+1}}{|U_t|}} \right). \label{eq:erste}
	\end{align}
	Note that the random variables $\tilde{\Delta}_{t+1} $ and $|\Delta_{t+1}|$ agree unless $|\Delta_{t+1}|=|U_{t}|$, i.e., all vertices get informed. For that the probability can be bounded by Markov's inequality,
	\begin{equation}\label{eq:deltabound}
	\Pro{ |\Delta_{t+1}|=|U_{t}| } \leq \delta_t,
	\end{equation}
	where we used the assumption that $|U_t|\leq n/2$. Hence,
	\[ 
		\EX{t}{\tilde{\Delta}_{t+1}}  = \sum_{i=1}^{|U_t|-1}\Pr{t}{|\Delta_{t+1}|=i}\cdot i +  \Pr{t}{|\Delta_{t+1}| = |U_t|} \cdot \left(|U_t|-\frac{1}{2}\right) =\EX{t}{|\Delta_{t+1}|} - \frac{1}{2}\cdot \Pr{t}{|\Delta_{t+1}| = |U_t|} . \] Now by \eqref{eq:deltabound}, and using $|U_t|\leq n/2$ in $(a)$, we have\begin{equation}
\EX{t}{\tilde{\Delta}_{t+1}}		\geq \EX{t}{|\Delta_{t+1}|} - \delta_t \cdot \frac{1}{2} \stackrel{(a)}{=} \delta_t \cdot |U_t| - \delta_t \cdot \frac{1}{2} = \delta_t \cdot |U_t| \cdot \left(1 - \frac{1}{2\cdot |U_t|}\right), \label{eq:upperboundExlem3.9}
\end{equation}
  Using this in \cref{eq:erste}
	we have, 
	\begin{align*}
		\EX{t}{X_t} &\leq \log\left( 1 - \frac{1}{|U_t|} \cdot \EX{t}{  \tilde{\Delta}_{t+1}}  \right) \\
		&\leq \log\left(1 - \left(1 - \frac{1}{2 |U_t|} \right) \cdot \delta_t \right) \\
		&\stackrel{(a)}{\leq} \left(1 - \frac{1}{2(1-\Cshrink) \cdot |U_t|} \right) \cdot \log\left(1 - \delta_t \right),
	\end{align*}
	where $(a)$ is by \cref{clm:log_claim_new}, having used the fact that we are working with a $\Cshrink$-shrinking process, and thus $\delta_t \leq \Cshrink < 1$ by (\BEGT). For the second factor in the minimum expression, note that by \cref{eq:upperboundExlem3.9}
	\[
	\EX{t}{ \tilde{\Delta}_{t+1} } \geq \frac{1}{2} \cdot \delta_t \cdot |U_t|,
	\]
	as $|U_t| \geq 1$ by assumption on $t$.
	Then, 
	\begin{equation*}
		\EX{t}{X_t} \leq \log\left( 1 - \frac{1}{|U_t|} \cdot \EX{t}{  \tilde{\Delta}_{t+1}}  \right)  \leq  \log\left( 1 - \frac{1}{2} \cdot \delta_t  \right). 
	\end{equation*}
Now, using the fact that $-2x\leq \log(1-x)\leq -x$ when $x\in [0,1/2]$, and noting that $ 0\leq \frac{1}{2}\cdot \delta_t \leq \frac{1}{2}\cdot < \frac{1}{2}$ by (\BEGT), we obtain
\begin{equation*}
		\EX{t}{X_t}\leq  \log\left( 1 - \frac{1}{2} \cdot \delta_t  \right)\leq - \frac{1}{2} \cdot \delta_t\leq \frac{1}{2} \cdot\log\left( 1 -  \delta_t  \right).
	\end{equation*}The result follows from taking the minimum of both cases. 
\end{proof}
After these preparations, we are now ready to prove \cref{thm:shrinkingphase}.
\begin{proof}[Proof of \cref{thm:shrinkingphase}]
In our proof it is convenient to assume that $t_2 = 0$, i.e.\ initially $|U_0 |\leq D \leq n/2$.  Let $\tau = \min\{t\geq 0: U_t = 0\}$, and as usual $\min \emptyset = 0$. Recall that, for $t< \tau \wedge \tau_3$
\begin{equation*}
        X_t:= \log \left(  1 - \frac{\tilde{\Delta}_{t+1}}{|U_t|} \right). 
    \end{equation*}
For technical reasons, we need to extend the definition of $X_t$ for $t\geq \tau \wedge \tau_3$, so we set $X_t:= 0$ for $t\geq \tau \wedge \tau_3$ but this does not affect the result, as we will not look further than $\tau_3$ and if we reach $\tau$ then the process is over.  Similarly to the growing phase, let us define for $t <\tau \wedge \tau_3$,
\begin{equation*}
    Y_t:= \sum^{t-1}_{t=0}\left(X_{i+1} - \EX{i}{X_{i}}\right)
\end{equation*}
with $Y_0 := 0$. By construction, $\left(Y_{t\wedge \tau_3 \wedge \tau} \right)_{t\geq 0}$ is a zero-mean martingale with respect to $U_{t_2},U_{t_2+1},\dots, U_{\tau_3-1}$. We would like to apply \cref{lem:azuma_variance_UT}, to upper bound the number of uninformed vertices. To this end, we need to provide a bound on $Y_{t+1} - Y_{t}$ when $t <\tau \wedge \tau_3$. We note that,
\begin{align*}
   Y_{t+1} - Y_{t} &= \sum^{t}_{i=0}\left(X_i - \EX{i}{X_i}\right) - \sum^{t-1}_{i=0}\left(X_i - \EX{i}{X_i}\right)=  X_t - \EX{t}{X_t}.
\end{align*}
Now, using that $X_t = \ind_{\{t<\tau\wedge \tau_3\}}\log\left(1 - \frac{\tilde{\Delta}_{t+1}}{|U_t|}\right) \leq 0$, and then the first statement of \cref{lem:pseudo_capping} (first statement) we get for $t <\tau\wedge \tau_3$
   \begin{align*}
 Y_{t+1} - Y_{t}   &\leq 0 - \EX{t}{\log\left(1 - \frac{\tilde{\Delta}_{t+1}}{|U_t|}\right)}\leq \kappa:= M,
\end{align*}
for some $\kappa>0$. Further, applying \cref{lem:pseudo_capping} (second statement) yields that for any $t\geq 0$,
\begin{align*}
    \VAR{t}{X_t}    \leq \nu \cdot \delta_t\ind_{\{t<\tau \wedge \tau_3\}} ,
\end{align*}
for some constant $\nu > 0$. Next, by \cref{claim:UBpreconSHRINKING} it follows that for any $T\geq 0$ that
\begin{equation*}
   \sum^{T-1}_{t=0} \VAR{t}{X_t} \leq \nu \sum^{T}_{t=0} \ind_{\{t<\tau \wedge \tau_3\}}\delta_t \leq \nu \sum^{\tau_3}_{t=0} \delta_t \leq \nu \cdot 4 \left(\log\left(\frac{C}{D}\right) - \log\left(1 - \Cshrink\right) + 1\right).
\end{equation*}
Now, we can apply the martingale concentration inequality of \cref{lem:azuma_variance_UT} to the stopped martingale $\hat Y_{t} = Y_{t\wedge \tau \wedge \tau_3}$ to show that for any $t\geq 0$ and any $h>0$, 

\begin{equation}\label{eq:preazum}
    \Pr{t_2}{\hat Y_{t} > h}\leq \exp\left(- \frac{h^2}{2 \cdot \left(\nu \cdot4 \left(\log\left(\frac{C}{D}\right) - \log\left(1 - \Cshrink\right) + 1\right) + \frac{1}{3}\cdot h\cdot \kappa\right)}\right). 
\end{equation}
  Let us set \begin{equation} \label{eq:defofhU} h:= \left(\log\left(\frac{C}{D}\right) - \log(1 - \Cshrink) + 1\right)^{2/3} \geq 1.\end{equation}
  Then, by \cref{eq:preazum}, we have

  \begin{equation}\label{eq:unvisitedazuma}
  	\Pr{t_2}{\hat Y_{t} >  h}  \leq \exp\left(-\frac{h^2}{2 \cdot \left(\nu \cdot 4\cdot h^{3/2} + \frac{1}{3}\cdot h \cdot \kappa\right)}\right) \leq  \exp\left(-C_2 \cdot h^{1/2}\right),
  \end{equation}
 where $C_2:= (8\cdot \nu + 2\kappa/3)^{-1}$ is a  constant.  We now claim:
 \begin{equation}\label{eq:claim3.6}
 	\text{Conditional on $|U_{t_2}|\leq C$, we have } \{Y_{\tau_3 \wedge \tau} \leq   h\}\cap 	\{\tau_3\wedge\tau<\infty\} \subseteq  \{ |U_{\tau_3\wedge \tau }|\leq D\}\cap \{\tau_2\wedge\tau<\infty\}.
 \end{equation} We prove this later, first we show how this establishes the theorem.  By \eqref{eq:unvisitedazuma}, we have 
 \begin{equation}\label{eq:probfiniteU} \Pr{t_2}{Y_{\tau_3 \wedge \tau}>  h,\;  \tau_2\wedge \tau <\infty} \leq \exp\left(-C_2\cdot h^{1/2} \right).\end{equation}
 Observe  that  $|U_{\tau_3}| \leq |U_{\tau_3\wedge \tau }|  $ by monotonicity (\MonoT). Using this fact, then \eqref{eq:claim3.6}, and finally \eqref{eq:probfiniteU}, we have
 \begin{align}
 	\Pr{t_2}{ |U_{\tau_3}| >D  ~\Big|~ |U_{t_2}| \leq B }&\leq 	\Pr{t_2}{ |U_{\tau_3\wedge\tau_3}| >D   ~\Big|~ |U_{t_2}| \leq B} \notag \\
 	&=   \Pr{t_2}{ |U_{\tau_3\wedge\tau_3}| >D,  \tau_2\wedge \tau <\infty ~\Big|~ |U_{t_2}| \leq B}\notag \\
 	&\qquad+ \Pr{t_2}{|U_{\tau_3\wedge\tau_3}| >D, \tau_2\wedge\tau =\infty ~\Big|~ |U_{t_2}| \leq B }\nonumber\\
 	&\leq   \Pr{t_2}{ Y_{\tau_3 \wedge \tau}> h,  \tau_2\wedge \tau <\infty ~\Big|~ |U_{t_2}| \leq B }\notag \\
 	&\qquad+ \Pr{t_2}{  \tau_2\wedge\tau =\infty ~\Big|~ |U_{t_2}| \leq B }\nonumber\\
 	&\leq \exp\left(-C_2\cdot h^{1/2} \right)+ \Pr{t_2}{  \tau_2  =\infty ~\Big|~ |U_{t_2}| \leq B }\notag, 
 \end{align}
 which, recalling the definition \eqref{eq:defofhU} of $h$, gives the bound in the statement. 
 
 It remains to prove the claimed containment in \eqref{eq:claim3.6}. For that we analyze the behavior of $|U_{\tau_3\wedge \tau }|$ when the event $ \{Y_{\tau_3 \wedge \tau} \leq   h\}\cap 	\{\tau_3\wedge\tau<\infty\}$ holds. We will split into two cases. 
 
 \medskip 
 
 \noindent In the first case $\{Y_{\tau_3 \wedge \tau} \leq  h\}\cap 	\{\tau<\infty, \tau \leq  \tau_3\}$. Hence, $|U_{\tau_3\wedge \tau }| = |U_{\tau}| =0\leq D$. 
 
 \medskip
 
 \noindent In the second case $\{Y_{\tau_3 \wedge \tau} \leq  h\}\cap 	\{\tau_3<\infty, \tau_3 <  \tau\}$. Thus, $Y_{\tau_3\wedge \tau}  = Y_{\tau_3} $, and so (deterministically) we have
\begin{equation*}
    Y_{\tau_3} = \sum^{\tau_3-1}_{t=0}\left(X_t - \EX{t}{X_t}\right) \leq h.
\end{equation*}
Rearranging this and applying \cref{lem:simple_reversed_jensen_two}, we get that,
\begin{equation*}
    \sum^{\tau_3 -1}_{t=0}X_t  \leq \sum^{\tau_3-1}_{t=0}\EX{t}{X_t} + h \stackrel{(a)}{\leq} \gamma\cdot \sum^{\tau_3-1}_{t=0}\log\left(1 - \delta_t\right)+h, \end{equation*} where $\gamma := \left(1 - \min\left( \frac{1}{2 (1-\Cshrink) \cdot D} ,  \frac{1}{2} \right) \right).$ Now, by our choice of $h$ from \cref{eq:defofhU},
    \begin{equation*}  \sum^{\tau_3 -1}_{t=0}X_t  \leq \gamma\cdot\sum^{\tau_3-1}_{t=0}\log\left(1 - \delta_t\right) + \left(\log\left(\frac{C}{D}\right) - \log\left(1 - \Cshrink\right) + 1\right)^{2/3} =: K.
\end{equation*}
To determine a sufficient precondition on the sum over the $X_t$'s, recall that we want to achieve the conclusion that $\sum^{\tau_3-1}_{t=0}X_t\leq \log\left(\frac{D}{C}\right)$ (see \cref{eq:goal}), which is implied by $K \leq \log\left(\frac{D}{C}\right)$. Rearranging that, we conclude
\begin{align*}
    \sum^{\tau_3-1}_{t=0}\log(1 - \delta_t) \leq -\frac{1}{\gamma}\cdot \left(\log\left(\frac{C}{D}\right) + \left(\log\left(\frac{C}{D}\right) - \log\left(1 - \Cshrink\right) + 1\right)^{2/3}\right),
\end{align*}
which is the condition on $\tau_3$ stated in the theorem.  
\end{proof}

\section{Applications}\label{sec:applications}

In this section we will apply our general results to more concrete credibility functions, protocols and graph classes. We do not give an exhaustive list of all results that could be derived from our analysis framework, but instead choose to analyze some natural models with decaying credibility, and show that despite the flexible and abstract nature of the framework, we can recover some known results. Roughly speaking, in this section we will first present results that are very general but not necessarily tight, followed by more specific results that are asymptotically tight up to lower order terms.
 
We will now outline the general approach followed in this section. To control the growth of $|I_t|$ we break the process into $j$ phases defined by time steps $[t_i,t_{i+1})$ for $1 \leq i \leq j$. With each phase $i$ we associate two values $A_i$ and $B_i$, where $A_i<B_i$, such that at the beginning of the $i$-th phase the informed set has size at least $A_i$ and w.h.p.\ when the phase ends the informed set has size at least $B_i$. We use the size of the informed set at the end of the previous phase as a lower bound on the size of the informed set throughout the current phase (i.e.\ $B_{i-1} = A_{i})$. The w.h.p. guarantees on the length and growth of phases are provided by \cref{cor:AzumaGrowing} and \cref{cor:shrinkingphase} (which are direct consequences of \cref{thm:AzumaGrowing} and \cref{thm:shrinkingphase} respectively). These results also give us expressions for the time to finish the phase i.e.\ $t_{i+1}-t_i$.

\begin{definition}\label{def:deltatI}
For a round $t \geq 0$ and any subset $I \subseteq V$ with $1 \leq |I| \leq n-1$, let
    \begin{equation}\label{eq:deltat}\delta_t(I) := \EX{t}{ \delta_t \, \mid \, I_t = I } = \frac{1}{\min(|I_t|,|U_t|)} \cdot \EX{t}{ |\Delta_{t+1}| \, \mid \, I_t = I },\end{equation}
    be the expected growth factor, conditional on $I_t=I$ (this is in fact, a deterministic quantity). 
    Further, for a fixed range of $[A,B]$, we define a worst-case lower bound on the expected growth factor (which only depends on $t$) by 
  \begin{equation}\label{eq:deltparam}
        \delta_t^{[A,B]} := \min_{I \subseteq V \colon A \leq |I| \leq B} \delta_t(I).
  \end{equation}
\end{definition}
Note that $\delta_t(I)$ depends on the structure of the set $I$ (e.g., the conductance), as well as on $q(t)$. However, for the more coarse quantity $\delta_t^{[A,B]}$, we only need $A \leq |I| \leq B$. In order to separate these two factors, we also define the following deterministic quantities, 

 \begin{equation}
     \Phigrow(t) := \min_{\substack{I \subseteq V \colon \\ 1 \leq |I| \leq n-1}} \frac{\delta_t(I)}{q(t)} \qquad \text{and}\qquad\Psigrow(t) := \max_{\substack{I \subseteq V \colon \\ 1 \leq |I| \leq  n-1}} \frac{\delta_t(I)}{q(t)}. \label{eq:phigrowtrivial}
 \end{equation}
 Moreover, we define $\Phigrow := \min_{t \geq 0} \Phigrow(t)$ and $\Psigrow := \max_{t \geq 0} \Psigrow(t)$.

\begin{definition}
    For any subset $I \subseteq V$ with $1 \leq |I| \leq k \leq n-1$, 
    \begin{equation*}
        \phi_k := \min_{1 \leq |I| \leq k}\varphi(I).
    \end{equation*}
\end{definition}

The following corollary is a direct consequence of \cref{thm:AzumaGrowing}.

 \begin{corollary}\label{cor:AzumaGrowing}
  Let $(G_t)_{t\geq 0}$ be any sequence of regular $n$-vertex graphs and consider a $\Cgrow$-growing process $\mathcal{P}$. Let $A,B$ be thresholds satisfying $1 \leq A \leq B \leq n/2$. Moreover, let $\nu_t^{[A,B]}$ be deterministic quantities such that $\nu_t^{[A,B]} \leq \delta_t^{[A,B]}$ for all $t \geq 0$. Let $t' \geq 0$ be any round such that $|I_{t'}|\geq A$, and define $t^* \in \mathbb{N}$ as
\begin{equation}
   t^* := \min \left\{ s \geq t' \colon \sum^{s-1}_{t=t_1} \log \left(1 +\nu_t^{[A,B]} \right) \geq \frac{ \log\left( \frac{B}{A} \right) + \left( \log\left( \frac{B}{A} \right) +  \log(1+\Cgrow) + 1 \right)^{2/3}}{\left(1-(1-\xi)\cdot A^{-\xi} \right)^2}\right\}, \label{eq:cor_prod_pre_grow}  
\end{equation}
where, $\xi:=10^{-30}$. Assume that $t^* < \infty$, then there is a constant $C_2 > 0$ such that
\begin{equation*}
    \Pr{t'}{ |I_{t^*}| \geq B  ~\Big|~ |I_{t'}| \geq A } \geq 1 - \exp\left( - C_2 \cdot \left( \log\left( \frac{B}{A} \right) \right)^{1/3}  \right).
\end{equation*}   
 \end{corollary}
The following corollary is a direct consequence of \cref{thm:shrinkingphase}.
 \begin{corollary}\label{cor:shrinkingphase}
  Let $(G_t)_{t\geq 0}$ be any sequence of regular $n$-vertex graphs and consider a $\Cshrink$-shrinking process $\mathcal{P}$. Let $C,D$ be thresholds that satisfy $n/2 \geq C\geq D\geq \frac{3}{4}$. Moreover, let $\nu_t^{[C,D]}$ be deterministic quantities such that $\nu_t^{[C,D]} \leq \delta_t^{[C,D]}$ for all $t\geq 0$. Let $t'\geq 0$ be a round such that $|U_{t'}|\leq C$. We define $\hat{t} \in \mathbb{N}$ as
\begin{equation}
   \hat{t} := \min \left\{ s \geq t' \colon \sum^{s -1}_{t=t_2}\log\left(1 - \nu_t^{[C,D]} \right) \leq  - \frac{1}{\tilde{\gamma}} \cdot \left(\log\left(\frac{C}{D}\right) + \left(\log\left(\frac{C}{D}\right) - \log\left(1 - \Cshrink\right) + 1\right)^{2/3}\right)\right\}, \label{eq:cor_product_precondition_LB_phaseshrinking}
\end{equation}
where
\begin{equation*}
\tilde{\gamma} := \left(1 - \min\left( \frac{1}{2 (1-\Cshrink) \cdot D} , \frac{1}{2} \right) \right).
\end{equation*}
Assume that $\hat{t} < \infty$, then there is a constant $C_2 >0$ such that
\begin{equation*}
    \Pr{t'}{ |U_{\hat{t}}| \leq D  ~\Big|~ |U_{t'}| \leq C } \geq 1 - \exp\left(- C_2 \cdot \left(\log\left(\frac{C}{D}\right)\right)^{1/3}\right).
\end{equation*}   
 \end{corollary}

\subsection{Arbitrary Credibility}\label{sec:applicationArbitrary}
 
 \begin{restatable}{theorem}{generallower}\label{thm:general_lower} 
 Let $(G_t)_{t\geq 0}$ be any sequence of regular $n$-vertex graphs, and $q(t)$ be an arbitrary credibility function. Let $T\geq 1$ be a deterministic number of rounds such that for some small $\rho \in (0,1)$ (not necessarily constant) such that,
 \[
  \sum_{t=0}^{T-1} \log\Bigl(1 +  \Psigrow(t) \cdot q(t)\Bigr) 
  \leq  \log n + \log \rho .
 \]
 Then,
 $
  \E{ |I_T| } \leq \rho \cdot n,
 $
 and hence by Markov's inequality, for any $\eta > 0$ (not necessarily constant),
 \[
  \Pro{ |I_T| \leq \rho \cdot n^{1+\eta} } \leq n^{-\eta}.
 \]
\end{restatable}
\begin{proof}
    Observe that, for any $T \geq 1$, by the definition of $\Psigrow(t)$ from \eqref{eq:phigrowtrivial} we have
     \begin{equation*}  \EX{T-1}{|\Delta_{T}|} \leq   |I_{T-1}|\cdot \max_{I \subseteq V \colon 1 \leq |I| \leq n-1}   \frac{ \EX{T-1}{ |\Delta_T| \, \big| \, I_{T-1}=I } }{ \min(|I_{T-1}|,|U_{T-1}|) } 
      =|I_{T-1}|\cdot\Psigrow(T-1) \cdot q(T-1)  .  \end{equation*}
  Consequently, for any $T\geq 1$, the above bound gives us 
     \begin{equation}\label{eq:Deltabound}
     \EX{T-1}{|I_T|} = |I_{T-1}| + \EX{T-1}{|\Delta_{T}|} \leq |I_{T-1}|\cdot(1+\Psigrow(T-1) \cdot q(T-1)).  
 \end{equation} 
  Now, for $T=0$, we set $Y_{T}=Y_{0}:=1$, and for $T\geq 1$ we define the random variable 
  \[Y_{T} := \frac{|I_T|}{ \prod_{t=0}^{T-1} \left(1 +  \Psigrow(t) \cdot q(t)\right) }.\] We will verify that $Y_T$ is a super-martingale. Indeed, for any $T \geq 1$, \cref{eq:Deltabound} gives us 
  \[\EX{T-1}{Y_T}  = \frac{\EX{T-1}{|I_T|}}{\prod^{T-1}_{t=0}\left(1 + \Psigrow(t) \cdot q(t)\right)} \leq \frac{|I_{T-1}|\cdot(1+\Psigrow(T-1) \cdot q(T-1))}{\prod^{T-1}_{t=0}\left(1 + \Psigrow(t) \cdot q(t)\right)} = \frac{|I_{T-1}|}{\prod^{T-2}_{t=0}\left(1 + \Psigrow(t) \cdot q(t)\right)}=Y_{T-1}.  \]
 Additionally, we note that $Y_{T}$ is bounded, as $Y_{T} \leq |I_T|\leq n$. Therefore, by the optional stopping theorem (\cref{lem:OST}),
 \[
 \E{ Y_T } \leq Y_0 = 1.
 \]
 Rearranging,
 \begin{align*}
     \E{ |I_T|} &\leq \prod_{t=0}^{T-1} \left(1 + \Psigrow(t) \cdot q(t) \right) = \exp\left( \sum_{t=0}^{T-1} \log \bigl(  1 + \Psigrow(t) \cdot q(t)  \bigr) \right) \leq \rho \cdot n,
 \end{align*}
 where the last inequality follows by the precondition.  
 \end{proof}
Next, we state two central results lower bounding the number of informed vertices, which both hold for arbitrary credibility functions.  
We start with a rather crude bound, which is simple to prove. 
 
\begin{restatable}{theorem}{general}\label{thm:general}Let $(G_t)_{t\geq 0}$ be any sequence of regular $n$-vertex graphs and $\kappa>0$ be any constant. Consider a process $\mathcal{P}$ which is both a $\Cgrow$-growing process and a $\Cshrink$-shrinking process, where $\Cshrink \leq 1-n^{-\kappa}$, with an arbitrary credibility function $q(t)$. If $T$ is a number of rounds satisfying,
\[
 \sum_{t=0}^{T-1} \log \left( 1 +  \delta_t^{[1,n-1]}\right) \geq   (2/\xi + \kappa)\cdot \log n,
\]
where $\xi :=10^{-30}$ then, we have
\[
 \Pro{|I_T| = n } \geq 1-o(1).
\]
\end{restatable}

\begin{proof}
We divide the process into two phases by defining (inductively) the following two deterministic times
\begin{align*}
    t_1:=& \min \left\{t\geq 0: \sum^{t_1-1}_{t=0}\log\left(1 + \delta_t^{[1,n/2-1]} \right) \geq \frac{3 \log n}{2\xi}\right\}\\
    t_2:= & \min \left\{t \geq t_1: \sum^{t_2-1}_{t=t_1} \log\left(1 - \delta_t^{[n/2,n-1]}\right)\leq -\frac{2 \log n}{\tilde{\gamma}}\right\},
\end{align*}
Recall that $\xi :=10^{-30}> 0$ (as in \cref{cor:AzumaGrowing}) and $\tilde{\gamma} = \tilde{\gamma}(\Cshrink)>1/2$ is the constant from \cref{cor:shrinkingphase}. Assuming that $t_1 \leq t_2 < \infty$, we apply \cref{cor:AzumaGrowing} for the first phase (rounds $[0,t_1)$) and \cref{cor:shrinkingphase} for the second phase (rounds $[t_1, t_2)$). By \cref{cor:AzumaGrowing} for $A=1,B=n/2$ and $\nu_t^{[A,B]} = \delta_t^{[1,n/2-1]}$ (which trivially satisfies the requirement of \cref{cor:AzumaGrowing} on $\nu_t^{[A,B]}$), we get that
\begin{equation*}
    \Pro{|I_{t_1}|<n/2} = o(1).
\end{equation*}
In the same way, by \cref{cor:shrinkingphase} for $C=n/2, D=3/4$ and $\nu_t^{[C,D]} = \delta_t^{[n/2,n-1]}$, 
\begin{equation*}
    \Pr{t_1}{|U_{t_2}| > 0 \, \mid \, |U_{t_1}| \leq n/2} = o(1).
\end{equation*}
Thus by the Union Bound, 
\[
 \Pro{|I_{t_2}|=n} \geq 1-o(1).
\]
To prove the theorem statement, it remains to prove that $t_2\leq T$. 

Note that since $\mathcal{P}$ is $\Cgrow$-growing process, by \cref{def:processes} (\BEG)   we have $ \delta_t^{[1,n/2-1]} \leq \Cgrow$ for any $t\geq 0$. Therefore, by minimality in the definition of $t_1$,
\begin{equation}\label{eq:up1}
 \sum_{t=0}^{t_1-1} \log\left(1 + \delta_t^{[1,n/2-1]}\right) \leq \frac{3 \log n}{2\xi} + \log (1+\Cgrow) \leq  \frac{3 \log n}{2\xi} + \log 2 \leq  \frac{5 \log n}{3\xi}.
\end{equation}
Similarly,  since $\mathcal{P}$ is a $\Cshrink$-shrinking process, by \cref{def:processes} (\BEGT), $\delta_t^{[n/2,n-1]} \leq \Cshrink$. Moreover, as $\Cshrink\leq 1-n^{-\kappa}$ by hypothesis, and as $\tilde{\gamma}\geq 1/2$, we have 
\begin{equation*} 
    \sum^{t_2-1}_{t=t_1} \log\left(1 - \delta_t^{[n/2,n-1]}\right) \geq -\frac{2 \log n}{\tilde{\gamma}} + \log\left(1 - \Cshrink\right) \geq -4 \log n - \kappa \log n = - (4 + \kappa)\log n.
\end{equation*}
Since $1+z \leq \frac{1}{1-z} $ for $z \in [0,1)$, we get that 
\begin{equation}\label{eq:up2}
\sum^{t_2-1}_{t=t_1} \log\left(1 + \delta_t^{[n/2,n-1]}\right)  \leq - \sum^{t_2-1}_{t=t_1} \log\left(1 - \delta_t^{[n/2,n-1]}\right) \leq  (4 + \kappa)\log n.
\end{equation}Since $\delta_t^{[1,n-1]} = \min \left(\delta_t^{[1,n/2-1]}, \delta_t^{[n/2,n-1]}\right)$ by \eqref{eq:deltparam},  by combining \cref{eq:up1} and \cref{eq:up2}, we obtain,
\begin{align*}
\sum^{t_2-1}_{t=0} \log\left(1 + \delta_t^{[1,n-1]}\right) \leq  \sum_{t=0}^{t_1-1} \log\left(1 + \delta_t^{[1,\frac{n}{2}-1]}\right)  + \sum^{t_2-1}_{t=t_1} \log\left(1 + \delta_t^{[\frac{n}{2},n-1]}\right) \leq \left(\tfrac{5 }{3\xi}+ 4 + \kappa\right)\log n <  \left(\tfrac{2 }{\xi}+ \kappa\right) \log n, 
\end{align*}
thus $t_2\leq T$ since $\xi := 10^{-30}$.  
\end{proof}
The next result applies to \push and \pull.

\begin{restatable}{theorem}{generalstrong}\label{thm:general_strong}
	Let $(G_t)_{t\geq 0}$ be a sequence of regular $n$-vertex strong expander graphs, with largest non-trivial eigenvalues $(\lambda_t)_{t\geq 0}$ and let $\lambda := \sup_{t\geq 0} \lambda_t$. Consider the \push or \pull model and let $q(t)$ be an arbitrary credibility function such that for,
	\begin{equation}
		\epsilon := 1 - \max_{t \geq \frac{1}{2 \log(2)} \cdot \log(n)} q(t), \label{eq:defepsilon}
	\end{equation}
	we have that $\epsilon \geq \frac{1}{\log n}$. Let $\mathcal{P} \in \{\push,\pull\}$, and assume that $T_{\mathcal{P}}$ and $q(t)$ satisfy, 
	\begin{equation}
		\sum^{T_{\mathcal{P}}}_{t=0} \log\left(1 + q(t)\right) \geq \frac{1}{\gamma_{\mathcal{P}}} \cdot \frac{\log n + 7 \left(\log n\right)^{2/3}}{ \left(1 - (1-\xi) \cdot \left(\log n\right)^{-\xi}\right)^{2}} \label{eq:mainprecond4.7},
	\end{equation}
	where $\xi:=10^{-30} $, $\gamma_{\pullindex}:= 1-\lambda$, and $\gamma_{\pushindex}:= 1 - 7 \sqrt{\lambda + 1/\log n}$. Then,
	\[
	\Pro{ |I_{T_{\mathcal{P}}}| \geq n \cdot \left(1 - \exp(- \sqrt{\log n})\right) 
	} \geq 1-o(1).
	\]\end{restatable}

Matching previous works \cite{daknama_panagiotou_reisser_2021,doerr2017randomized}, for \pull and fixed $q(t) \in (0,1)$ our result implies that in $(1+o(1)) \cdot \frac{\log n}{\log(1+q)}$ rounds the majority of the vertices get informed. The same result also holds for \push. However, it is important to note that in the results above we do not consider the time to inform \emph{all} $n$ vertices, see \cref{sec:fixed_credibility} for more results on this model. 

\begin{proof}Let us first give the proof of \pull here. Let $D:=n\cdot e^{-\sqrt{\log n}}$. Recall that by \cref{lem:protocol_growth} \eqref{itm:PROTGROWTH_pull}, \pull is a $\Cgrow$-growing process, for $\Cgrow:=1$. Moreover, definition of $\epsilon$ and \cref{lem:protocol_growth} \eqref{itm:PROTGROWTH_pull}, for all rounds $t \geq t_3 \geq \frac{1}{2 \cdot \log 2}\cdot \log n$, \pull is a $\Cshrink$-shrinking process for $\Cshrink:= 1 - \epsilon$. We inductively define the following time steps, where $\xi:=10^{-30} > 0$ is the constant defined in \cref{cor:AzumaGrowing}.
	\begin{align*}
		t_1 &:= \min \left\{t \geq 0 : \sum^{t_1-1}_{t=0} \log \left(1 + q(t) \cdot \phi_{\log n}\right) \geq \frac{\log \log n +  \left(\log \log n  + \log\left(1 + \Cgrow \right) + 1\right)^{2/3}}{\xi^2}\right\},\\
		t_2 &:= \min \left\{t \geq t_1 : \sum^{t_2-1}_{t=t_1} \log \left(1 + q(t) \cdot \phi_{\frac{n}{\log n}}\right) \geq \frac{\log \left(\frac{n}{(\log n)^2}\right) + \left(\log \left(\frac{n}{\left(\log n\right)^2}\right) + \log(1 + \Cgrow) + 1\right)^{2/3}}{\left(1 - \left(1 - \xi\right)\cdot (\log n)^{-\xi}\right)^2}\right\},\\
		t_3 &:= \min \left\{t \geq t_2: \sum^{t_3-1}_{t=t_2} \log \left(1 + q(t) \cdot \phi_{n/2}\right) \geq \frac{\log \log n + \left(\log \log n + \log\left(1 + \Cgrow\right) + 1\right)^{2/3}}{\left(1 - \left(1 - \xi\right)\cdot \left(n/\log n\right)^{-\xi}\right)^{2}} \right\},\\
		t_4 &:= \min \left\{t \geq t_3: \sum^{t_4 -1}_{t=t_3} \log\left(1 - q(t)\cdot \phi_{n/2}\right) \leq - \frac{\log \left(\frac{n}{2 \cdot D}\right) + \left(\log\left(\frac{n}{2 \cdot D}\right) - \log\left(1 - \Cshrink\right) + 1\right)^{2/3}}{1 - \min \left(\frac{1}{2\left(1 - \Cshrink\right) \cdot D}, \frac{1}{2}\right)}\right\},
	\end{align*}
	We first verify that $t_3 \geq \frac{1}{2 \cdot \log 2} \cdot \log n$. This, by the precondition of the theorem, implies that for any $t \geq t_3$, we have $q(t) \leq 1- \epsilon$. By adding up the sums in the definition of $ t_2$, and recalling that $t_2\leq t_3$, it follows that,
	\begin{align*}
		\frac{\log n}{2} & \leq  \frac{\log \left(\frac{n}{\log n}\right) + \left(\log \left(\frac{n}{\log n}\right) + \log(1 + \Cgrow) + 1\right)^{2/3}}{\left(1 - \left(1 - \xi\right)\cdot (\log n)^{-\xi}\right)^{2}} \leq \sum^{t_2-1}_{t=t_1}\log\left(1 + q(t) \cdot \phi_{n/\log n}\right) \leq t_3 \cdot \log(1+1 \cdot 1).
	\end{align*}
	Hence, rearranging implies that $t_3 \geq \frac{1}{2 \cdot \log 2} \cdot \log n$. If we let $t_0:=0$ then we can define Phase $i$ to be the interval $(t_{i-1},t_i]$ for $i\in [4]$. 
Since $t_1\leq t_2\leq t_3\leq t_4< \infty$, the preconditions of \cref{cor:AzumaGrowing} are satisfied in Phases $1-3$, and the preconditions of \cref{cor:shrinkingphase} are satisfied in Phase $4$.

	For Phase $1$, we let $A=1, B=\log n$ and $\nu_t^{[A,B]}= q(t) \cdot \phi_{\log n}$. We note that by \cref{lem:PPP-P} \eqref{itm:Pull}, $\delta_t = q(t) \cdot \varphi(I_t)$. Therefore, our choice of $\nu_t^{[A,B]}$ satisfies the requirement of \cref{cor:AzumaGrowing}. Thus, 
	\[\Pro{|I_{t_2}| \geq \log n} \geq 1- o(1).\]
	Similarly, in Phase 2 we let $A=\log n, B=\frac{n}{\log n}$ and $\nu_t^{[A,B]} = q(t) \cdot \phi_{\frac{n}{\log n}}$ to get that 
	\[\Pr{t_1}{\Condtwo{|I_{t_2}| \geq \frac{n}{\log n}}{|I_{t_1}| \geq \log n}} \geq 1-o(1),\]
	and in Phase 3, $A = \frac{n}{\log n}, B= n/2$ and $\nu_t^{[A,B]} = q(t) \cdot \phi_t$, resulting in  
	\[\Pr{t_2}{\Condtwo{|I_{t_3}|\geq \frac{n}{2}}{|I_{t_2}|\geq \frac{n}{\log n}}}\geq 1-o(1).\] 
	Lastly, in Phase 4, we apply \cref{cor:shrinkingphase} with $C=n/2$, $D=n \cdot \exp(- \sqrt{\log n})$ and $\nu_t^{[C,D]}= q(t) \cdot \phi_{n/2}$ (where again $\nu_t^{[C,D]}$ satisfies the requirement of \cref{cor:shrinkingphase} by \cref{lem:PPP-P} \eqref{itm:Pull}). Therefore,
	\[\Pr{t_3}{\Condtwo{|U_{t_4}|\leq n \cdot \exp(- \sqrt{\log n})}{|I_{t_3}|\geq n/2}} \geq 1-o(1).\]
	By taking the union bound we obtain, 
	\[
	\Pro{|U_{t_4}| \geq  n \cdot \left(1-\exp(- \sqrt{\log n})\right)}\geq 1 - o(1). 
	\]
	Since $t_1\leq t_2\leq t_3\leq t_4$ by definition, it remains to show that if the main precondition of the theorem (\cref{eq:mainprecond4.7}) is satisfied, then  $ t_4\leq T_{\pull} $ holds. Then, since $I_{t_4} \subseteq I_{T_{\pull}}$, the proof would be complete. In the remainder of the proof for \pull, we will thus establish that $t_4 \leq T_{\pull}$ indeed holds. First, since $t_1$ and $t_2$ are chosen minimally, we know that for large $n$,
	\begin{align}
		\sum^{t_1-2}_{t=0} \log \left(1 + q(t) \cdot \phi_{\log n}\right) < \frac{\log \log n +  \left(\log \log n  + \log\left(1 + \Cgrow \right) + 1\right)^{2/3}}{\xi^2}\leq (\log n)^{2/3},\label{eq:t1_est}
	\end{align}
	and, since $\Cgrow \leq 1$, we have  
	\begin{align}
		\sum^{t_2-2}_{t=t_1} \log \left(1 + q(t) \cdot \phi_{n/\log n}\right) &<  \frac{\log \left(\frac{n}{\left(\log n\right)^2}\right) + \left(\log \left(\frac{n}{\left(\log n\right)^2}\right) + \log(1 + \Cgrow) + 1\right)^{2/3}}{\left(1 -(1 - \xi)\cdot (\log n)^{-\xi}\right)^2}\notag \\ &\leq \frac{\log n +  \left(\log n\right)^{2/3}}{\left(1 - (1 - \xi)\cdot (\log n)^{-\xi}\right)^2}\label{eq:t2_est}.
	\end{align}
	Therefore, since $\phi_{\log n} \geq \phi_{n/\log n}$ and $ \log\left(1 + q(t_1-1) \cdot \phi_{\log n} \right), \log\left(1 + q(t_2-1)\cdot \phi_{n/\log n}  \right)\leq \log 2$,
	\begin{equation*}
	\sum^{t_2-1}_{t=0} \log \left(1 + q(t) \cdot \phi_{n/\log n}\right)
		\leq \sum^{t_1-2}_{t=0} \log \left(1 + q(t) \cdot \phi_{\log n}\right) +  \sum^{t_2-2}_{t=t_1} \log \left(1 + q(t) \cdot \phi_{n/\log n}\right) + 2\log\left(2 \right).\end{equation*} Now inserting the bounds on these sums from \cref{eq:t1_est} and \cref{eq:t2_est} gives 	\begin{equation*}	\sum^{t_2-1}_{t=0} \log \left(1 + q(t) \cdot \phi_{n/\log n}\right) \leq \frac{\log n + 3 \left(\log n\right)^{2/3}}{\left(1 - (1 - \xi)\cdot (\log n)^{-\xi}\right)^2}:= M,
	\end{equation*}Now, since $\phi_{n/\log n}\geq \left(1 - \lambda\right)\left(1 - \frac{1}{\log n}\right)$ (by \cref{lem:setconductancebound}) and using \cref{clm:log_claim},  we obtain,
	\begin{equation*}
		(1 - \lambda) \left(1 - \frac{1}{\log n}\right)\cdot \sum^{t_2-1}_{t=0}\log(1 + q(t)) \leq M
	\end{equation*}
Now by rearranging and exponentiating both sides, we have
	\begin{align}
		\prod^{t_2-1}_{t=0}(1 + q(t)) \leq \exp\left(\left((1 - \lambda)\left(1 - \frac{1}{\log n}\right)\right)^{-1} M\right) \leq \exp\left(\frac{1}{1 - \lambda}\cdot \frac{\log n + 4 \left(\log n\right)^{2/3}}{\left(1 - (1 - \xi)\cdot (\log n)^{-\xi}\right)^2}  \right).\label{eq:phase124.7}
	\end{align}
Similarly to the above, as $t_3$ is also chosen minimally, 
	\begin{align}
		\sum_{t=t_2}^{t_3-1} \log \left(1 + q(t) \cdot \phi_{n/2}\right) \leq \frac{\log \log n + \left(\log \log n + \log\left(1 + \Cgrow\right) + 1\right)^{2/3}}{\left(1 - (1 - \xi)\cdot  (n/\log n)^{-\xi}\right)^{2}} +\log 2 \leq (\log n)^{2/3}. \label{eq:t3two}
	\end{align}
We can now lower bound the right-hand side as follows,
	\begin{equation}
		\sum^{t_3-1}_{t=t_2} \log \left(1 + q(t) \cdot \phi_{n/2}\right)  \stackrel{(a)}{\geq} \sum^{t_3-1}_{t=t_2} \log \left(1 + q(t) \cdot \frac{1-\lambda}{2}\right) \stackrel{(b)}{\geq} \frac{1-\lambda}{2}\cdot  \sum^{t_3-1}_{t=t_2} \log\left(1 + q(t)\right), \label{eq:twotwo}
	\end{equation}
	where $(a)$ uses $\phi_{n/2}\geq \frac{1}{2}(1 - \lambda)$ from \cref{lem:setconductancebound}, and $(b)$ follows by \cref{clm:log_claim}. Combining \cref{eq:t3two} and \cref{eq:twotwo} gives us $\sum^{t_3-1}_{t=t_2}\log \left(1 + q(t)\right) \leq \frac{2}{1-\lambda} \cdot (\log n)^{2/3}$. Exponentiating, we get
	\begin{equation}
		\prod^{t_3-1}_{t=t_2}(1 + q(t)) \leq \exp\left( \frac{2}{1-\lambda} \cdot (\log n)^{2/3}\right).\label{eq:phase34.7}
	\end{equation}
	Now, combining \cref{eq:phase124.7} and \cref{eq:phase34.7},
	\begin{align}
	 \prod_{t=0}^{t_3-1} (1 + q(t))  &\leq \exp\left(\frac{1}{1 - \lambda}\cdot \frac{\log n + 4 \left(\log n\right)^{2/3}}{\left(1 - (1 - \xi)\cdot (\log n)^{-\xi}\right)^2}  \right)\cdot \exp\left( \frac{2}{1-\lambda} \cdot (\log n)^{2/3}\right)\notag \\  &\leq \exp\left(\frac{1}{1 - \lambda}\cdot \frac{\log n + 6 \left(\log n\right)^{2/3}}{\left(1 - (1 - \xi)\cdot (\log n)^{-\xi}\right)^2}  \right). \label{eq:finalphase1234.7}
	\end{align}
Lastly, we consider Phase $4$. In this phase, since $t_3 \geq \frac{1}{2 \log 2}\cdot \log n$ and $\epsilon \geq \frac{1}{\log n}$, we have $\Cshrink \leq 1 - \frac{1}{\log n}$. Recall that $D:= n \cdot \exp\left(-\sqrt{\log n}\right)$, hence $1 - \min \left(\frac{1}{2\left(1 - \Cshrink\right) \cdot D}, 1/2\right) \geq  1 -1/\sqrt{n}$. Thus, similarly to as we did for $t_1$-$t_3$, by minimality of $t_4$,
	\begin{equation}
		\sum^{t_4 -2}_{t=t_3}\log \left(1 - q(t) \cdot \phi_{n/2}\right) \geq - \frac{\log \left(\frac{n}{2 \cdot D}\right) + \left(\log\left(\frac{n}{2 \cdot D}\right) - \log \left(1 - \Cshrink\right) + 1\right)^{2/3}}{1 - \min \left(\frac{1}{2\left(1 - \Cshrink\right) \cdot D}, 1/2\right)} \geq -\frac{(\log n)^{2/3}}{5} . \label{eq:t44.7pull}
	\end{equation}
	Thus we have, since $1-q(t) \geq \eps \geq 1/\log n$ we have   
	\begin{equation}
		\sum^{t_4 -1}_{t=t_3}\log \left(1 - q(t) \cdot \phi_{n/2}\right)     \geq   -\frac{(\log n)^{2/3}}{5} + \log\left(1 - q(t_4-1) \cdot \phi_{n/2}\right) \geq   -\frac{(\log n)^{2/3}}{4}.   \label{eq:phase4sum4.7}
	\end{equation}
 Using \cref{lem:setconductancebound} to give $\phi_{n/2}\geq \frac{1}{2}\cdot \left(1 - \lambda\right)$ in $(a)$, and then using the fact that $-2x\leq \log(1-x)\leq -x$ when $x\in [0,1/2]$, noting that $q(t)(1-\lambda)/2\in [0,1/2]$, we get
	\begin{equation}
		\sum_{t=t_3}^{t_4-1}\log(1 - q(t) \cdot \phi_{n/2}) \overset{(a)}{\leq} \sum_{t=t_3}^{t_4-1}\log\left(1 - q(t) \cdot \frac{1 - \lambda}{2} \right)\leq \sum_{t=t_3}^{t_4-1} \frac{1 - \lambda}{2}\cdot q(t)  \leq \frac{1-\lambda}{4} \sum_{t=t_3}^{t_4-1}\log(1-q(t)). \label{eq:phase4sum4.72}
	\end{equation}
	Combining \cref{eq:phase4sum4.7} and \cref{eq:phase4sum4.72} gives $\sum_{t=t_3}^{t_4-1}-\log(1-q(t)) \leq \frac{(\log n)^{2/3}}{1- \lambda}  $. Exponentiating, we obtain $\prod^{t_4-1}_{t=t_3}\left(1 - q(t)\right)^{-1} \leq \exp\left(\frac{(\log n)^{2/3}}{1- \lambda} \right) .$
	Then, since $1+z \leq \frac{1}{1-z}$ for $z \in [0,1)$, we conclude that
	\begin{align}   \prod_{t=t_3}^{t_4-1}\left(1 + q(t)\right) \leq \prod^{t_4-1}_{t=t_3}\left(1 - q(t)\right)^{-1} \leq \exp\left(\frac{(\log n)^{2/3}}{1- \lambda} \right).\label{eq:finalt44.7}
	\end{align}
	Combining \cref{eq:finalphase1234.7} and \cref{eq:finalt44.7}, we obtain, 
	\begin{align*}
		\prod_{t=0}^{t_4-1}\left(1 + q(t)\right) &\leq    \exp\left(\frac{1}{1 - \lambda}\cdot \frac{\log n + 7 \left(\log n\right)^{2/3}}{\left(1 - (1 - \xi)\cdot (\log n)^{-\xi}\right)^2}  \right).
	\end{align*}
If we define $\gamma_{\pull} := (1-\lambda)$, then this implies that $t_4 \leq T_{\pull}$ by the theorem statement, concluding the proof. 
		
		The proof for \push is very similar, however the constant $\gamma_{\pull} $ is slightly different. This is since in Phase $2$ we have to use the bound $\delta_t \geq q(t)\cdot \left(1 - 7 \sqrt{\lambda + \frac{1}{\log n}}\right)$ from \cref{lem:expushSTRONG} \eqref{itm:Push} for \push, rather than the bound $\delta_t = q(t) \cdot \varphi(I_t)$ from \cref{lem:PPP-P} \eqref{itm:Pull} for \pull. Other than that, very little changes in the proof.\end{proof}

\subsection{Power-Law Credibility}\label{sec:applicationPowerLaw}
In this part we consider a natural credibility function with a polynomial decay. 
\begin{definition}[Power-law credibility]
Let $\alpha \in (0,\infty)$ be any constant . Then, the power-law credibility function is defined for any round $t \geq 0$ as
\begin{equation*}
    q_{\alpha}(t):= (t+1)^{-\alpha}.
\end{equation*}
In particular, in the first round the credibility function is $1$.
\end{definition}

We first observe that if $\alpha > 1$, we only inform a constant number of vertices in expectation.

\begin{restatable}{proposition}{multiplicativeupperpower}\label{pro:powerlaw} Let $(G_t)_{t\geq 0}$ be any sequence of regular graphs, and consider a Growing Process such that $\EX{t}{ \frac{|\Delta_{t+1}|}{|I_t|}  }\leq \Cgrow \cdot q(t)$ for all $t\geq 0$ . Then, for any constant $\alpha > 1$, there is a constant $\kappa=\kappa(\alpha) > 0$, such that for any $T \geq 0$,
    \begin{equation*}
        \E{ |I_T| } \leq \kappa.
    \end{equation*}
\end{restatable}

The condition $\EX{t}{ \frac{|\Delta_{t+1}|}{|I_t|}  }\leq \Cgrow \cdot q(t)$ is a refinement of $\mathcal P_2$ in \Cref{def:processes}, and is satisfied by the \pull, \push, \pp processes as shown in \Cref{lem:PPP-P} by choosing $\Cgrow$ as 1, 1, and 2, respectively.
\begin{proof} 
Recall that $|I_{t+1}| = \left(1+\frac{|\Delta_{t+1}|}{|I_{t}|}\right) \cdot |I_{t}|$. Then,
\begin{align*}
    \EX{t}{|I_{t+1}|} \leq (1+\Cgrow \cdot q(t)) \cdot |I_{t}| \leq e^{\Cgrow \cdot q(t)}\cdot |I_{t}|
\end{align*}
and recursively we have $\E{|I_{T}|} \leq e^{ \Cgrow \sum_{t=0}^{T-1} q(t)}$. Since $q(t) = (t+1)^{-\alpha}$ with $\alpha>1$, it holds that $\sum_{t=0}^{\infty} q(t) < \infty$. By choosing $\kappa =e^{ \Cgrow \sum_{t=0}^{\infty} q(t)}$ we obtain the desired result.
\end{proof}

The next result considers the regime $\alpha \leq 1$, and proves that after a sufficiently long time, the rumor reaches all $n$ vertices. In particular, when $\alpha=1$, the spreading time becomes polynomial in $n$ (even if $(G_t)_{t \geq 0}$ was a sequence of expander graphs). 
\begin{restatable}{theorem}{multiplicativelowerpower}\label{thm:multiplicative_lower_power}
Let $(G_t)_{t\geq 0}$ be any sequence of regular $n$-vertex graphs, and consider a process $\mathcal{P}$ that is both a $\Cgrow$-growing process and a $\Cshrink$-shrinking process, where $\Cshrink<1$ is constant, with a power law credibility function. Then, for any constant $\alpha < 1$, there are constants $0 < \kappa_1:=\kappa_1(\alpha) < \kappa_2:=\kappa_2(\alpha)$
such that for any $T_1 \leq  \kappa_1 \cdot  (\frac{1}{\Psigrow} \cdot \log n)^{1/(1-\alpha)}$, $T_2 \geq \kappa_2 \cdot (\frac{1}{\Phigrow} \cdot \log n)^{1/(1-\alpha)}$ and any $\eta > 0$ we have,
\begin{itemize}
    \item [$(i)$] $\Pro{ |I_{T_1}| < n^{1/2 + \eta}} \geq 1-n^{-\eta}$,
    \item [$(ii)$] $\Pro{ |I_{T_2}| = n} \geq 1-o(1)$.
\end{itemize}Further, if $\alpha=1$, then there are constants $0 < \kappa_1 < \kappa_2$, such that
for any $T_1 \leq \left(\frac{1}{\Psigrow} \cdot n \right)^{\kappa_1}$ and $T_2 \geq \left(\frac{1}{\Phigrow} \cdot n \right)^{\kappa_2}$,
\begin{itemize}
    \item [$(iii)$] $\Pro{ |I_{T_1}| < n^{1/2 + \eta}} \geq 1-n^{-\eta}$,
    \item [$(iv)$] $\Pro{ |I_{T_2}| = n} \geq 1-o(1)$.
\end{itemize}
\end{restatable}

\begin{proof}We begin with the proof of $(i)$. Observe that

\begin{equation*} 
 \prod_{t=0}^{T_1-1} \left(1 + \Psigrow(t) \cdot q(t) \right)  \leq \exp\left( \Psigrow  \cdot \sum_{t=1}^{T_1} t^{-\alpha} \right) \stackrel{(a)}{\leq} \exp\left( \frac{1}{1-\alpha} \cdot \Psigrow \cdot  T_1^{1-\alpha} \right) \stackrel{(b)}{=} \exp\left(\frac{1}{1-\alpha} \cdot \kappa_1 \log n\right) \stackrel{(c)}{\leq} \sqrt{n},
 \end{equation*}
 where $(a)$ holds by \cref{clm:generalharmonic} $(b)$ holds by the definition of $T_1$ and $(c)$ holds  if $\kappa_1 \leq (\frac{1 - \alpha}{2})^{\frac{1}{1-\alpha}}$. The result then follows from \cref{thm:general_lower}, where we take $\rho = 1/\sqrt{n}$. 
 
 \smallskip

\noindent\textit{Proof of $(ii)$}: We seek to apply \cref{thm:general}.  Since $1 + x \geq e^{x/4}$ for $0 \leq x \leq 1$, and by \cref{clm:generalharmonic}, we have
\begin{equation*} 
 \prod_{t=0}^{T_2-1} \left(1 + \Phigrow(t) \cdot q(t) \right) \geq \exp\left( \frac{1}{4} \cdot \Phigrow \cdot \sum_{t=1}^{T_2} t^{-\alpha} \right) \geq  \exp\left(\frac{1}{4}\cdot \frac{1}{1-\alpha}\cdot \Phigrow \cdot \left(T_2^{1-\alpha} -1\right)\right). \end{equation*} 
Now, by the definition of $T_2$ and if $\kappa_2 \geq (\tfrac{16(1-\alpha)}{\Phigrow \xi})^{\frac{1}{1-\alpha}}$, where $\xi =10^{-30}$, then we have  \begin{equation*}
\prod_{t=0}^{T_2-1} \left(1 + \Phigrow(t) \cdot q(t) \right)   \geq \exp\left(\frac{1}{4}\cdot \frac{1}{1-\alpha}\cdot \Phigrow \cdot\left(\tfrac{16(1-\alpha)}{\Phigrow \xi}\log n - 1 \right)\right)\geq \exp\left(\left(\frac{2}{\xi} + 1\right)\cdot \log n\right) ,
\end{equation*}
 Hence, taking logarithms, we satisfy the precondition of \cref{thm:general}.

The proofs of statements $(iii)$ and $(iv)$ for $\alpha=1$ are analogous, noting that $\sum_{t=1}^{T} t^{-1} = \ln(T) + \Theta(1)$.\end{proof}

\subsection{Additive Credibility}\label{sec:additive_models}

\begin{definition}[Additive credibility]
    Let $\alpha \in (0,1)$ . Then, the additive credibility function is defined for any round $t \geq 0$ as
    \begin{equation*}
        q_{\alpha}(t)  = q(t) := (1 - t \cdot \alpha)^{+},
    \end{equation*}
    where $z^{+}=\max(z,0)$.
   In particular, in the first round (when $t=0$) the credibility function is $1$. 
\end{definition}
In comparison to the power-law credibility function, the additive credibility function has a time-independent decrease. As we will see below, the interesting regime (for expanders) is when $\alpha=\Theta(1/\log n)$. That means, unlike the power-law-credibility, the additive credibility function remains close to $1$ for a significant number of rounds. However, after $O(\log^2 n)$ steps, the credibility becomes polynomially small; much smaller than any power-law credibility at this point.

Let us consider the additive credibility function in the \push and \pull model for regular graphs. 
We also observe that if we let $T = 1/\alpha$, $I_{T}$ is the maximal set of informed vertices in every execution, as $q(t)=0$ for $t \geq T$. We start by proving an upper bound on $I_{T}$ for $T=1/\alpha$, followed by a lower bound. We remark that, due to the specific nature of $q(t)$, we can use Stirling's approximation to determine a rather precise threshold for the parameter $\alpha$.

\begin{restatable}{theorem}{additive}\label{thm:additive}
  Let $(G_t)_{t\geq 0}$ be a sequence of regular $n$-vertex strong expander graphs, and consider the \push or \pull protocol with an additive credibility function. Let $\mathcal{P} \in \{\push, \pull\}$. 
  \begin{itemize}
      \item [$(i)$]
  Let $\alpha \geq  \frac{\log \left(\frac{4}{e}\right)}{\log n + \log \zeta}$, where $\frac{1}{n} < \zeta < \frac{1}{\sqrt{2}\cdot 2}$ . Then, for any $T := 1/\alpha$ and for any $\eta > 0$ (not necessarily constant),  
    \begin{equation*}
        \Pro{ |I_T| \leq \sqrt{2}\cdot \zeta \cdot n^{1 + \eta}} \geq 1 - n^{-\eta}.
    \end{equation*}
\item [$(ii)$]
Further, let $\alpha \leq \frac{ \log\left( \frac{4}{e} \right)}{\log\left(2 \sqrt{2} \cdot \exp\left(\frac{1}{\gamma_{\mathcal{P}}}\cdot \frac{\log n + 7(\log n)^{2/3}}{\left(1 - (1-\xi)\cdot (\log n)^{-\xi}\right)^2}\right)\right)}$, for $\gamma_{\mathcal{P}}$ as in \cref{thm:general_strong}. Then, for $T:=1/\alpha$,
\[
 \Pro{ |I_T| \geq n \cdot \left( 1 - \exp(- \sqrt{\log n}) \right) } \geq 1-o(1).
\]
\end{itemize}
\end{restatable}

\begin{proof}
We begin with Item $(i)$, where we seek to apply \cref{thm:general_lower}. We note that by \cref{lem:PPP-P} \eqref{itm:Push}  and  \cref{lem:PPP-P} \eqref{itm:Pull}, for both \push and \pull, $\delta_t \leq q(t) \cdot \varphi(I_t)$. Thus, for \push and \pull, 
\begin{equation}
    \Psigrow(t) := \max_{\substack{I \subseteq V \colon \\ 1 \leq |I| \leq  n-1}} \frac{\delta_t(I)}{q(t)} \leq \max_{\substack{I \subseteq V \colon \\ 1 \leq |I| \leq  n-1}} \frac{q(t) \cdot \varphi(I)}{q(t)} \leq 1. \label{eq:pp-psigrow}
\end{equation}
Hence, by using \cref{eq:pp-psigrow} in $(a)$, \cref{clm:HelperStirling} in $(b)$,  and the definition of $\alpha$ in $(c)$, we have  
\begin{align*}
 \prod^{T-1}_{t=0}\left(1 +  \Psigrow(t) \cdot q(t)\right)  &\stackrel{(a)}{\leq} \prod^{T-1}_{t=0}(1 + q(t)) \leq \prod^{1/\alpha-1}_{t=0}(1 + (1 - t\cdot \alpha))  = \prod_{t=0}^{1/\alpha - 1} \left(2 - t \cdot \alpha \right) \stackrel{(b)}{\leq} \sqrt{2}\cdot \left(\frac{4}{e}\right)^{1/\alpha}  \stackrel{(c)}{\leq} \sqrt{2} \cdot \zeta n.
\end{align*}
Taking logarithms yields,
\[
 \sum^{T-1}_{t=0} \log(1 + \Psigrow(t) \cdot q(t)) \leq \log\left( \sqrt{2}\cdot \zeta\cdot n\right).
\]
By \cref{thm:general_lower},
we obtain
$ \Ex{ |I_{T}|} \leq \sqrt{2}\cdot \zeta \cdot n$ as well as the tail bound by Markov's inequality. 

\medskip 

\noindent\textit{Proof of Item $(ii)$}: For the second statement, we seek to apply \cref{thm:general_strong}. By a simple coupling argument, we may assume that $\alpha = \log\left( \frac{4}{e} \right)/\log\left(2 \sqrt{2} \cdot \exp\left(\frac{1}{\gamma_{\mathcal{P}}}\cdot \frac{\log n + 7(\log n)^{2/3}}{\left(1 - (1-\xi)\cdot (\log n)^{-\xi}\right)^2}\right)\right)$, since by making $\alpha$ smaller, $q(t)$ increases and therefore only more vertices get informed. Recall that as in \cref{thm:general_strong},
\begin{equation*}
		\epsilon := 1 - \max_{t \geq \frac{1}{2 \log(2)} \cdot \log(n)} q(t).  
\end{equation*}
To satisfy the precondition on $q(t)$ of \cref{thm:general_strong}, we need that $\epsilon \geq \frac{1}{\log n}$ for all $t\geq \frac{1}{2 \log 2}$. By definition of the additive credibility function (and the fact it is non-increasing in $t \geq 0$), it holds for any $t \geq \frac{1}{2 \log 2} \cdot \log n$ that
\[
 \epsilon = 1 - q\left( \frac{1}{2 \log 2} \cdot \log n \right) \geq 
 1 - (1 - \frac{1}{2 \log 2} \cdot \log n \cdot \alpha).
\]
Since the choice of $\alpha$ implies that $\alpha \geq \frac{\log(4/e)}{\frac{1}{2} \cdot \log n}$, it follows that
\[
 \epsilon \geq \frac{\log(4/e)}{4 \log 2} > \frac{1}{10}. 
\]
 
Therefore the precondition on $q(t)$ of \cref{thm:general_strong} holds. Furthermore, by using \cref{clm:HelperStirlingLB} in $(a)$,
\[
  \prod_{t=0}^{T-1} (1+ q(t)) = \prod^{1/\alpha -1}_{t=0}(1 +  q(t))  = 
 \prod^{1/\alpha -1}_{t=0}\left(2 - t\cdot \alpha\right)  \stackrel{(a)}{\geq} \frac{1}{\sqrt{2}}\left(\frac{4}{e}\right)^{1/\alpha}\cdot e^{-1/2\cdot \alpha}\]
   Now, using in $(b)$ that $e^{-1/2 \cdot \alpha} \geq \frac{1}{2}$ for our choice of $\alpha=O(\frac{1}{\log n})$, and in $(c)$ again our specific choice of $\alpha$ gives 
   \begin{equation*}
  \prod_{t=0}^{T-1} (1+ q(t))    \stackrel{(b)}{\geq} \frac{1}{\sqrt{2}}\left(\frac{4}{e}\right)^{1/\alpha}\cdot \frac{1}{2} \stackrel{(c)}{\geq} \exp\left( \frac{1}{\gamma_{\mathcal{P}}}\cdot \frac{\log n + 7\left(\log n\right)^{2/3}}{\left(1 - (1 - \xi)\cdot (\log)^{-\xi}\right)^2} \right),
\end{equation*}
Taking logarithms yields,
\[
 \sum_{t=0}^{T-1} \log(1+q(t)) \geq  \frac{1}{\gamma_{\mathcal{P}}}\cdot \frac{\log n + 7\left(\log n\right)^{2/3}}{\left(1 - (1 - \xi)\cdot (\log)^{-\xi}\right)^2}, 
\]
and hence also the precondition~\cref{eq:mainprecond4.7} in \cref{thm:general_strong} holds. Therefore, statement $(ii)$ follows by \cref{thm:general_strong}.
\end{proof}

\subsection{Multiplicative Credibility}\label{sec:applicationMultiplicative}

\begin{definition}[Multiplicative credibility]\label{def:multcred}
Let $\alpha \in (0,1)$ . Then, the multiplicative credibility function is defined for any round $t \geq 0$ as
\begin{equation*}
    q_{\alpha}(t) := (1-\alpha)^{t}.
\end{equation*}
In particular, in the first round the credibility function is $1$. 
\end{definition}
The next result is the multiplicative analogue of \cref{thm:additive}.

\begin{restatable}{theorem}{multiplicative}\label{thm:multiplicative}
Let $(G_t)_{t\geq 0}$ be any sequence of regular $n$-vertex strong expander graphs, and consider the \push or \pull protocol with a multiplicative credibility function . Then, there are constants $\kappa_1 \leq \frac{1}{2}$ and $\kappa_2 \geq  \frac{1}{8}$, such that the following holds. 
\begin{itemize}
\item[$(i)$]
If $\alpha \geq \frac{\kappa_1}{\log n}$, then for any $T \geq 1$, $\EX{t}{|I_{T}|} \leq \sqrt{n}$, and hence for any $\eta > 0$ (not necessarily constant),
    \begin{equation*}
        \Pro{ |I_T| \leq    n^{1/2 + \eta}} \geq 1 - n^{-\eta}.
    \end{equation*}

\item[$(ii)$] Further, if $\alpha \leq \frac{\kappa_2}{\log n}$, then, for any $T\geq 4 \log n$, 
\[
 \Pro{ |I_T| \geq n \cdot \left(1 - \exp(- (\log n)^{1/2} ) \right) } \geq 1-o(1).
\]
\end{itemize}
\end{restatable}

\begin{remark}
We  believe that with a more refined analysis it would be also possible to show that $\kappa_2 \geq (1-o(1)) \cdot \kappa_1$, but for the sake of simplicity and space we only prove this weaker dichotomy here.
\end{remark}

\begin{proof} We note again that by \cref{lem:PPP-P}, $\delta_t \leq q(t) \cdot \varphi(I_t)$. Hence, similarly as in the proof of \cref{thm:additive}
\begin{equation}
\Psigrow(t) := \max_{\substack{I \subseteq V \colon \\ 1 \leq |I| \leq  n-1}} \frac{\delta_t(I)}{q(t)} \leq \max_{\substack{I \subseteq V \colon \\ 1 \leq |I| \leq  n-1}} \frac{q(t) \cdot \varphi(I)}{q(t)} \leq 1. \label{eq:pp-psigrow2}
\end{equation}
\noindent\textit{Proof of $(i)$}: Moreover, we set $\kappa_1:= \frac{1}{2}$, then for any $T \geq 0$,
\begin{align*}
 \prod_{t=0}^{T-1}\left(1 + \Psigrow(t) \cdot q(t)\right) &\stackrel{(a)}{\leq} \prod_{t=0}^{T-1} \left(1 + q(t) \right)\leq  \prod_{t=0}^{\infty} \left(1 + \left(1 - \alpha \right)^t \right)\stackrel{(b)}{\leq} \prod_{t=0}^{\infty} \left(1 + \left(1 - \frac{1/2}{\log n} \right)^t \right) \stackrel{(c)}{\leq} \sqrt{n},
\end{align*}
where $(a)$ holds by \cref{eq:pp-psigrow2}, $(b)$ by our bound on $\alpha$, and $(c)$ by \cref{clm:nightmare}$(i)$. Taking logarithms, the statement holds by \cref{thm:general_lower}. 

\medskip

\noindent\textit{Proof of $(ii)$}: For the second statement, we set $\kappa_2:= \frac{1}{8}$, and then for $T:=4 \log n$,
\begin{equation}\label{eq:addlower}
\prod_{t=0}^{4 \log n-1} (1+q(t) ) = \prod_{t=0}^{4 \log n -1}\left(1 + (1 - \alpha)^t\right) \stackrel{(a)}{\geq} \prod_{t=0}^{4 \log n -1}\left(1 + \left(1 - \frac{1/8}{\log n}\right)^t\right)\stackrel{(b)}{\geq} n^{3/2},    
\end{equation}
where $(a)$ follows for our value of $\alpha$ and $(b)$ by \cref{clm:nightmare} $(ii)$.

Taking the logarithm of \eqref{eq:addlower}, we see that the first precondition of \cref{thm:general_strong} is met for $T=4 \log n$. Also, by definition of the multiplicative credibility function, we have for any constant $c_2 > 0$ another constant $\epsilon=\epsilon(c_2) > 0$ such that for any $t \geq c_2 \log n$, $q(t) \leq 1- \epsilon$, which is the second precondition . Therefore, \cref{thm:general_strong} applies, completing the proof.
\end{proof}

\subsection{Fixed Credibility}\label{sec:fixed_credibility}
Here, we consider $q(t) = q$ to be constant over time (however, $q(t)$ may depend on $n$). This model was studied in previous works \cite{daknama_panagiotou_reisser_2021,doerr2017randomized} on complete graphs and strong expanders \eqref{eq:strongexpander}, respectively (under the guise of ``robustness'').
Here we provide upper bounds for the spreading time of the \push, \pull and \pp model on regular strong expander graphs, using our framework. As the analysis between the protocols are very similar, we will only give details in the case of \push here. 
\begin{theorem}[{cf.~\cite{daknama_panagiotou_reisser_2021}}]\label{thm:fixed_credibility}
Let $(G_t)_{t\geq 0}$ be any sequence of regular $n$-vertex strong expander graphs. Let the credibility function $q(t) = q$ be constant in $(0,1-\epsilon]$ for some constant $\epsilon > 0$ and define the following times
\begin{itemize}\itemsep1pt
    \item $T_{\push} := (1+o(1)) \cdot \left( \frac{1}{\log( 1+q)} +   \frac{1}{q}\right)\cdot \log n$,
    \item $T_{\pull} := (1+o(1)) \cdot \left( \frac{1}{\log( 1+q)} - \frac{1}{\log(1-q )}\right)\cdot \log n$,
    \item $T_{\pp} := (1+o(1)) \cdot \left(\frac{1}{\log( 1+2q)}+ \frac{1}{q - \log( 1-q )}\right)\cdot \log n$.
\end{itemize}Then for each $\mathcal{P}\in \{\push, \pull, \pp\}$ we have 
\[
\Pro{ |I_{T_{\mathcal{P}}}| = n } \geq 1-o(1).
\]
\end{theorem}
We note that the corresponding result \cite[Theorem 1.2]{daknama_panagiotou_reisser_2021} in the original paper is stated only for static graphs, however it is likely that the methods in that paper would also extend to dynamic graphs. 
    
\begin{proof} Here, we only provide a proof of the \push model. We note that if $q \leq 1 - \epsilon$, by \cref{lem:protocol_growth} \eqref{itm:PROTGROWTH_push} \push is a $\Cshrink$-shrinking process for $\Cshrink=1 -\epsilon$. We will again follow the approach outlined at the beginning of this section, where we break up the analysis into phases. An overview of the running times of these phases is given in \cref{tab:RunningTimes}, also for the \pull and \pp processes. 
\renewcommand{\arraystretch}{2}
\begin{table} 
\begin{center}
\begin{tabular}{|c|c|c|c|c|}
\hline

    Phase  & Start/finish sizes             & \push & \pull & \pp          \\ \hline\hline
1 &$A=1, B=\log n$          & $\frac{\log \log n}{\log\left(1+q \right)}$ &$\frac{\log \log n}{\log(1 + q)}$ & $\frac{\log \log n}{\log\left(1+2q\right)}$           \\ \hline
2& $A=\log n, B=\frac{n}{\log n}$          & \cellcolor{Greenish} $\frac{\log n}{\log(1 + q)}$ &\cellcolor{Greenish}  $\frac{\log n}{\log\left(1 + q\right)}$ & \cellcolor{Greenish}  $\frac{\log n}{\log(1 + 2q)}$ \\ \hline
3& $A=\frac{n}{\log n}, B=\frac{n}{2}$          & $\frac{\log \log n}{\log (1 + q)}$ & $\frac{\log\log n}{\log(1 + q)}$&$\frac{\log \log n}{\log (1 + 2q)}$ \\ \hline
4& $C=n/2$, $D=\frac{n}{\log n}$          & $\frac{1}{q}  \log \log n$ & $\frac{\log \log n}{-\log\left(1-q\right)}$&$\frac{\log \log n}{q-\log\left(1- q\right)}$\\ \hline
5& $C=\frac{n}{\log n}, D=\log n$          & \cellcolor{Greenish} $\frac{1}{q}\log n$ & \cellcolor{Greenish} $\frac{\log n}{-\log\left(1 - q\right)}$& \cellcolor{Greenish} $ \frac{\log n}{q-\log\left( 1 - q\right)}$ \\ \hline
6& $C=\log n, D=\frac{3}{4}$         & $\frac{1}{q}\log\log n$ & $\frac{\log \log n}{-\log\left( 1 - q\right)}$&$\frac{\log\log n}{q - \log\left(1 - q\right)}$\\ \hline
\end{tabular}
\caption{Runtimes for \push, \pull and \pp for different phases obtained by \cref{thm:AzumaGrowing} (row 1,2,3) and \cref{thm:shrinkingphase} (row 4,5,6), all bounds hold w.h.p.. The upper bounds contained within cells shaded in \hlgreenish{~Green~} hold up to a multiplicative $(1+o(1))$ factor and it is these bound which contribute to the to total run time, all other bounds hold up to a multiplicative constant and are negligible.  Our result for \pull holds only when $q$ is bounded away from $1$, the remaining cases where $q$ is equal (or tending to) $1$ are covered in \cite{doerr2017randomized,daknama_panagiotou_reisser_2021}.}
\label{tab:RunningTimes}
\end{center}
\end{table}

\paragraph{Phase 1:}
Recall that  $\xi=10^{-30}$, as defined in \cref{cor:AzumaGrowing}, and define the time 
\begin{equation*}
    t_1 := \min \left\{s \geq 0 \colon \sum^{s-1}_{t=0} \log\left(1 + q \cdot \left(1 - \frac{q}{2}\right)\cdot\frac{1}{2} \right) \geq \frac{ \log\log n + o(\log\log n)}{\xi^2}\right\}.
\end{equation*}
We aim to apply \cref{cor:AzumaGrowing} with $A=1$ and $B=\log n$. We note that by \cref{lem:PPP-P} \eqref{itm:Push}
\begin{equation*}
\delta_t \geq \delta_t^{[1,\log n]}\geq q \cdot \left(1 - \frac{q}{2}\right)\cdot \phi_{\log n}.
\end{equation*}
Moreover, since $(G_t)_{t\geq 0}$ is a strong expander and we can assume $n$ is large, by \cref{lem:setconductancebound}
\begin{equation}\label{eq:condbound}
\phi_{\log n} \geq (1 - \lambda)\cdot \left(1 - \frac{\log n}{n}\right)
\geq \frac{1}{2}.
\end{equation}
Hence, setting $\nu_t^{[A,B]} = q \cdot \left(1 - \frac{q}{2}\right) \cdot  \frac{1}{2} $, by \cref{cor:AzumaGrowing} we get that,
\begin{equation*}
    \Pro{ |I_{t_1}| \geq \log n  ~\Big|~ |I_{0}| \geq 1 } \geq 1 - o(1).
\end{equation*} 
Let us now solve for $t_1$ to get an upper bound on how long Phase 1 takes.
We note that by minimality of $t_1$,
\begin{align*}
     (t_1 -1) \cdot \log \left(1 + q \cdot \left(1 - \frac{q}{2}\right)\cdot \frac{1}{2}\right) &= \sum_{t=0}^{t_1-1}\log\left(1 + q \cdot \left(1 - \frac{q}{2}\right)\cdot \frac{1}{2}\right)  
   \\ 
    &\leq \frac{\log \log n + o(\log\log n)}{\xi^2} + \log\left(1 + q \cdot \left(1 - \frac{q}{2}\right)\cdot \frac{1}{2} \right).
\end{align*}
Rearranging, we get that,
\begin{align*}
t_1 &\leq 2 + \frac{\log \log n + o(\log\log n)}{\xi^2 \cdot \log\left(1 + q \cdot \left(1 - \frac{q}{2}\right)\cdot \frac{1}{2} \right)}\\ 
&\stackrel{(a)}{\leq} 2 + \frac{\log \log n + o(\log \log n)}{\xi^2\cdot \left(1 - \frac{q}{2}\right)\cdot \frac{1}{2} \cdot \log \left(1 + q \right)}\\
&\stackrel{(b)}{=} O\left(\frac{\log \log n}{\log\left(1+q \right)}\right),
\end{align*}
where $(a)$ follows from \cref{clm:log_claim} and $(b)$ holds since we assume $G_t$ to be a (strong) expander.

\paragraph{Phase 2:}
Let 
\begin{equation*}
    t_2 := \min \left\{s \geq t_1 \colon \sum^{s-1}_{t=t_1} \log\left(1 + q \cdot \left(1 - 7\sqrt{\lambda + \frac{1}{\log n}}\right) \right) \geq \frac{ \log n - 2 \log \log n + o(\log n)}{\left(1 - (1 - \xi)\cdot \log n^{-\xi}\right)^2}\right\}.
\end{equation*}
Again, we aim to apply \cref{cor:AzumaGrowing}, this time with $A=\log n$ and $B = \frac{n}{\log n}$. By \cref{lem:expushSTRONG}
\begin{align*}
   \delta_t \geq \delta^{[\log n, \frac{\log n}{n}]}_t \geq  q \cdot \left(1 - 7\sqrt{\lambda + \frac{1}{\log n}}\right).\label{eq:spectralestdeltat} 
\end{align*}
Thus by \cref{cor:AzumaGrowing} for $\nu^{[A,B]}_t = q \cdot \left(1 - 7\sqrt{\lambda + \frac{1}{\log n}}\right)$,
\begin{equation*}
    \Pr{t_1}{ |I_{t_2}| \geq \frac{n}{\log n}  ~\Big|~ |I_{t_1}| \geq \log n } \geq 1 - o(1).
\end{equation*} 
Similarly as in Phase 1, solving for $t_2 - t_1$, we obtain,
\begin{align*}
    t_2 - t_1 &\leq  2 + \frac{ \log n - 2 \log \log n + o(\log n)}{\left(1 - (1 - \xi)\cdot \log n^{-\xi}\right)^2\cdot \log \left(1 + q \cdot \left(1 - 7\sqrt{\lambda + \frac{1}{\log n}}\right) \right)}\\  
    &\stackrel{(a)}{\leq} 2+ \frac{ \log n - 2 \log \log n + o(\log n)}{\left(1 - (1 - \xi)\cdot \log n^{-\xi}\right)^2\cdot\left(1 - 7\sqrt{\lambda + \frac{1}{\log n}}\right)\log(1 + q)} \\
    &\stackrel{(b)}{=} (1+o(1)) \cdot \frac{\log n}{\log(1 + q)},
\end{align*}
where again $(a)$ follows from \cref{clm:log_claim} and $(b)$ follows from $G_t$ being a strong expander (note that here, unlike as in Phase $1$ and $3$, we do need that $G_t$ is a \emph{strong} expander (\cref{eq:strongexpander})).
 \paragraph{Phase 3:}
Let 
\begin{equation*}
    t_3 := \min \left\{s \geq t_2 \colon \sum^{s-1}_{t=t_2} \log\left(1 +  q \cdot \left(1 - \frac{q}{2}\right)\cdot \frac{1}{2}\cdot (1 - \lambda)\right) \geq \frac{\log \log n - \log \log 2  + o(\log \log n)}{\left(1 - (1 - \xi)\cdot \left(\frac{n}{\log n}\right)^{-\xi}\right)^2}\right\}.
\end{equation*}
We want to apply \cref{cor:AzumaGrowing} with $A=\frac{n}{\log n}, B=n/2$. For $\frac{n}{\log n} \leq |I_t| \leq \frac{n}{2}$ it holds that,
\begin{equation*}
    \delta_t \geq \delta_t^{[\frac{n}{\log n}, n/2]} \stackrel{(a)}{\geq} q\cdot \left(1 - \frac{q}{2}\right)\cdot \varphi(I_t) \stackrel{(b)}{\geq} q\cdot \left(1 - \frac{q}{2}\right)\cdot \frac{1}{2}\left(1-\lambda\right),
\end{equation*}
where $(a)$ holds by  \cref{lem:PPP-P} \eqref{itm:Push} and $(b)$ by \cref{lem:setconductancebound}.
Thus, by \cref{cor:AzumaGrowing} with $\nu^{[A,B]}_t = q\cdot \left(1 - \frac{q}{2}\right)\cdot \frac{1}{2}\left(1-\lambda\right)$, we get that
\begin{equation*}
    \Pr{t_2}{ |I_{t_3}| \geq \frac{n}{2}  ~\Big|~ |I_{t_2}| \geq \frac{n}{\log n} } \geq 1 - o(1).
\end{equation*} 
Solving for $t_3-t_2$ results in,
\begin{align*}
    t_3 - t_2 &\leq 2 + \frac{\log \log n - \log \log 2  + o(\log \log n)}{\left(1 - (1 - \xi)\cdot \left(\frac{n}{\log n}\right)^{-\xi}\right)^2\cdot \log \left(1 + q \cdot \left(1 - \frac{q}{2}\right)\cdot \frac{1}{2} \cdot \left(1 - \lambda \right) \right)} \\ 
    &\stackrel{(a)}{\leq} 2+  \frac{\log \log n - \log \log 2  + o(\log \log n)}{\left(1 - (1 - \xi)\cdot \left(\frac{n}{\log n}\right)^{-\xi}\right)^2\cdot \frac{1}{2} \cdot \left(1 - \lambda \right)\cdot \left(1 - \frac{q}{2}\right)\cdot \log \left(1 + q\right)}  \\
    &\stackrel{(b)}{=} O\left(\frac{\log \log n}{\log (1 + q)}\right),
\end{align*}
where again $(a)$ follows form \cref{clm:log_claim} and $(b)$ follows from $G_t$ being a (strong) expander.

\paragraph{Phase 4: }
Let
\begin{equation*}
     t_4 := \min \left\{s \geq t_3 \colon \sum^{s-1}_{t=t_3}\log\left(1 - \left(1 - e^{-q}\right) \cdot \frac{1}{2}\cdot \left(1 - \lambda \right)\right) \leq -\frac{\log \log n - \log \log 2  +o(\log \log n)}{\left(1 -\frac{\log n}{2\cdot \left(1 - \Cshrink\right)\cdot n}\right)}\right\}.
\end{equation*}
For this Phase we aim to apply \cref{cor:shrinkingphase} with $C = n/2$ and $D=\frac{n}{\log n}$. Note that $t_4$ is a deterministic upper bound on the stopping time $\tau$ from \cref{cor:shrinkingphase}, and that for $D= \frac{n}{\log n}$,$\gamma = 1 - \min \left( \frac{1}{2\left(1 - \Cshrink\right)\cdot \frac{n}{\log n}}, \frac{1}{2}\right) = 1 - \frac{\log n}{2\left(1 - \Cshrink\right)\cdot n}$. By \cref{lem:expushUT} and \cref{lem:setconductancebound}
\begin{align*}
    \delta_t \geq \delta_t^{[\frac{n}{\log n},n/2]} \geq \Bigl(1 -  e^{-q(t)}\Bigr) \cdot \varphi(I_t) \geq \Bigl(1 -  e^{-q(t)}\Bigr) \cdot \frac{1}{2}\cdot \left(1 - \lambda\right).
\end{align*}
Thus, letting $\nu^{[C,D]}_t = \Bigl(1 -  e^{-q(t)}\Bigr) \cdot \frac{1}{2}\cdot \left(1 - \lambda\right)$, by \cref{cor:shrinkingphase}
\begin{equation*}
    \Pr{t_3}{\Condtwo{|U_{t_4}|\leq \frac{n}{\log n}}{|U_{t_3}|\leq \frac{n}{2}}}\geq 1 - o(1).
\end{equation*}
We note that by minimality of $t_4$,
\begin{align*}
    \sum_{t=t_3}^{t_4-1} \log\left(1 - \left(1 - e^{-q}\right) \cdot \frac{1}{2}\cdot \left(1 - \lambda \right)\right) &= \left(t_4 - t_3 -1\right)\cdot \log\left(1 - \left(1 - e^{-q}\right) \cdot \frac{1}{2}\cdot \left(1 - \lambda \right)\right)\\
    &\geq -\frac{\log \log n - \log \log 2  +o(\log \log n)}{\left(1 -\frac{\log n}{2\cdot \left(1 - \Cshrink\right)\cdot n}\right)} + \log \left(1 - \left(1 - e^{-q}\right) \cdot \frac{1}{2}\cdot \left(1 - \lambda \right)\right)
\end{align*}

Solving for $t_4 - t_3$, we get that,
\begin{align*}
    t_4 - t_3 &\leq 2 -\frac{\log \log n   - \log\log 2 + o(\log \log n)}{\left(1 -\frac{\log n}{2\cdot \left(1 - \Cshrink\right)\cdot n}\right)\cdot \log\left(1 - \left(1 - e^{-q}\right) \cdot \frac{1}{2}\cdot \left(1 - \lambda \right)\right)} \\
    &\stackrel{(a)}{\leq} 2+ \frac{\log \log n   - \log \log 2 + o(\log \log n)   }{\left(1 -\frac{\log n}{2\cdot \left(1 - \Cshrink\right)\cdot n}\right)\cdot  \left(1 - e^{-q}\right) \cdot \frac{1}{2}\cdot \left(1 - \lambda \right)} \\
    &\stackrel{(b)}{\leq} 2+ \frac{\log \log n  -\log \log 2 + o(\log\log n)}{\left(1 -\frac{\log n}{2\cdot \left(1 - \Cshrink\right)\cdot n}\right)\cdot q \cdot  \left(1 - \lambda \right)} \\
    &\stackrel{(c)}{=} O\left(\frac{\log\log n}{q}\right),
\end{align*}
where in $(a)$ we used the bound $-\log(1-x)\geq x$ for any $x\in (0,1)$, in $(b)$ we have used $ e^{-x} \leq  1 - x/2$ which holds for all  $x \in  [0, 1.59]$ and finally in $(c)$ the fact that $1 - \Cshrink = \epsilon > 0$ for $\epsilon$ constant, and that $G_t$ is a strong expander. 

 \paragraph{Phase 5:} Let 
\begin{equation*}
	t_5 := \left\{s \geq t_4 \colon \sum^{s-1}_{t=t_4}  \log \left(1 - \left(1 -  e^{-q}\right)\cdot  (1 - \lambda)\cdot \left(1 - \frac{1}{\log n} \right)\right) \leq -\frac{\log n - 2 \log\log n + o(\log n)}{1 -  \frac{1}{2\left(1 - \Cshrink\right) \log n}}\right\}.
\end{equation*}
By \cref{lem:expushUT} and \cref{lem:setconductancebound}
\begin{equation*}
    \delta_t \geq \delta_t^{[\log n,\frac{n}{\log}]} \geq \left(1 -  e^{-q}\right) \cdot \left(1 - \lambda\right) \cdot \left(1 - \frac{1}{\log n}\right)
\end{equation*}
By \cref{cor:shrinkingphase} with $C=\frac{n}{\log n}, D=\log n$ and $\nu_t^{[C,D]} = \left(1 -  e^{-q(t)}\right) \cdot \left(1 - \lambda\right) \cdot \left(1 - \frac{1}{\log n}\right)$, we get that
\begin{equation*}
		\Pr{t_4}{\Condtwo{|U_{t_5}|\leq \log n}{|U_{t_4}|\leq \frac{n}{\log n}}}\geq 1 - o(1).
\end{equation*}
Again, solving for $t_5-t_4$ we obtain that,

\begin{equation}\label{eq:phase5}
	t_5 - t_4 \leq 2  -\frac{\log n - 2 \log\log n + o(\log n) }{(1 -  \frac{1}{2 \left(1 - \Cshrink\right)\log n})\cdot \log \left(1 -  \left(1 -  e^{-q(t)}\right)\cdot (1 - \lambda)\cdot \left(1 - \frac{1}{\log n} \right)\right)}. \end{equation}Now, using \cref{clm:log_claim_new} in $(a)$, and in $(b)$ using the fact that $G_t$ is a strong expander, we have 
	\begin{align*}
		\log \left(1 - (1-e^{-q}) \cdot (1 - \lambda)\cdot \left(1 - \frac{1}{\log n} \right)\right)  
		&= \log \left(1 -  \left(1 -  \left(\lambda  + \frac{1}{\log n} -\frac{\lambda }{\log n} \right)\right) \cdot (1-e^{-q})  \right)\\
		&\stackrel{(a)}{\leq} \left(1- \frac{\lambda  + \frac{1}{\log n} -\frac{\lambda }{\log n}}{e^{-q}}\right) \log \left(1 -(1-e^{-q}) \right)\\
		&\stackrel{(b)}{=} -  (1-o(1))\cdot q. \end{align*} Inserting this into \cref{eq:phase5} gives us
	\begin{equation*}
		t_5 - t_4\leq 2 + \frac{\log n - 2 \log\log n + o(\log n)}{(1 -  \frac{1}{2 \left(1 - \Cshrink\right) \log n})\cdot  q  } \stackrel{(a)}{=} (1+o(1))\cdot \frac{\log n  }{  q   }, 
	\end{equation*}   
	where in $(a)$ we have used $1 - \Cshrink = \epsilon > 0$ for $\epsilon$ constant.

\paragraph{Phase 6:}
Let 
\begin{align*}
    t_6 := \min \biggl\{s \geq t_5 \colon & \sum^{s-1}_{t=t_5}\log\left(1 - \left(1 - e^{-q}\right) \cdot\frac{1}{2}\right) \leq -2 \left({\log\log n - \log\left(\frac{3}{4}\right)+ o(\log \log n)}\right)\biggr\}.
\end{align*}
By \cref{lem:expushUT} and \cref{lem:setconductancebound}, and as $(G_t)_{t\geq 0}$ is a strong expander, we have
\begin{equation*}
    \delta_t \geq \delta_t^{[3/4,\log n]} \geq \left(1 -  e^{-q(t)}\right) \cdot (1 - \lambda) \cdot \left(1 - \frac{\log n}{n}\right)\geq \left(1 -  e^{-q(t)}\right) \cdot \frac{1}{2}.
\end{equation*}
By \cref{cor:shrinkingphase} with $C=\log n, D=3/4$ and $\nu_t^{[C,D]} =  \left(1 -  e^{-q}\right) \cdot \frac{1}{2}$, 
\begin{equation*}
		\Pr{t_5}{\Condtwo{|U_{t_6}|\leq 3/4}{|U_{t_5}|\leq \log n}}\geq 1 - o(1).
\end{equation*}
Solving for $t_6 - t_5$, we get that
\begin{align*}
    t_6-t_5 &\leq 2 -2\cdot \frac{\log\log n - \log\left(\frac{3}{4}\right)+ o(\log \log n)}{\log\left(1 - \left(1 - e^{-q}\right) \cdot\frac{1}{2}\right)}\\
    & \stackrel{(a)}{\leq} 2 +4\cdot \frac{\log\log n+ o(\log\log n)}{ \left(1 - e^{-q}\right)  }\\
    & \stackrel{(b)}{\leq} 2 + 8 \cdot \frac{\log\log n  + o(\log\log n)}{ q }\\
    &= O\left(\frac{\log \log n}{q}\right),
\end{align*}
where, similarly to Phase 4, in $(a)$ we used $-\log(1-x)\geq x$ for any $x < 1$, and in $(b)$ we have used $ e^{-x} \leq  1 - x/2$ which holds for all  $x \in  [0, 1.59]$.

To conclude, we see that the dominant contributions to the running time come from Phases $2$ and $5$, thus, the total running time is bounded from above by $(1+o(1))\cdot\left(\frac{1}{\log (1+q)} + \frac{1}{q} \right)\cdot \log n$. 
\end{proof}

\section{Conclusions}

In this work, we presented a general framework for analyzing spreading processes with a time-dependent credibility function. The key idea is to link the spreading progress to an aggregate sum of growth (or shrinking) factors over consecutive rounds. In that way, our approach generalizes various previous works that were based on estimating the worst-case growth across all sets via the conductance of the graph. We also obtained several dichotomy results in terms of the number of vertices that get informed, both for general and more concrete credibility functions~(see \cref{sec:applications}).

In terms of open problems, a natural direction is to generalize our main technical results from regular graphs to arbitrary graphs, which we believe to be doable. Another avenue for future research is to allow more complex interactions between the credibility function $q(t)$ and the evolving set of informed vertices $I_t$, which could more accurately model an external influence on the network (e.g., fact-checkers). Lastly, one could consider more general spreading processes including other epidemic models (e.g., SIR model or independent cascade model), majority dynamics or variants of the voter model, in which informed vertices may also become uninformed in future steps.

\appendix

\section{Tools}\label{sec:AppendixD}

\subsection{Auxiliary Probabilistic Tools}
For convenience, we add several well known claims.

\begin{lemma}\label{lem:variance_basic}
Let $Y$ be any random variable, and $Y_1,Y_2$ be two independent copies with the same distribution as $Y$. Then,
\[
    \Var{Y} = \frac{1}{2} \cdot \Ex{ (Y_1-Y_2)^2 }.
     \]
\end{lemma}
\begin{proof}
   If $Y_1,Y_2 \sim Y$ are independent samples from the same distribution $Y$, then
\begin{align*}
  \Ex{ (Y_1-Y_2)^2 } &= \Ex{ (Y_1 - \Ex{Y} + \Ex{Y} - Y_2) ^2 } \\
  &= \Ex{ (Y_1 - \Ex{Y} )^2 } + 2 \cdot \Ex{ (Y_1 - \Ex{Y}) \cdot (\Ex{Y} - Y_2) } + \Ex{ (Y_2 - \Ex{Y} )^2 } \\
  &= 2 \cdot \Var{Y}.\qedhere
\end{align*} 
\end{proof}

\begin{lemma}\label{lem:variance}
Let $Z$ be a non-negative random variable. Then,
\[
 \Var{ \log(1+Z) } \leq  \Var{Z}.
\]
\end{lemma}
\begin{proof}
Let $Y$ be any random variable, and $Y_1,Y_2$ be two independent copies with the same distribution as $Y$. Then, by \cref{lem:variance_basic}, $2 \cdot \Var{Y} = \Ex{ (Y_1-Y_2)^2 } $. If $x, y \geq 0$ then $|\log(1+x) - \log(1+y)| \leq |x - y|$. To see the previous inequality, assuming  that $x>y$, the mean value theorem yields $\log(1+x)-\log(1+y) = \frac{1}{1+z} (x-y)$, where $z$ lies between $y$ and $x$, then clearly $z \geq 0$, so $0<\frac{1}{1+z}\leq 1$, proving the inequality. Now
\begin{align*}
 2 \cdot \Var{ \log(1+Z) } &= \Ex{ \left( \log(1+Z_1) - (\log(1+Z_2) \right)^2 }  \leq \Ex{ \left( Z_1 - Z_2 \right)^2 }  = 2 \cdot \Var{Z}.\qedhere
\end{align*}
\end{proof}

\begin{lemma}\label{lem:variance_two}
Let $Z$ be a non-negative random variable . Then, for any $a \in \mathbb{R}$,
\[
 \Var{ \max\{Z,a\} } \leq \Var{Z},
\qquad \text{as well as}
\qquad
 \Var{ \min(Z,a) } \leq \Var{Z}.
\]
\end{lemma}
\begin{proof}
This is identical to the proof of \cref{lem:variance} by noting that $|\max\{Z,a\}-\max\{Z',a\}|\leq |Z-Z'|$ and $|\min(Z,a)-\min(Z',a)|\leq |Z-Z'|$.
\end{proof}

\begin{lemma}[Jensen's Inequality]\label{lem:jensen}
    If $f$ is a convex function, then
    \begin{equation*}
        \Ex{f(X)} \geq f\left(\Ex{X}\right).
    \end{equation*}
    If $f$ is a concave function, then
    \begin{equation*}
        \Ex{f(X)} \leq f\left(\Ex{X}\right).
    \end{equation*}
\end{lemma}

\begin{theorem}[Optional Stopping Theorem {\cite[Theorem~10.10]{williams1991probability}}]\label{lem:OST}	Let $(X_t)_{t\geq 0}$ be a discrete-time supermartingale and $\tau \in \mathbb{N}_0\cup \{\infty\}$ be a stopping time, both with respect to the same filtration  $(\mathcal{F}_t)_{t\geq 0}$ . Assume that one of the following three conditions holds:
	
	\begin{itemize}
		\item There exists a constant $c$ such that almost surely $\tau \leq c$.
		\item $\E{\tau } < \infty$  and there exists a constant $c$ such that $\mathbb{E}\bigl[|X_{t+1}-X_t|\,\big\vert\,{\mathcal F}_t\bigr]\leq c$ almost surely on the event $\{\tau > t\}$ for all $t\in \mathbb{N}$.
		\item There exists a constant $c$ such that $|X_{t\wedge \tau}|\leq c$ almost surely for all $t \in  \mathbb {N}_0$. 
	\end{itemize}  
	Then $X_\tau$ is an almost surely well defined random variable and $\mathbb{E}[X_{\tau}]\leq \mathbb{E}[X_0]$. \end{theorem}

\subsection{Concentration Inequalities}
Here, we provide several concentration inequalities; stating with several versions of the Chernoff bound.

\begin{lemma}[Chernoff bound, cf.~\cite{dubhashi2009concentration}]
\label{lem:chernoff}
Suppose that $X_1,\ldots,X_n$ are independent Bernoulli random variables and let $X:=\sum_{i=1}^n X_i$ . Then, for any $0\leq \eta\leq 1$,
\begin{itemize}
\item [(i)]$\Pro{X\leq \left(1-\eta\right)\Ex{X}} \leq \exp\left({-\frac{\eta^2\Ex{X}}{2}}\right)$,
\item [(ii)]$\Pro{X \geq \left(1+\eta\right)\Ex{X}}\leq \exp\left({-\frac{\eta^2\Ex{X}}{3}}\right)$.
\end{itemize}
More generally, for any $\eta > 0$,
\begin{itemize}
    \item[(iii)] $\Pro{ X \geq (1+\eta) \Ex{X}} \leq \exp\left( - \frac{ \eta^2 \Ex{X}}{2+\eta} \right).$
    \item[(iv)] $\Pro{X \geq (1 + \eta) \Ex{X}} \leq \left(\frac{e^{\eta}}{(1 + \eta)^{1+\eta}}\right)^{\Ex{X}}  \leq \left(\frac{e}{1 + \eta}\right)^{(1+\eta) \cdot \Ex{X}} $
\end{itemize}
\end{lemma}

Not that these Chernoff bounds are valid as long as for all subsets $S \subseteq \{1,\ldots,n\}$, $\Pro{ \cap_{i \in S} X_i = 1} \leq \prod_{i \in S} \Pro{ X_i=1}$, as shown by Panconesi and Srinivasan~\cite{PS97}.

\begin{lemma}[{\cite[Theorem~6.5]{CL07}}]\label{lem:azuma_variance}
Let $(Z_i)_{i\geq 0}$ be a discrete-time martingale associated with a filter $\mathcal{F}$ satisfying
\begin{enumerate}
 \item $\Var{Z_i ~\Big|~ \mathcal{F}_{i-1} } \leq \sigma_i^2$ for all $1 \leq i \leq n$;
 \item $Z_{i-1}-Z_{i} \leq M$ for $1 \leq i \leq n$.
\end{enumerate}
Then for any $h \geq 0$, 
\begin{align*}
 \Pro{ Z_n - \Ex{Z_n} \leq - h } \leq \exp\left(- \frac{ h^2 }{ 2 
\cdot \left( \sum_{i=1}^n \sigma_i^2 + M h / 3 \right) } \right).
\end{align*}
\end{lemma}

\begin{lemma}[{\cite[Theorem~6.3]{CL07}}]\label{lem:azuma_variance_UT}
Let $(Z_i)_{i\geq 0}$ be a discrete-time martingale associated with a filter $\mathcal{F}$ satisfying
\begin{enumerate}
 \item $\Var{Z_i ~\Big|~ \mathcal{F}_{i-1} } \leq \sigma_i^2$ for all $1 \leq i \leq n$;
 \item $Z_{i}-Z_{i-1} \leq M$ for $1 \leq i \leq n$.
\end{enumerate}
Then for any $h \geq 0$, 
\begin{align*}
 \Pro{ Z_n - \Ex{Z_n} \geq h } \leq \exp\left(- \frac{ h^2 }{ 2 
\cdot \left( \sum_{i=1}^n \left(\sigma_i^2\right) + M h / 3 \right) } \right).
\end{align*}
\end{lemma}

\subsection{Variations of the expander mixing lemma}

\begin{lemma}[Expander mixing lemma - weak version \cite{alon2016probabilistic} Corollary 9.2.5]\label{lem:weakexpandermixinglemma}
\label{expander mixing lemma}
Let $G=(V,E)$ be a $d$-regular graph and let $S,T \subseteq V$ be sets of vertices. For $M:= D^{-1/2}AD^{-1/2}$ the normalized adjacency matrix of $G$ and $1 = \lambda_1 \geq \lambda_2 \geq \dots \geq \lambda_n \geq -1$ the eigenvalues of $M$ and $\lambda := \max\{|\lambda_2|,\lambda_3|,\dots,|\lambda_n|\}$. Then, 
\begin{equation*}
    \left|e(S,T) - \frac{|S||T|d}{n}\right| \leq \lambda \cdot d \sqrt{|S|\cdot |T|}.
\end{equation*}
\end{lemma}

\begin{lemma}[Expander mixing lemma - strong version \cite{chung1997spectral}]\label{lem:strongexpandermixinglemma}
    Let $G = (V,E)$ be a $d$-regular graph and let $S,T \subseteq V$ be sets of vertices. For $M:= D^{-1/2}AD^{-1/2}$ the normalized adjacency matrix of $G$ and $1 = \lambda_1 \geq \lambda_2 \geq \dots \geq \lambda_n \geq -1$ the eigenvalues of $M$ and $\lambda := \max\{|\lambda_2|,\lambda_3|,\dots,|\lambda_n|\}$. Then, it holds that,
 
    \begin{equation*}
        \left|e(S,T) - \frac{d}{n}\cdot |S|\cdot |T| \right| \leq \lambda \cdot \frac{d}{n}\sqrt{|S|\cdot |V\setminus S|\cdot |T|\cdot |V\setminus T|}.
    \end{equation*}
\end{lemma}

\begin{corollary}\label{cor:exmixingcor}
    Let $G=(V,E)$ be a $d$-regular graph. Then, for any subset $S \subseteq V$,
    \begin{equation*}
       (1 - \lambda) \cdot \frac{d}{n}\cdot |S| \cdot |V\setminus S|\leq e(S,V\setminus S) \leq (1 + \lambda) \frac{d}{n}\cdot |S|\cdot |V\setminus S|.
    \end{equation*}
\end{corollary}

\begin{proof}
    By choosing $T=V\setminus S$ in \cref{lem:strongexpandermixinglemma}, we obtain
    \begin{align*}  
    \left|e(S,V\setminus S) - \frac{d}{n}\cdot |S|\cdot |V\setminus S|\right| &\leq \lambda \cdot \frac{d}{n} \cdot |S|\cdot |V\setminus S|, 
    \end{align*}which gives the claimed bounds.
\end{proof}

\begin{lemma}\label{lem:setconductancebound}
    Let $M:= D^{-1/2}AD^{-1/2}$ be the normalized adjacency matrix of $G$, $1 = \lambda_1 \geq \lambda_2 \geq \dots \geq \lambda_n \geq -1$ the eigenvalues of $M$ and $\lambda := \max\{|\lambda_2|,\lambda_3|,\dots,|\lambda_n|\}$. Let $S \subseteq V$ such that $1 \leq |S|\leq n/2$. Then, the following holds,
    \begin{equation*}
        \varphi(S) \geq (1 - \lambda)\cdot \left(1 - \frac{|S|}{n}\right).
    \end{equation*}
\end{lemma}

\begin{proof}
    By \cref{lem:strongexpandermixinglemma}, we note that,
    \begin{align*}
        e(S,V\setminus S) &\geq  (1 - \lambda) \cdot \frac{d}{n}\cdot |S| \cdot |V\setminus S| \geq (1 - \lambda) \cdot d\cdot |S| \cdot \left(1 - \frac{|S|}{n}\right) = (1 - \lambda) \cdot \vol{S} \cdot \left(1 - \frac{|S|}{n}\right).
    \end{align*}
    Therefore, since $|S|\leq n/2$ and thus $\vol{S}\leq \vol{V\setminus S}$,
    \begin{equation*}
       \varphi(S) = \frac{e(S,V\setminus S)}{\min\left(\vol{S},\vol{V \setminus S}\right)} = \frac{e(S,V\setminus S)}{\vol{S}} \geq (1 - \lambda) \cdot \left(1 - \frac{|S|}{n}\right).
\qedhere    \end{equation*}
\end{proof}
 
\subsection{Simple Technical Claims}
	\begin{claim}\label{clm:tidyup} For all $x\geq 1$ and $\kappa\in (0,1/1000)$ we have, 
		\[
 \min \left( 64 \cdot x^{-1/4}, 1 - \kappa \right)  \leq (1-\kappa^2)\cdot x^{-\kappa^2}.
	\]
	\end{claim}
\begin{proof} Observe that for any $x< \left(\frac{64}{1-\kappa}\right)^4$ we have $1-\kappa < 64\cdot x^{-1/4}$, and 
	\[(1-\kappa^2)\cdot x^{-\frac{\log(1+\kappa)}{4 \log(\frac{64}{1-\kappa})}}>(1-\kappa^2)\cdot \left(\left(\frac{64}{1-\kappa}\right)^4\right)^{-\frac{\log(1+\kappa)}{4 \log(\frac{64}{1-\kappa})}}=(1-\kappa^2)\cdot \frac{1}{1+\kappa}= 1-\kappa.\] Conversely, for any $x\geq \left(\frac{64}{1-\kappa}\right)^4$ we have $   64\cdot x^{-1/4}\leq 1-\kappa$, and so we clam that also 
	\[64\cdot x^{-1/4} \leq (1-\kappa^2)\cdot x^{-\frac{\log(1+\kappa)}{4 \log(\frac{64}{1-\kappa})}}. \] 
	To see this, note that any such $x$ can be expressed as $c\cdot \left(\frac{64}{1-\kappa}\right)^4 $ for some $c\geq 1$. After this change of variables the inequality above reads 
	\[(1-\kappa)\cdot c^{-1/4} \leq (1-\kappa)\cdot c^{-\frac{\log(1+\kappa)}{4 \log(\frac{64}{1-\kappa})}},  \]which holds since $1/4 \geq \frac{\log(1+\kappa)}{4 \log(\frac{64}{1-\kappa})}$. The result follows by using the inequalities $\frac{x}{1+x}\leq \log (1+x)\leq x$ since \[ \frac{\log(1+\kappa)}{4 \log(\frac{64}{1-\kappa})} \geq \frac{\frac{\kappa}{1+\kappa}}{\frac{4\cdot 64}{1-\kappa}} \geq \kappa^2,  \]as $\kappa \in (0, 1/1000)$. 
\end{proof}
\begin{claim}\label{clm:log_claim}
For any $z \geq  0$ and any $a \in [0,1]$, the following holds:
\[
 \log \left(1+(1-a) \cdot z \right) \geq (1-a) \cdot \log \left(1+z \right).
\]
\end{claim}
\begin{proof}To begin observe that if $z=0$ then the statement holds, for any $a\in [0,1]$. So from now on we can assume $z>0$. 	For any fixed $z > 0$, we consider the function 
\[
f_z(a) := \log \left(1+(1-a) \cdot z \right)
- (1-a) \cdot \log \left(1+z \right).
\]
Note that for $a=0$, we have
\[
 f_z(0) = \log(1+z) - \log(1+z) = 0,
\]
and for $a=1$, we have
\[
 f_z(1) = \log(1) - 0 \cdot \log(1+z) = 0.
\]
Taking the derivatives,
\[
 \frac{\partial f_z(a)}{\partial a} = \frac{z}{(a-1) z - 1} + \log(z+1),
\]
and
\[
\frac{\partial^2 f_z(a)}{\partial^2a} = \frac{-z^2}{((a-1)z-1)^2}.
\]
Thus we see that $\frac{\partial^2 f_z(a)}{\partial a} <0$ always, since $z>0$. Therefore, all stationary points must be local maxima and so $f_z(a) \geq 0$ for all $a \in [0,1]$, which yields the claim.
\end{proof}
  
\begin{claim}\label{clm:log_claim_new}
For any $x < 1$ and $a <1$,
\[
 \log \left(1-(1-a)\cdot x  \right) \leq \left(1 - \frac{a}{1-x} \right)\cdot \log \left(1 - x\right).
\]
\end{claim}
\begin{proof}
The derivative of $f(x)=\log(x)$ at point $1-x$ is $\frac{1}{1-x}$. Therefore, as $\log(.)$ is concave,
\[
  \log \left(1-x+ a \cdot x \right) \leq \log \left(1 - x\right) + \frac{a \cdot x}{1-x}.
\]
Further, since $\log(1-x) \leq -x$, we have $x \leq -\log(1-x)$ and thus
\[
 \log \left(1-x+ a \cdot x \right) \leq \log(1-x) \cdot \left( 1 - \frac{a}{1-x} \right). \qedhere
\]
\end{proof}

\begin{claim}\label{clm:generalharmonic}
	For any $0 < \alpha < 1$ and $T\geq 1$,  
	\[
	\frac{T^{1-\alpha}-1}{1-\alpha} \leq \sum_{k=1}^Tk^{-\alpha} \leq \frac{T^{1-\alpha}}{1-\alpha}.
	\]
\end{claim}
\begin{proof}Since the function $x^{1-\alpha}$ is decreasing on $[1,\infty)$ and  $0<1-\alpha <1$ we have 
	\[\sum_{k=1}^Tk^{-\alpha} \leq 1 + \int_{1}^Tx^{-\alpha}\,\mathrm{d}x = 1 +\left[ \frac{x^{1-\alpha}}{1-\alpha}\right]_{x=1}^T = 1 +  \frac{T^{1-\alpha}}{1-\alpha} -  \frac{1}{1-\alpha} \leq  \frac{T^{1-\alpha}}{1-\alpha}.\]The lower bound follows similarly as $\sum_{k=1}^Tk^{-\alpha} \geq  \int_{1}^Tx^{-\alpha}\,\mathrm{d}x$. 
\end{proof}

\begin{lemma}[Stirling's approximation upper and lower bound]\label{lem:stirling}
For every integer $n \geq 1$,
\begin{equation*}
    \sqrt{2\pi n}\left(\frac{n}{e} \right)^n e^{\frac{1}{12n + 1}} < n! < \sqrt{2\pi n}\left(\frac{n}{e} \right)^n e^{\frac{1}{12n}}.
\end{equation*}
\end{lemma}

\begin{restatable}{claim}{HelperStirling}\label{clm:HelperStirling}
For any $\alpha > 0$,  
\begin{equation*}
    \prod_{i=0}^{1/\alpha-1} (2 - i \cdot \alpha) \leq \sqrt{2}\left(\frac{4}{e}\right)^{1/\alpha}.
\end{equation*}
\end{restatable}
\begin{proof}
We have
    \begin{equation*}
   \prod^{1/\alpha}_{i=1}(2 - (i-1)\cdot \alpha) = \frac{\prod_{i=1}^{1/\alpha}(2/\alpha - (i-1))}{(1/\alpha)^{1/\alpha}} = \frac{(2/\alpha )!}{(1/\alpha)! (1/\alpha)^{1/\alpha}} \stackrel{(a)}{\leq}  \frac{\sqrt{2\pi \cdot 2/\alpha}\left(\frac{2/\alpha}{e}\right)^{2/\alpha}e^{\frac{1}{12\cdot 2/\alpha}}}{\sqrt{2\pi \cdot 1/\alpha}\left(\frac{1/\alpha}{e}\right)^{1/\alpha}e^{\frac{1}{12\cdot 1/\alpha + 1}} (1/\alpha)^{1/\alpha}},\end{equation*}where $(a)$ holds by applying the upper bound in \cref{lem:stirling} to the factorial in the numerator and the lower bound in \cref{lem:stirling} to the factorial in the denominator. With some factoring and cancelling we now have 
      \begin{equation*}
      \prod^{1/\alpha}_{i=1}(2 - (i-1)\cdot \alpha)    \leq \frac{\sqrt{2/\alpha}\cdot  2^{2/\alpha} \left(1/\alpha\right)^{2/\alpha}\left(1/e\right)^{2/\alpha}e^{\frac{1}{24/\alpha}}}{\sqrt{1/\alpha} \cdot (1/\alpha)^{1/\alpha}(1/\alpha)^{1/\alpha}\left(1/e\right)^{1/\alpha}e^{\frac{1}{12/\alpha + 1} }} = \sqrt{2} \left(\frac{4}{e}\right)^{1/\alpha}\cdot \frac{e^{1/(24/\alpha)}}{e^{1/(12/\alpha) +1}} \leq \sqrt{2} \left(\frac{4}{e}\right)^{1/\alpha},
    \end{equation*}
as claimed.
\end{proof}

The next result is the lower bound version of the previous claim.

\begin{restatable}{claim}{HelperStirlingLB}\label{clm:HelperStirlingLB}
 For any $\alpha > 0$,
\begin{equation*}
     \prod^{1/\alpha-1}_{i=0} (2 - i\cdot \alpha) \geq  \frac{1}{\sqrt{2}} \left(\frac{4}{e}\right)^{1/\alpha}\cdot e^{-\frac{1}{2}\alpha}.
\end{equation*}
\end{restatable}
\begin{proof}Shifting the indices by one gives us
\begin{equation*}
    \prod^{1/\alpha-1}_{i=0}(2 - i\cdot \alpha)  = 
      \prod^{1/\alpha}_{i=1}(2 - i\cdot \alpha)   =
    \frac{\prod_{i=1}^{1/\alpha}(2/\alpha - i)}{(1/\alpha)^{1/\alpha}} = \frac{(2/\alpha - 1)!}{(1/\alpha-1)! (1/\alpha)^{1/\alpha}}  = \frac{(2/\alpha)!(1/\alpha)}{(1/\alpha)!(2/\alpha)(1/\alpha)^{1/\alpha}}.\end{equation*}
        By applying the Stirling's inequalities  \cref{lem:stirling} to the factorials in the numerator denominator we have 
    \begin{equation*}
       \prod^{1/\alpha-1}_{i=0}(2 - i\cdot \alpha)     \geq   \frac{\sqrt{2\pi \cdot 2/\alpha}\left(\frac{2/\alpha}{e}\right)^{2/\alpha}e^{\frac{1}{12\cdot 2/\alpha + 1}}(1/\alpha)}{\sqrt{2\pi \cdot 1/\alpha}\left(\frac{1/\alpha}{e}\right)^{1/\alpha}e^{\frac{1}{12\cdot 1/\alpha}} (1/\alpha)^{1/\alpha}(2/\alpha)} = \frac{\sqrt{2/\alpha} (1/\alpha) 2^{2/\alpha} \left(1/\alpha\right)^{2/\alpha}\left(1/e\right)^{2/\alpha}e^{\frac{1}{24/\alpha + 1}}}{\sqrt{1/\alpha}(2/\alpha)(1/\alpha)^{1/\alpha}(1/\alpha)^{1/\alpha}\left(1/e\right)^{1/\alpha}e^{\frac{1}{12/\alpha} }}.\end{equation*}We can then cancel many of the terms to give \begin{align*}
       \prod^{1/\alpha-1}_{i=0}(2 - i\cdot \alpha)   \geq  \frac{1}{\sqrt{2}} \left(\frac{4}{e}\right)^{1/\alpha}\cdot e^{\frac{1}{24/\alpha + 1} - \frac{1}{12/\alpha}} \geq \frac{1}{\sqrt{2}} \left(\frac{4}{e}\right)^{1/\alpha}\cdot e^{\frac{1}{25}\alpha - \frac{1}{12}\alpha} \geq \frac{1}{\sqrt{2}} \left(\frac{4}{e}\right)^{1/\alpha}\cdot e^{-\frac{1}{2}\alpha},
\end{align*} as claimed. 
\end{proof}

\begin{claim}\label{clm:nightmare}
For the constants $c_1:= \frac{1}{2}$ and $c_2:=\frac{1}{8}$ we have,
\[
 \prod_{i=0}^{\infty} \left(1 + \left(1-\frac{c_1}{\log n} \right)^{i} \right) \leq \sqrt{n}, \qquad 
and \qquad  \prod_{i=0}^{4 \log n-1} \left(1 + \left(1-\frac{c_2}{\log n} \right)^{i} \right) \geq n^{3/2}.
\]
\end{claim}
\begin{proof}
For the first statement,
\[
 \prod_{i=0}^{\infty} \left(1 + \left(1-\frac{c_1}{\log n} \right)^{i} \right) \leq \exp\left( \sum_{i=0}^{\infty} \left(1-\frac{c_1}{\log n} \right)^{i} \right) \leq \exp\left( \frac{\log n}{c_1} \right) = \sqrt{n},
\]
since $c_1=\frac{1}{2}$. For the second statement, recall $c_2:=\frac{1}{8}$, and note that for any $i \leq 4 \log n$,
\[
 1 + \left(1-\frac{c_2}{\log n} \right)^{i} 
 \geq 2 - 4c_2 \geq \frac{3}{2},
\]
where we have used Bernoulli's inequality. Therefore,
\[
 \prod_{i=0}^{4 \log n-1} \left(1 + \left(1-\frac{c_2}{\log n} \right)^{i} \right) \geq \left( \frac{3}{2} \right)^{4 \log n-1} \geq n^{3/2},
\]
as $(3/2)^4 > e^{3/2}$.
\end{proof}

\end{document}